\documentclass[11pt]{article}
\usepackage{fullpage}
\usepackage{graphicx,float}
\usepackage{amsmath,amssymb,amsthm,amsfonts,dsfont,mathtools,latexsym}
\usepackage{relsize}
\usepackage[margin=1cm]{caption}
\usepackage{wrapfig}
\usepackage{frame,color}
\usepackage{environ,subfig}
\usepackage{hyperref}
\usepackage{xspace}
\usepackage{framed}
\usepackage{comment}
\usepackage{indentfirst}
\usepackage{enumerate}
\theoremstyle{plain}
\newtheorem{theorem}{Theorem}[section]
\newtheorem{question}[theorem]{Question}
\newtheorem{lemma}[theorem]{Lemma}

\newtheorem{corollary}[theorem]{Corollary}

\theoremstyle{definition}
\newtheorem{definition}{Definition}[section]

\def\E{\mathrm{E}}
\def\Var{\mathrm{Var}}
\def\Pr{\mathrm{Pr}}

\def\eps{\varepsilon}
\def\vol{{\sf vol}}
\renewcommand{\epsilon}{\varepsilon}
\renewcommand{\tilde}{\widetilde}
\renewcommand{\hat}{\widehat}

\newcommand{\sumsign}{\mathcal{B}}
\newcommand{\sign}{\mathrm{sign}}
\newcommand{\bado}{\mathsf{Bad}}
\newcommand{\spa}{\mathrm{span}}

\newcommand{\ceil}[1]{\left\lceil #1 \right\rceil}
\newcommand{\rank}{{\operatorname{rank}}}
\newcommand{\argmin}{{\operatorname{argmin}}}

\newcommand{\wt}{\widetilde}

\def\Rbb{\mathbb{R}}

\def\R{\Rbb}

\DeclareMathOperator{\poly}{poly}
\DeclareMathOperator{\polylog}{polylog}

\title{The Communication Complexity of Optimization\thanks{Santosh S. Vempala was supported in part by NSF awards CCF-1717349 and DMS-1839323. Ruosong Wang and David P. Woodruff were supported in part by Office of Naval Research (ONR) grant N00014-18-1-2562. Part of this work was done while the authors were visiting the Simons Institute for the Theory of Computing.}}
\author{Santosh S. Vempala \\ Georgia Tech \\ \footnotesize \texttt{vempala@cc.gatech.edu} \and
Ruosong Wang\\ Carnegie Mellon University \\ \footnotesize \texttt{ruosongw@andrew.cmu.edu}\and
David P. Woodruff\\ Carnegie Mellon University \\ \footnotesize \texttt{dwoodruf@cs.cmu.edu}
}
\date{}
\begin{document}
\begin{titlepage}
\maketitle
\thispagestyle{empty}
\begin{abstract}
We consider the communication complexity of a number of distributed optimization problems. We start with the problem of solving a linear system. Suppose there is a coordinator together with $s$ servers $P_1, \ldots, P_s$, the $i$-th of which holds a subset $A^{(i)} x = b^{(i)}$ of $n_i$ constraints of a linear system in $d$ variables, and the coordinator would like to output an $x \in \mathbb{R}^d$ for which $A^{(i)} x = b^{(i)}$ for $i = 1, \ldots, s$. We assume each coefficient of each constraint is specified using $L$ bits. We first resolve the randomized and deterministic communication complexity in the point-to-point model of communication, showing it is $\tilde{\Theta}(d^2L + sd)$ and $\tilde{\Theta}(sd^2L)$, respectively. We obtain similar results for the blackboard communication model. As a result of independent interest, we show the probability a random matrix with integer entries in $\{-2^L, \ldots, 2^L\}$ is invertible is $1-2^{-\Theta(dL)}$, whereas previously only $1-2^{-\Theta(d)}$ was known.

When there is no solution to the linear system, a natural alternative is to find the solution minimizing the $\ell_p$ loss, which is the $\ell_p$ regression problem. While this problem has been studied, we give improved upper or lower bounds for every value of $p \geq 1$. One takeaway message is that sampling and sketching techniques, which are commonly used in earlier work on distributed optimization, are neither optimal in the dependence on $d$ nor on the dependence on the approximation $\epsilon$, thus motivating new techniques from optimization to solve these problems. 

Towards this end, we consider the communication complexity of optimization tasks which generalize linear systems, such as linear, semidefinite, and convex programming. For linear programming, we first resolve the communication complexity when $d$ is constant, showing it is $\tilde{\Theta}(sL)$ in the point-to-point model. For general $d$ and in the point-to-point model, we show an $\tilde{O}(sd^3 L)$ upper bound and an $\tilde{\Omega}(d^2 L + sd)$ lower bound. In fact, we show if one perturbs the coefficients randomly by numbers as small as $2^{-\Theta(L)}$, then the upper bound is $\tilde{O}(sd^2 L) + \textrm{poly}(dL)$, and so this bound holds for almost all linear programs. Our study motivates understanding the bit complexity of linear programming, which is related to the running time in the unit cost RAM model with words of $O(\log(nd))$ bits, and we give the fastest known algorithms for linear programming in this model. 
\end{abstract}
\end{titlepage}
\section{Introduction}\label{sec:intro}
Large-scale optimization problems often cannot fit into a single machine, and so they are distributed across a number $s$ of machines. That is, each of servers $P_1, \ldots, P_s$ may hold a subset of constraints that it is given locally as input, and the goal of the servers is to communicate with each other to find a solution satisfying all constraints. Since communication is often a bottleneck in distributed computation, the goal of the servers is to communicate as little as possible. 

There are several different standard communication models, including the point-to-point model and the blackboard model. In the point-to-point model, each pair of servers can talk directly with each other. This is often more conveniently modeled by looking at the {\it coordinator model}, for which there is an extra server called the coordinator, and all communication must pass through the coordinator. This is easily seen to be equivalent, from a total communication perspective, to the point-to-point model up to a factor of $2$, for forwarding messages from server $P_i$ to server $P_j$, and a term of $\log s$ per message to indicate which server the message should be forwarded to. Another model of computation is the {\it blackboard model}, in which there is a shared broadcast channel among all the $s$ servers. When a server sends a message, it is visible to each of the other $s-1$ servers and determines who speaks next, based upon an agreed upon protocol. We mostly consider randomized communication, in which for every input, we require the coordinator to output the solution to the optimization problem with high probability. For linear systems we also consider deterministic communication complexity. 

A number of recent works in the theory community have looked at studying specific optimization problems in such communication models, such as principal component analysis \cite{kvw14,LBKW14,BWZ16} and kernel \cite{BLSW016} and robust variants \cite{wz16,cdw18}, computing higher correlations \cite{kvw14}, $\ell_p$ regression \cite{wz13b,cdw18} and sparse regression \cite{BGMNW16}, estimating the mean of a Gaussian \cite{ZDW13,GMN14,BGMNW16}, database problems \cite{gwwz15,wz18}, clustering \cite{CSWZ16}, statistical \cite{WZ13}, graph problems \cite{phillips2012lower,WZ13} and many, many more. 

There are also a large number of distributed learning and optimization papers, for example \cite{bbfm12,ZDW13,ZWSL10,ACDL14,BPCPE11,m13,y13,cw09,r16,d12,c11,d12b,j14,s14,s14b,z15,as15,kane2017communication}. 
With a few exceptions, these 
works do not study general communication complexity, but 
rather consider specific classes of algorithms. Namely, a number of these works only allow
gradient and Hessian computations in each round, and do not allow arbitrary communication. Another aspect of these
works is that they typically do not count total bit complexity, but rather only count number
of rounds, whereas we are interested in total communication. 
In a number of optimization problems, the bit complexity of storing a single
number in an intermediate computation may be as large as storing the entire original
optimization problem. It is therefore infeasible to transmit such a number. 
While one could round this number,
the effect of rounding is often unclear, and could destroy the desired
approximation guarantee. One exception to the above is the work of \cite{t87}, which studies
the problem in which there are two servers, each holding a convex function, who would like
to find a solution so as to {\it minimize the sum} of the two functions. The upper bounds
are in a different communication model than ours, where the functions are added together,  
while the lower bounds only apply to a restricted class of protocols. 

Noticeably absent from previous work is the communication complexity of {\it solving linear systems}, which is a fundamental primitive in many optimization tasks. 
Formally, suppose there is a coordinator together with $s$ servers $P_1, \ldots, P_s$, the $i$-th of which holds a subset $A^{(i)} x = b^{(i)}$ of $n_i$ constraints of a $d$-dimensional linear system, and the coordinator would like to output an $x \in \mathbb{R}^d$ for which $A^{(i)} x = b^{(i)}$ for $i = 1, \ldots, s$. We further assume each coefficient of each constraint is specified using $L$ bits. The first question we ask is the following.

\begin{question}\label{question:linSystem}
What is the communication complexity of solving a linear system?
\end{question}

When there is no solution to the linear system, a natural alternative is to find the solution minimizing the $\ell_p$ loss, which is the $\ell_p$ regression problem  $\min_{x \in \mathbb{R}^d}\|Ax-b\|_p$, where for an $n$-dimensional vector $y$, $\|y\|_p = \left (\sum_{i=1}^n |y_i|^p \right )^{1/p}$ is its $\ell_p$ norm. 

In the distributed $\ell_p$ regression problem, each server
has a matrix $A^{(i)} \in \mathbb{R}^{n_i \times d}$ and a vector 
$b^{(i)} \in \mathbb{R}^{n_i}$, and the coordinator would like
to output an $x \in \mathbb{R}^d$ so that 
$\|Ax-b\|_p$ is 
approximately minimized, namely, that $\|Ax-b\|_p \leq (1+\epsilon) \min_{x'} \|Ax'-b\|_p$. 
Note that here $A \in \mathbb{R}^{n \times d}$ 
is the matrix obtained by stacking the matrices
$A^{(1)}, \ldots, A^{(s)}$ on top of each other, where $n = \sum_{i=1}^s n_i$. 
Also, $b \in \mathbb{R}^{n}$ is the vector obtained by stacking
the vectors $b^{(1)}, \ldots, b^{(s)}$ on top of each other. 
We assume that each entry of $A$ and $b$
is an $L$-bit integer, and we are interested in the randomized communication
complexity of this problem.

While previous work \cite{mm13,wz13b} has looked at the distributed 
$\ell_p$ regression problem, 
such work is based on two main ideas: sampling and sketching. 
Such techniques reduce a large optimization problem to a much smaller one, 
thereby allowing servers to send succinct synopses of their constraints in order 
to solve a global optimization problem. 

Sampling and sketching are the key techniques of recent work on distributed low rank 
approximation \cite{wz13b,kvw14} and regression algorithms. 
A natural question, which will motivate our study of more complex optimization problems
below, is whether other techniques in optimization can be used to obtain
more communication-efficient algorithms for these problems. 

\begin{question}\label{question:other}
Are there tractable optimization problems for which sampling and sketching 
techniques are suboptimal in terms of total communication? 
\end{question}

To answer Question \ref{question:other} it is useful to study optimization problems generalizing both
linear systems and $\ell_p$ regression for certain values of $p$. Towards this end, we 
consider the communication complexity of linear, semidefinite, and convex programming.
Formally, in the linear programming problem, 
suppose there is a coordinator together with $s$ servers $P_1, \ldots, P_s$, the $i$-th of 
which holds a subset $A^{(i)} x \leq b^{(i)}$ of $n_i$ constraints of a $d$-dimensional linear system, and the 
coordinator, who holds a vector $c \in \mathbb{R}^d$, 
would like to output an $x \in \mathbb{R}^d$ for which $c^Tx$ is maximized subject to
$A^{(i)} x \leq b^{(i)}$ for $i = 1, \ldots, s$. 
We further assume each coefficient of each constraint, as well as the objective function $c$, 
is specified using $L$ bits. 

\begin{question}\label{question:linSystem}
What is the communication complexity of solving a linear program?
\end{question}

One could try to implement known linear programming algorithms as distributed protocols. The main challenge here is that known linear programming algorithms operate in the {\it real RAM model of computation}, meaning that basic arithmetic operations on real numbers can be performed in constant time. This is problematic in the distributed setting, since it might mean real numbers need to be communicated among the servers, resulting in protocols that could have infinite communication. Thus, controlling the bit complexity of the underlying algorithm is essential, and this motivates the study of linear programming algorithms in the {\it unit cost RAM model of computation}, meaning that a word is $O(\log(nd))$ bits, and only basic arithmetic operations on words can be performed in constant time. Such a model is arguably more natural than the real RAM model. If one were to analyze the fastest linear programming algorithms in the unit cost RAM model, their time complexity would blow up by poly$(dL)$ factors, since the intermediate computations require manipulating numbers that grow exponentially large or small. Surprisingly, we are not aware of any work that has addressed this question:

\begin{question}\label{question:bitLP}
What is the best possible running time of an algorithm for linear programming in the unit cost RAM model?
\end{question}

As far as time complexity is concerned, it is not even known if linear programming is inherently more difficult than just solving a linear system. Indeed, a long line of work on interior point methods, with the current most recent work of \cite{cts18}, suggests that solving a linear program may not be substantially harder than solving a linear system. One could ask the same question for communication.

\begin{question}\label{question:equations}
Is solving a linear program inherently harder than solving a linear system? What about just checking
the feasibility of a linear program versus that of a linear system?
\end{question}

\paragraph{Recent Independent Work.}
A recent independent work \cite{assadi2019distributed} also studies solving linear programs in the distributed setting, although their focus is to study the tradeoff between round complexity and communication complexity in low dimensions, while our focus is to study the communication complexity in arbitrary dimensions. Note, however, that we also provide nearly optimal bounds for constant dimensions for linear programming in both coordinator and blackboard models.

\subsection{Our Contributions}
We make progress on answering the above questions, with nearly tight bounds in many cases. 
For a function $f$, we let 
$\tilde{O}(f) = f \polylog(sndL/\epsilon)$ and similarly
define $\tilde{\Theta}$ and $\tilde{\Omega}$. 

\subsubsection{Linear Systems}
We begin with linear systems, for which we obtain a complete answer for both randomized and
deterministic communication, in both coordinator and blackboard models of communication. 

\begin{theorem}\label{thm:linSystem}
In the coordinator model, the randomized communication complexity of solving a linear system
is $\tilde{\Theta}(d^2L + sd)$, while the deterministic communication complexity is
$\tilde{\Theta}(sd^2L)$. In the blackboard model, both the randomized communication complexity
and the deterministic communication complexity are $\tilde{\Theta}(d^2L + s)$. 
\end{theorem}

Theorem \ref{thm:linSystem} shows that randomization provably helps for solving linear systems.
The theorem also shows that in the blackboard model the problem becomes substantially easier. 

\subsubsection{Approximate Linear Systems, i.e., $\ell_p$ Regression}
We next study the $\ell_p$ regression problem in both the coordinator and blackboard models
of communication. Finding a solution to a linear system is a special case
of $\ell_p$ regression; indeed in the case that there is an $x$ for which $Ax = b$ we
must return such an $x$ to achieve $(1+\epsilon)$ relative error in objective function value for
$\ell_p$ regression. Consequently, our lower bounds for linear systems apply also
to $\ell_p$ regression for any $\epsilon > 0$. 

We first summarize our results in Table~\ref{tab:coordinator} and Table~\ref{tab:blackboard} for constant $\epsilon$. We state our results primarily for randomized communication. However, in the case of $\ell_2$ regression, we also discuss deterministic communication complexity. 
\begin{table}[htb!]

	\renewcommand{\arraystretch}{1.2}%

	\centering

	\begin{tabular*}{\textwidth}{c @{\extracolsep{\fill}}cccc}

		\hline\hline
		Error Measure & Upper Bound & Lower Bound & Theorem\\
		\hline
		$\ell_1$ (randomized) & $\tilde{O}(sd^2L)$ & $\tilde{\Omega}(d^2L + sd)$ &  Theorem \ref{thm:l1alg1}, \ref{thm:rand_eq_lb}\\
                $\ell_1$ (deterministic) & $\tilde{O}(sd^2L)$ & $\tilde{\Omega}(sd^2L)$ &Theorem \ref{thm:l1alg1}, \ref{thm:eq_lb}\\
		$\ell_2$ (randomized) & $\tilde{O}(sd^2L)$ & $\tilde{\Omega}(d^2L + sd)$ &  Theorem \ref{thm:l2alg1}, \ref{thm:rand_eq_lb}\\
                $\ell_2$ (deterministic) & $\tilde{O}(sd^2L)$ & $\tilde{\Omega}(sd^2L)$ &Theorem \ref{thm:l2alg1}, \ref{thm:eq_lb}\\
		$\ell_p$ for constant $p > 2$ & $\tilde{O}(s d^3 L)$ & $\tilde{\Omega}(d^2L + sd)$ &Theorem \ref{thm:l_p}, \ref{thm:rand_eq_lb}\\
		$\ell_{\infty}$ & $\tilde{O}(sd^3L)$ & $\tilde{\Omega}(d^2L + sd)$ &Theorem \ref{thm:l_infty}, \ref{thm:rand_eq_lb}\\
		\hline
	\end{tabular*}
	\caption{Summary of our results for $\ell_p$ regression in the coordinator model for constant
$\epsilon$.
}\label{tab:coordinator}
\end{table}

One of the main takeaway messages from Table \ref{tab:coordinator} is that sampling-based approaches,
namely those based upon the so-called Lewis weights \cite{cohen2015p}, would require $\Omega(d^{p/2})$ samples for $\ell_p$ regression when $p > 2$, and thus communication. Another way for solving $\ell_p$ regression
for $p > 2$ is via {\it sketching}, as done in \cite{wz13b}, but then the communication is $\Omega(n^{1-2/p})$.
Our method, which is deeply tied to linear programming, discussed more below, solves this problem
in $\tilde{O}(s d^3 L)$ communication. Thus, 
this gives a new method, departing from sampling and sketching techniques, which achieves much
better communication. Our method involves embedding $\ell_p$ into $\ell_{\infty}$, and then using
distributed algorithms for linear programming to solve $\ell_{\infty}$ regression. 

\begin{table}[htb!]

	\renewcommand{\arraystretch}{1.2}%

	\centering

	\begin{tabular*}{\textwidth}{c @{\extracolsep{\fill}}cccc}

		\hline\hline
		Error Measure & Upper Bound & Lower Bound & Theorem\\
		\hline
		$\ell_1$ & $\tilde{O}(s+d^2L)$ & $\tilde{\Omega}(s + d^2L)$ & Theorem \ref{thm:l1alg2}, \ref{thm:rand_eq_lb}\\
		$\ell_2$ & $\tilde{O}(s+d^2L)$ & $\tilde{\Omega}(s + d^2L)$ &  Theorem \ref{thm:l2alg2}, \ref{thm:rand_eq_lb}\\
		$\ell_p$ for constant $p > 2$ &$O(\min\{sd + d^4L, sd^3L\})$ & $\tilde{\Omega}(s+d^2L)$ & Theorem \ref{thm:l_p}, \ref{thm:rand_eq_lb}\\
		$\ell_{\infty}$ & $O(\min\{sd + d^4L, sd^3L\})$ & $\tilde{\Omega}(s+d^2L)$ & Theorem \ref{thm:l_infty}, \ref{thm:rand_eq_lb}\\
		\hline
	\end{tabular*}
	\caption{Summary of our results for $\ell_p$ regression in the blackboard model for constant 
$\epsilon$.}\label{tab:blackboard}
\end{table}

As with linear systems, one takeaway message from the results in Table \ref{tab:blackboard}
is that the problems have significantly more communication-efficient upper bounds in the 
blackboard model than in the coordinator model. Indeed, here we obtain tight bounds for
$\ell_1$ and $\ell_2$ regression, matching those that are known for the easier problem of linear
systems.

We next describe our results for non-constant $\epsilon$ in both the coordinator model and the 
blackboard model. Here
we focus on $\ell_1$ and $\ell_2$, which illustrate several surprises. 
\begin{table}[htb!]

	\renewcommand{\arraystretch}{1.2}%

	\centering

	\begin{tabular*}{\textwidth}{c @{\extracolsep{\fill}}cccc}

		\hline\hline
		Error Measure & Upper Bound & Lower Bound & Theorem\\
		\hline
	$\ell_1$ & $\tilde{O}(\min(sd^2L + \frac{d^2L}{\epsilon^2}, \frac{sd^3L}{\epsilon})$ & $\tilde{\Omega}(d^2L + sd)$ & Theorem \ref{thm:l1alg2}, \ref{thm:gd}, \ref{thm:rand_eq_lb} \\
	$\ell_2$ (randomized) & $\tilde{O}(sd^2L)$ & $\tilde{\Omega}(d^2L + sd)$ &  Theorem \ref{thm:l2alg1}, \ref{thm:rand_eq_lb}\\
        $\ell_2$ (deterministic) & $\tilde{O}(sd^2L)$ & $\tilde{\Omega}(sd^2L)$ & Theorem \ref{thm:l2alg1}, \ref{thm:eq_lb}\\
		\hline
	\end{tabular*}
	\caption{Summary of our results for $\ell_1$ and $\ell_2$ regression in the coordinator model for general 
$\epsilon$.}\label{tab:coordinatorEps}
\end{table}

One of the most interesting aspects of the results in Table \ref{tab:coordinatorEps} is our
dependence on $\epsilon$ for $\ell_1$ regression, where for small enough $\epsilon$ relative to $sd$,
we achieve a $1/\epsilon$ instead of a $1/\epsilon^2$ dependence. We note that all sampling 
\cite{cohen2015p}
and sketching-based solutions \cite{wz13b} to $\ell_1$ regression have a $1/\epsilon^2$ dependence. Indeed,
this dependence on $\epsilon$ comes from basic concentration inequalities. In contrast, our approach
is based on preconditioned first-order methods described in more detail below.

\begin{table}[htb!]

	\renewcommand{\arraystretch}{1.2}%

	\centering

	\begin{tabular*}{\textwidth}{c @{\extracolsep{\fill}}cccc}

		\hline\hline
		Error Measure & Upper Bound & Lower Bound & Theorem\\
		\hline
	$\ell_1$ & $\tilde{O}(s + \frac{d^2L}{\epsilon^2})$ & $\tilde{\Omega}(s + \frac{d}{\epsilon} + d^2L)$ \textrm{ for } $s > \Omega(1 / \eps)$ & Theorem \ref{thm:l1alg2}, \ref{thm:rand_eq_lb}, \ref{thm:l1_bb_lb}\\
	$\ell_2$ & $\tilde{O}(s + \frac{d^2L}{\epsilon})$ & $\tilde{\Omega}(s + \frac{d}{\epsilon^{1/2}} + d^2L)$ \textrm{ for } $s > \Omega(1 / \sqrt{\eps})$ & Theorem \ref{thm:l2alg2}, \ref{thm:rand_eq_lb}, \ref{thm:l2_bb_lb}\\
		\hline
	\end{tabular*}
	\caption{Summary of our results for $\ell_1$ and $\ell_2$ regression in the blackboard model for
general $\epsilon$.}\label{tab:blackboardEps}
\end{table}

A takeaway message from Table \ref{tab:blackboardEps} is that our lower bound shows some dependence
on $\epsilon$ is necessary both for $\ell_1$ and $\ell_2$ regression, provided $\epsilon$ is not too
small. This shows that in the blackboard model, one cannot obtain the same $\tilde{O}(d^2L + s)$ upper
bound for these problems as for linear systems, thereby separating their complexity from that of 
solving a linear system.  

\subsubsection{Linear Programming}
One of our main technical ingredients is to recast $\ell_p$ regression problems
as linear programming problems and develop the first communication-efficient solutions for distributed
linear programming. Despite this problem being one of the most important problems that we know how to solve
in polynomial time, we are not aware of any previous work considering its communication complexity in generality besides a recent independent work \cite{assadi2019distributed}

First, when the dimension $d$ is constant, 
we obtain nearly optimal upper and lower bounds. 

\begin{theorem}\label{thm:LP_fixed_d}
In constant dimensions, the randomized communication complexity of linear programming is $\widetilde{\Theta}(sL)$ 
in the coordinator model and $\widetilde{\Omega}(s+L)$ in the blackboard model. Our upper bounds allow the
coordinator to output the solution vector $x \in \mathbb{R}^d$, while the lower
bounds hold already for testing if the linear program is feasible. 
Here the $\widetilde{\Theta}(\cdot)$ notation and the $\widetilde{\Omega}(\cdot)$ notation suppress only $\polylog(sL)$ factors. 
\end{theorem}

Despite the fact that we do not have tight upper bounds matching the $\widetilde{\Omega}(s+L)$ lower bounds in the blackboard model, under the additional assumption that each constraint in the linear program is placed on a random server, we develop an algorithm with a matching $\widetilde{O}(s+L)$ communication cost. Partitioning constraints randomly across servers instead is common in distributed computation, see, e.g., \cite{ak17}. Neverthelss we leave it as an open problem in the blackboard model in constant dimensions, to remove this requirement.

For solving a linear system in constant dimensions, the randomized communication complexity is $\tilde{\Theta}(s + L)$ in both models.
Again, the $\widetilde{\Theta}(\cdot)$ notation suppresses only $\polylog(sL)$ factors. 
Thus, in the coordinator model, we separate the communication complexity of these problems. 
We can also separate the complexities in the blackboard model if we instead look at the feasibility problem. Here instead
of requiring the coordinator to output the solution vector, we just want to see if the linear system or linear program
is feasible. We have the following theorem for this.

\begin{theorem}\label{thm:Axb}
In constant dimensions, the randomized communication complexity of checking whether a system of linear equations is feasible is $O(s\log L)$ in either the coordinator or blackboard model of communication.
\end{theorem}
Combining Theorem \ref{thm:LP_fixed_d} and Theorem \ref{thm:Axb}, we see that for feasibility in the blackboard model, linear programming requires $\tilde{\Omega}(s + L)$ bits, while linear system feasibility takes $\tilde{O}(s)$ bits, and thus we separate
these problems in the blackboard model as well. 

Returning to linear programs, we next consider the complexity in arbitrary dimensions.

\begin{theorem}\label{thm:LP}
In the coordinator model, 
the randomized communication complexity of exactly solving a linear program $\max \{c^Tx \, : \, Ax \le b\}$ with $n$ constraints in dimension $d$ and all coefficients specified by $L$-bit numbers is $\tilde{O}(sd^3L)$. Moreover it is lower bounded by $\tilde{\Omega}(d^2L + sd)$. Here the upper and lower bounds require the coordinator to output 
the solution vector $x \in \mathbb{R}^d$. 
\end{theorem}

The lower bound in Theorem \ref{thm:LP} just follows from our lower bound for linear systems. The upper bound is based on an optimized distributed cutting-plane algorithm. We describe the idea below.

While the upper bound is
$\tilde{O}(sd^3L)$, one can further improve it as follows. We show that if
the coefficients of $A$ in the input to the linear program
are perturbed independently by i.i.d. discrete Gaussians with variance as
small as $2^{-\Theta(L)}$, then we can improve the upper bound for solving
this perturbed problem to
$\tilde{O}(sd^2L + d^4L)$, where now the success probability of the algorithm
is taken over both the randomness of the algorithm and the random input
instance, which is formed by a random perturbation of a worst-case instance. Note that
this is an improvement for sufficiently large $s$. 
Our model coincides with the well-studied {\it smooth complexity}
model of linear programming \cite{spielman2001smoothed,b02,s03}. However, a major
difference is that the variance of the perturbation
needs to be at least inverse polynomial in their works, 
whereas we allow our variance to be as small as $2^{-\Theta(L)}$. 

\begin{theorem}
In the smoothed complexity model with discrete Gaussians of variance $2^{-\Theta(L)}$, the communication complexity of exactly solving a linear program $\max \{c^Tx \, : \, Ax \le b\}$ with $n$ constraints in dimension $d$ and all coefficients specified by $L$-bit numbers, with probability at least $9/10$ over the input distribution and randomness of the protocol, is $\tilde{O}(sd^2L + d^4 L)$ in the coordinator model. 
\end{theorem}

While our focus in this paper is on communication, our upper bounds 
also give a new technique for improving the
time complexity in the unit cost RAM model of linear programming,
where arithmetic operations on words of size $O(\log(nd))$ can be performed
in constant time. For this fundamental problem
we obtain the fastest known algorithm even in the {\it non-smoothed
setting} of linear programming.

\begin{theorem}\label{thm:LPbit}
The time complexity of solving an $n \times d$ linear program with $L$-bit coefficients is $\tilde{O}(nd^\omega L + \poly(dL))$ in the unit cost RAM model. 
\end{theorem} 

We note that this is for solving an LP {\it exactly} in the RAM model with words of size $O(\log (nd))$ bits. The current fastest linear programming algorithms \cite{ls14,ls15,cts18} state the bounds in terms of additive error $\epsilon$, which incurs a multiplicative factor of at least $\Omega(dL)$ to solve the problem exactly. Also such algorithms manipulate large numbers at intermediate points in the algorithm, which are at least $L$ bits, which could take $\Omega(L)$ time to perform a single operation on. It seems that transferring such results
to the unit cost RAM model with $O(\log(nd))$ bit words incurs
time at least $\Omega(nd^{2.5}L^2 + d^{w+1.5}L^2)$. This holds true even
of the recent work \cite{cts18}, which focuses on the setting $n=O(d)$ and does not improve the leading $nd^{2.5} L^2$ term. Even such a bit-complexity bound needs careful checking of the number of bits required as recent improvements use sophisticated inverse maintenance methods to save on the number of operations (an exercise that was carried out thoroughly for the Ellipsoid method in \cite{GLS}).


\subsubsection{Implications for Convex Optimization and Semidefinite Programming}
Our upper bounds also extend to more general convex optimization problems. For these, we must modify the problem statement to finding an $\eps$-additive approximation rather than the exact solution. We obtain the following upper bound for a convex program in $\R^d$. 
\begin{theorem}\label{thm:convex}
The communication complexity of solving the convex optimization problem 
$
\min  \{c^T x\, :\, x \in \bigcap_i K_i \}
$
for convex sets $K_i \subseteq RB^n$, one per server, to within an additive error $\epsilon$, i.e., finding a point $y$ s.t. $c^Ty \le OPT+\eps$ and $y \in \bigcap_i K_i + \eps B^n$ is $O(sd^2\log(Rd/\eps)\log d)$. 
\end{theorem}
If the objective function is not known to all servers, we incur an additional $O(sdL)$ communication. For semidefinite programs with $d \times d$ symmetric matrices and $n$ linear constraints this gives a bound of $\tilde{O}(sd^4\log(1/\epsilon))$. Note that we can simply send all the constraints to one server in $O(nd^2L)$ communication, so this is always an upper bound.



\subsection{Our Techniques}
\subsubsection{Linear Systems}
To solve linear systems in the distributed setting, the coordinator can go through the servers one by one.
The coordinator and all servers maintain the same set $C$ of linearly independent linear equations. 
For each server $P_i$, if there is a linear equation stored by $P_i$ that is linearly independent with linear equations in $C$, then $P_i$ sends that linear equation to all other servers and adds that linear equation into $C$.
In the end, $C$ will be a maximal set of linearly independent equations, and thus the coordinator can simply solve the linear equations in $C$.
This protocol is deterministic and has communication complexity $O(sd^2L)$ in the coordinator model and $O(s + d^2L)$ in the blackboard model, since at most $d$ linear equations will be added into the set $C$.

In fact, the preceding protocol is optimal for deterministic protocols, even just for testing the feasibility of linear systems.
To prove lower bounds, we first prove the following new theorem about random matrices which may be of independent interest.
\begin{theorem}[Informal version of Theorem \ref{thm:singularity_main}]\label{theorem:random_inverse}
Let $R$ be a $d \times d$ matrix with i.i.d. random integer entries in $\{-2^L, \ldots, 2^L\}$. 
The probability that $R$ is invertible is $1-2^{-\Theta(dL)}$.
\end{theorem}
The previous best known probability bound was only $1-2^{-\Theta(d)}$
\cite{tao2006random,bourgain2010singularity}; we stress that the results of
\cite{bourgain2010singularity} are not sufficient
\footnote{We have verified this with Philip Matchett Wood, who is an
  author of \cite{bourgain2010singularity}. The issue is that in their Corollary
  1.2, they have an explicit constraint on the cardinality of the set $S$, i.e., $|S| = O(1)$.
  In their Theorem 2.2, it is assumed that
  $|S| = n^{o(n)}$. Thus, as far as we are aware, there are
  no known results sufficient to prove
  our singularity probability bound.}
to prove our stronger
bound with the extra factor of $L$ in the exponent, which is crucial for our
lower bound. 


With Theorem \ref{theorem:random_inverse}, in Lemma \ref{lem:hard_instance}, we use the probabilistic method to construct a set of $|\mathcal{H}| = 2^{\Omega(d^2L)}$ matrices $\mathcal{H} \subseteq \mathbb{R}^{d \times d}$ with integral entries in $[-2^L, 2^L]$, such that for any $S, T \in \mathcal{H}$, $S^{-1} e_d \neq T^{-1} e_d$, where $e_d$ is the $d$-th standard basis vector. 

Now consider any deterministic protocol for testing the feasibility of linear systems. 
Suppose the linear system on the $i$-th server is $H_i x = e_d$ for some $H_i \in \mathcal{H}$, then the entire linear system is feasible if and only if $H_1 = H_2 = \ldots = H_s$.
This is equivalent to the problem in which each server receives a binary string of length $\log(|\mathcal{H}|)$, and the goal is to test whether all strings are the same or not.
In the coordinator model, a deterministic lower bound of $\Omega(s\log(|\mathcal{H}|))$ for this problem can be proved using the symmetrization technique in \cite{phillips2012lower, woodruff2014optimal}, which gives an optimal $\Omega(sd^2L)$ lower bound. 
An optimal $\Omega(s + d^2L)$ deterministic lower bound can also be proved in the blackboard model. 
The formal analysis is given in Section \ref{sec:det_lb_ls}.

For solving linear systems, an $\Omega(d^2L)$ lower bound holds even for randomized algorithms in the coordinator model.
When there is only a single server which holds a linear system $Hx = e_d$ for some $H \in \mathcal{H}$, 
in order for the coordinator to know the solution $x = H^{-1}e_d$, standard information-theoretic argument shows that $\log(|\mathcal{H}|)$ bits of communication is necessary, which gives an $\Omega(d^2L)$ lower bound.  
This idea is formalized in Section \ref{sec:ls_rand_lb}.
A natural question is whether the $O(sd^2L)$ upper bound is optimal for randomized protocols. 

We first show that in order to test feasibility, it is {\em possible} to achieve a communication complexity of $O(sd^2\log(dL))$, which can be exponentially better than the bound for deterministic protocols. 
The idea is to use hashing. 
With randomness, the servers can first agree on a random prime number $p$, and test the feasibility over the finite field $\mathbb{F}_p$.
It suffices to have the prime number $p$ randomly generated from the range $[2, \poly(dL)]$, and thus the $L$ factor in the communicataion complexity of deterministic protocols can be improved to $\log p = \log(dL)$.
However, it is still unclear if {\it solving linear systems} in the coordinator model will require $\Omega(sd^2L)$ bits of communication for randomized protocols. 

Quite surprisingly, we show that $O(sd^2L)$ is not the optimal bound for randomized protocols, and the optimal bound is $\widetilde{\Theta}(d^2L + sd)$.
In the deterministic protocol with communication complexity $O(sd^2L)$, most communication is wasted on synchronizing the set $C$, which requires the servers to send linear equations to all other servers. 
In our new protocol, {\em only the coordinator} maintains the set $C$.
The issue now, however, is that the servers no longer know which linear equation they own is linearly independent with those equations in $C$.
On the other hand, each server can simply generate a random linear combination of all linear equations it owns. 
We can show that if a server does have a linear equation that is linearly independent with those in $C$, with constant probability, the random linear combination is also linearly independent with those in $C$, and thus the coordinator can add the random linear combination into $C$.
Notice that taking random linear combinations to preserve the rank of a matrix is a special case of dimensionality reduction or {\it sketching}, which comes up in a number of applications, see, for example compressed sensing \cite{candes2005decoding,d06}, data streams \cite{andoni2013eigenvalues}, and randomized numerical linear algebra \cite{woodruff2014sketching}.
Here though, a crucial difference is that we just need the fact that if a set of vectors $S$ is not contained in the span of another set of vectors $V$, then a random linear combination of the vectors in $S$ is also not in the span of $V$ with high probability. This allows us to adaptively take as few linear combinations as possible to solve the linear system, enabling us to achieve much lower communication than would be possible by just sketching the linear systems at each server and non-adaptively combining them.

If we implement this protocol na\"ively, then the communication complexity will be $\wt{O}(d^2L + sdL)$, since at most $d$ linear equations will be added into $C$, and there is an $\widetilde{O}(dL)$ communication complexity associated with each of them. Furthermore, even if a server does not have any linear equation that is linearly independent with $C$, it still needs to send random linear combinations to the coordinator, which would require $\widetilde{O}(sdL)$ communication.
To improve this further to $\widetilde{O}(sd)$, we can still use the hashing trick mentioned before.
If a server generates a random linear combination, it can first test whether the linear combination is linearly independent with $C$ over the finite field $p$, for a random prime $p$ chosen in $[2, \poly(dL)]$. 
This will reduce the communication complexity to $\widetilde{O}(d)$ for each test.
If the linear equation is indeed linearly independent with $C$, then the server sends the original linear equation (without taking the residual modulo $p$) to the coordinator. 
Again the total communication complexity for sending the original linear equations is upper bounded by $O(d^2L)$.
Thus, the total communication complexity is upper bounded by $\widetilde{O}(d^2L + sd)$.
See Section \ref{sec:rand_ls_solve} for the formal analysis. 

By a reduction from the OR of $s - 1$ copies of the two-server set-disjointness problem to solving linear systems, we can prove an extra $\widetilde{\Omega}(sd)$ lower bound, which holds even for testing feasibility of linear systems. 
Here the idea is to interpret vectors in $\{0, 1\}^d$ as characteristic vectors of subsets of $[d]$. 
One of the servers will fix the solution of the linear system to be a predefined vector $x$.
Each server $P_i$ has a single linear equation $a_i^T x = 1$.
By interpreting vectors as sets, $a_i^T x = 1$ implies the set represented by $a_i$ and $x$ are intersecting. 
Thus, the servers are actually solving the OR of $s - 1$ copies of the two-server set-disjointness problem, which is known to have $\widetilde{\Omega}(sd)$ communication complexity \cite{phillips2012lower, WZ13}.
This lower bound is formally given in Section \ref{sec:ls_rand_lb}.

\subsubsection{Linear Regression}
For an $\ell_2$ regression instance $\min_x \|Ax - b\|_2$, the optimal solution can be calculated using the {\em normal equations}, i.e., the optimal solution $x$ satisfies $A^TA x = A^T b$.
This already gives a simple yet nearly optimal deterministic protocol for $\ell_2$ regression in the coordinator model: the coordinator calculates $A^T A$ and $A^T b$ using only $\widetilde{O}(sd^2L)$ bits of communication by collecting the covariance matrices from each server and summing them up.
 The $\widetilde{O}(sd^2L)$ communication complexity matches our lower bound for solving linear systems for deterministic protocols in the coordinator model. 
However, when implemented in the blackboard model, the communication complexity of this protocol is still $\widetilde{O}(sd^2L)$. 
To improve this bound, we first show how to efficiently obtain approximations to leverage scores in both models. 
Our protocol is built upon the algorithm in \cite{cohen2015uniform}, but implemented in a distributed manner. 
The resulting algorithm has $\widetilde{O}(sd^2L)$ communication complexity in the coordinator model but only $\widetilde{O}(s + d^2L)$ communication complexity in the blackboard model. 
With approximate leverage scores, the coordinator can then sample $\widetilde{O}(d / \varepsilon^2)$ rows of the matrix $A$ to obtain a {\em subspace embeeding}, at which point it will be easy to calculate a $(1 + \varepsilon)$-approximate solution to the $\ell_2$ regression problem. 
The number of sampled rows can be further improved to $\widetilde{O}(d / \varepsilon)$ using S\'arlos's argument \cite{sarlos2006improved} since solving $\ell_2$ regression does not necessarily require a full $(1 + \varepsilon)$ subspace embedding, which results in a protocol with communication complexity $\widetilde{O}(s + d^2L / \varepsilon)$ in the blackboard model. 
Full details can be found in Section \ref{sec:l2regression}.

One may wonder if the dependence on $1 / \varepsilon$  is necessary for solving $\ell_2$ regression in the blackboard model. 
In Section \ref{sec:bb_reg_lb}, we show that some dependence on $1 / \varepsilon$ is actually necessary.
We show an $\Omega(d / \sqrt{\varepsilon})$ lower bound whenever $s > \Omega(1 / \sqrt{\varepsilon})$.
The hardness follows from the fact that if the matrix $A$ satisfies $A^{(i)} = I$ for all $i \in [s]$, then the optimal solution is just the {\em average} of $b^{(1)}, b^{(2)}, \ldots, b^{(s)}$.
Thus, if we can get sufficiently good approximation to the $\ell_2$ regression problem, then we can actually recover the {\em sum} of $b^{(1)}, b^{(2)}, \ldots, b^{(s)}$, at which point we can resort to known communication complexity lower bound in the blackboard model \cite{phillips2012lower}.
This argument will also give an $\Omega(d / \varepsilon)$ lower bound for $(1 + \varepsilon)$-approximate $\ell_1$ regression in the blackboard model, whenever $s > \Omega(1 / \varepsilon)$.
The formal analysis can be found in Section \ref{sec:bb_reg_lb}.

For $\ell_1$ regression, we can no longer use the normal equations. 
However, we can obtain approximations to $\ell_1$ Lewis weights by using approximations to leverage scores, as shown in \cite{cohen2015p}. 
With approximate $\ell_1$ Lewis weights of the $A$ matrix, the coordinator can then obtain a $(1 + \varepsilon)$ $\ell_1$ subspace embedding by sampling $\widetilde{O}(d / \varepsilon^2)$ rows.
This will give an $O(sd^2L + d^2L / \varepsilon^2)$ upper bound for $(1 + \varepsilon)$-approximate $\ell_1$ regression in the coordinator model, and an $O(s + d^2L / \varepsilon^2)$ upper bound in the blackboard model. 
It is unclear if the number of sampled rows can be further reduced since there is no known $\ell_1$ version of S\'arlos's argument.
A natural question is whether the $1 / \varepsilon^2$ dependence is optimal. 
We show that the dependence on $\varepsilon$ can be further improved to $1 / \varepsilon$, by using optimization techniques, or more specifically, first-order methods. 
Despite the fact that the objective function of $\ell_1$ regression is neither smooth nor strongly-convex, it is known that by using Nesterov's Accelerated Gradient Descent and smoothing reductions \cite{nesterov2005smooth}, one can solve $\ell_1$ regression using only $O(1 / \varepsilon)$ full gradient calculations. 
On the other hand, the complexity of first-order methods usually has dependences on various parameters of the input matrix $A$, which can be unbounded in the worst case.
Fortunately, recent developments in $\ell_1$ regression \cite{durfee2017ell_1} show how to {\em precondition} the matrix $A$ by simply doing an $\ell_1$ Lewis weights sampling, and then rotating the matrix appropriately. 
By carefully combining this preconditioning procedure with Accelerated Gradient Descent, we obtain an algorithm for $(1 + \varepsilon)$-approximate $\ell_1$ regression with communication complexity $\widetilde{O}(sd^3L / \varepsilon)$ in the coordinator model, which shows it is indeed possible to improve the $\varepsilon$ dependence for $\ell_1$ regression. 
A formal analysis is given in Section \ref{sec:l1regression}.

For general $\ell_p$ regression, if we still use Lewis weights sampling, then the number of sampled rows and thus the communication complexity will be $\Omega(d^{p / 2})$.
Even worse, when $p = \infty$, Lewis weights sampling will require an unbounded number of samples. 
However, $\ell_{\infty}$ regression can be easily formulated as a linear program, which we show how to solve exactly in the distributed setting.
Inspired by this approach, we further develop a general reduction from $\ell_p$ regression to linear programming.
Our idea is to use the max-stability of exponential random variables \cite{andoni2017high} to embed $\ell_p$ into $\ell_{\infty}$, write the optimization problem in $\ell_{\infty}$ as a linear program and then solve the problem using linear program solvers. 
However, such embeddings based on exponential random variables usually produce heavy-tailed random variables and makes the dilation bound hard to analyze. 
Here, since our goal is just to solve a linear regression problem, we only need the dilation bound for the optimal solution of the regression problem. 
The formal analysis in Section \ref{sec:reg_lp} shows that $(1 + \varepsilon)$-approximate $\ell_p$ regression can be reduced to solving a linear program with $\widetilde{O}(d / \varepsilon^2)$ variables, which implies a communication protocol for $\ell_p$ regression without the $\Omega(d^{p / 2})$  dependence. 

\subsubsection{Linear and Convex Programs}
We adapt two different algorithms from the literature for efficient communication and implement them in the distributed setting.
The first is Clarkson's algorithm, which works by sampling $O(d^2)$ constraints in each iteration and finds an optimal solution to this subset; the sampling weights are maintained implicitly. In each iteration the total communication is $O(d^3L)$ for gathering the constraints and an additional $\tilde{O}(sd^2L)$ per round to send the solution to this subset of constraints to all servers. This solution is used to update the sampling weights. Clarkson's algorithm has the nice guarantee that it needs only $O(d\log n)$ rounds with high probability. A careful examination of this algorithm shows that the bit complexity of the computation (not the communication) is dominated by checking whether a proposed solution satisfies all constraints, i.e., computing $Ax$ for a given $x$. We show this can be done with time complexity $\tilde{O}(nd^\omega L)$ in the unit cost RAM model and this is the leading term of the claimed time bound. 

Notice that the $\widetilde{O}(sd^3L)$ term in the communication complexity of Clarkson's algorithm comes from the fact that the protocol needs to send an optimal solution $x^*$ of a linear program with size $O(d^2) \times d$ for a total of $O(d \log n)$ times. 
However, when each server $P_i$ receives $x^*$, all $P_i$ will do is to check whether $x^*$ satisfies the constraints stored on $P_i$ or not. 
Notice that here entries in the constraints have bit complexity $L$, whereas the solution vector $x^*$ has bit complexity $\widetilde{O}(dL)$ for each entry. 
Intuitively, for most linear programs, we don't need such a high precision for the solution vector $x^*$.
This leads to the idea of smoothed analysis. 
We show that if
the coefficients of $A$ in the input to the linear program
are perturbed independently by i.i.d. discrete Gaussians with variance as
small as $2^{-\Theta(L)}$, then we can improve the upper bound for solving
this perturbed problem to
$\tilde{O}(sd^2L + d^4L)$.
The reason here is that with Gaussian noise, we can round each entry of the solution vector $x^*$ to have bit complexity $\widetilde{O}(L)$, which would suffice for verifying whether $x^*$ satisfies the constraints or not, for most linear programs. 
Full details regarding Clarkson's algorithm and the smoothed analysis model can be found in Section \ref{sec:clarkson}. 

One minor drawback of Clarkson's algorithm is it has a dependence on $\log n$.
In constant dimensions, our $\widetilde{\Omega}(s + L)$ lower bound in the blackboard model holds only when $n = 2^{\Omega(L)}$, in which case the communication complexity of Clarkson's algorithm will be $\widetilde{O}(sL + L^2)$.

Under the additional assumption that each constraint in the linear program is placed on a random server, we develop an algorithm with communication complexity $\widetilde{O}(s+L)$ in the blackboard model. 
To achieve this goal, we modify Seidel's classical algorithm and implement it in the distributed setting.
Seidel's algorithm benefits from the additional assumption from two aspects.
On the one hand, Seidel's classical algorithm needs to go through all the constraints in a random order, which can be easily achieved now since all constraints are placed on a random server.
On the other hand, Seidel's classical algorithm needs to make a recursive call each time it finds one of $d$ constraints that determines the optimal solution, and will make $\sum_{i = 1}^n d / i = \Theta(d \log n)$ recursive calls in expectation. 
To implement Seidel's algorithm in the distributed setting, each time we find one of the $d$ constraints that determines the optimal solution, the current server also needs to broadcast that constraint.
Thus, na\"ively we need to broadcast $O(d \log n)$ constraints during the execution, which would result in $O(s + L \log n)$ communication. 
Under the additional assumption, with good probability, the first server $P_1$ stores at least $\Omega(n / s)$ constraints.
Since the first server $P_1$ does not need to make any recursive calls or broadcasts, the total number of recursive calls (and thus broadcasts) will be $\sum_{i = \Omega(n / s)}^n d / i = \Theta(d \log s)$.
The formal analysis is given in Section \ref{sec:seidel}.


For convex programming, we have to use a more general algorithm. We use a refined version of the classical center-of-gravity method.
The basic idea is to round violated constraints that are used as cutting planes to $O(d\log d)$ bits. We 
optimize over the ellipsoid method in the following two ways.
First, we round the violated constraint sent in each iteration by
locally maintaining an ellipsoid to ensure the rounding error does not affect the algorithm. 
Roughly speaking, each server maintains a well-rounded current feasible set, and the number of bits needed in each round is thus only $\tilde{O}(d)$.
Secondly, we use the center of gravity method to make sure the volume is cut by a constant
factor rather than a $(1-1/d)$ factor in each iteration, even when constraints are rounded. 
See Section \ref{sec:cog} for the formal analysis. 

\section{Preliminaries}
\subsection{Notation}
For $m$ matrices $A^{(1)} \in \mathbb{R}^{d \times n_1}, A^{(2)} \in \mathbb{R}^{d \times n_2}, \ldots, A^{(m)} \in \mathbb{R}^{d \times n_m}$, 
we use $[A^{(1)} ~ A^{(2)} ~ \cdots ~ A^{(m)}]$ to denote the matrix in $\mathbb{R}^{d \times (n_1 + n _2 + \cdots + n_m)}$ 
whose first $n_1$ columns are the same as $A^{(1)} $, 
the next $n_2$ columns are the same as $A^{(2)} $,
\ldots,
and the last $n_m$ columns are the same as $A^{(m)}$.

For a matrix $A \in \mathbb{R}^{n \times d}$, we use $\spa(A) = \{Ax \mid x \in \mathbb{R}^d\}$ to denote the subspace spanned by the columns of the matrix $A$.
For a set of vectors $S \subseteq R^d$, we use $\spa(S)$ to denote the subspace spanned by the vectors in $S$.
For a set of linear equations $C$, we also $\spa(C)$ to denote all linear combinations of linear equations in $C$.
We use $A_i$ to denote the $i$-th column of $A$ and $A^i$ to denote the $i$-th row of $A$.
We use $A^{\dagger}$ to denote the Moore-Penrose inverse of $A$.
We use $\rank(A)$ to denote the rank of $A$ over the real numbers and $\rank_p(A)$ to denote the rank of $A$ over the finite field $\mathbb{F}_p$.

For a vector $x \in \mathbb{R}^d$, we use $\|x\|_p = \left(\sum_{i=1}^d |x_i|^p \right)^{1 / p}$ to denote its $\ell_p$ norm.
For two vectors $x$ and $y$, we use $\langle x, y \rangle$ to denote their inner product. 

For matrices $A$ and $B$, we say $A \approx_{\kappa} B$ if and only if 
\[
\frac{1}{\kappa} B \preceq A \preceq \kappa B,
\]
where $\preceq$ refers to the L\"owner partial ordering of matrices, i.e., $A \preceq B$ if $B - A$ is positive semi-definite. 

\subsection{Models of Computation and Problem Settings}
We study the distributed linear regression problem in two distributed models: the coordinator model (a.k.a. the message passing model) and the blackboard model. 
The coordinator model represents distributed computation systems with point-to-point communication, while the blackboard model represents those where messages can be broadcasted to each party.

In the {\em coordinator model}, there are $s \ge 2$ servers $P_1, P_2, \ldots, P_s$, and one coordinator.
These $s$ servers can directly send messages to the coordinator through a two-way private channel. 
The computation is in terms of rounds: at the beginning of each round, the coordinator sends a message to some of the $s$ servers, and then each of those servers that have been contacted by the coordinator sends a message back to the coordinator. 

In the alternative {\em blackboard model}, the coordinator is simply a blackboard where the $s$ servers $P_1, P_2, \ldots, P_s$ can share information; in other words, if one server sends a message to the coordinator/blackboard then the other $s-1$ servers can see this information without further communication. 
The order for the servers to send messages is decided by the contents of the blackboard.

For both models we measure the {\em communication cost} which is defined to be the total number of bits sent through the channels.

In the {\em distributed linear system problem}, there is a data matrix $A \in \mathbb{R}^{n \times d}$ and a vector $b$ of observed values.
All entries in $A$ and $b$ are integers between $[-2^L, 2^L]$, where $L$ is the {\em bit complexity}.
The matrix $[A~b]$ is distributed row-wise among the $s$ servers $P_1, P_2, \ldots, P_s$.
More specifically, for each server $P_i$, there is a matrix $[A^{(i)}~b^{(i)}]$ stored on $P_i$, which is a subset of rows of $[A ~ b]$.
Here we assume $\{[A^{(1)}~b^{(1)}], [A^{(2)}~b^{(2)}], \ldots, [A^{(s)}~b^{(s)}]\}$ is a {\em partition} of all rows in $[A~b]$.
The goal of the {\em feasibility testing} problem is to design a protocol, such that upon termination of the protocol, the coordinator reports whether the linear system $Ax = b$ is feasible or not.
The goal of the {\em linear system solving} problem is to design a protocol, such that upon termination of the protocol, either the coordinator outputs a vector $x^* \in \mathbb{R}^d$, such that $Ax^* = b$, or the coordinator reports the linear system $Ax = b$ is infeasible. 
It can be seen that the linear system solving problem is strictly harder than the feasibility testing problem. 

In the {\em distributed linear regression problem}, there is a data matrix $A \in \mathbb{R}^{n \times d}$ and a vector $b$ of observed values, which is distributed in the same way as in the distributed linear system problem. 
The goal of the distributed $\ell_p$ regression problem is to design a protocol, such that upon termination of the protocol, the coordinator outputs a vector $x^* \in \mathbb{R}^d$ to minimize $\|Ax - b\|_p$.

In the {\em distributed linear programming problem}, there is a matrix $A \in \mathbb{R}^{n \times d}$ and a vector $b$, which is distributed in the same way as in the distributed linear system problem. 
The goal of the {\em feasibility testing} problem is to design a protocol, such that upon termination of the protocol, the coordinator reports whether the linear program $Ax \le b$ is feasible or not.
In the {\em linear programming solving} problem, the goal is to design a protocol, such that upon termination of the protocol, the coordinator outputs a vector $x^* \in \mathbb{R}^d$ such that $Ax^* \le b$ is satisfied.
There can also be a vector $c \in \R^d$ which is known to all servers, and in this case the goal is to minimize (or maximize) $\langle c, x \rangle$ under the constraint that $Ax \le b$.

\subsection{Row Sampling Algorithms}
\begin{definition}[\cite{cohen2015uniform}]\label{def:leverage_score}
Given a matrix $A \in \mathbb{R}^{n \times d}$. The leverage score of a row $A^i$ is defined to be
\[
\tau_i(A) = A^i (A^TA)^{\dagger} (A^i)^T. 
\]
Given another matrix $B \in \mathbb{R}^{n' \times d}$, the generalized leverage score of a row $A^i$ w.r.t. $B$ is defined to be
\[
\tau^B_i(A) = 
\begin{cases}
A^i (B^TB)^{\dagger} (A^i)^T & \text{if } A^i \perp \ker(B),\\
\infty & \text{otherwise}.
\end{cases}
\]
\end{definition}
\begin{definition}[\cite{cohen2015p}]\label{def:lewis_weight}
Given a matrix $A \in \mathbb{R}^{n \times d}$. The $\ell_1$ Lewis weights $\{\overline{w}_i\}_{i = 1}^n$ are the unique weights such that for each $i \in [n]$ we have
\[
\overline{w}_i = \tau_i\left(\overline{W}^{-1/2} A\right),
\]
where $\overline{W}$ is the diagonal matrix formed by putting $\{\overline{w}_i\}_{i = 1}^n$ on the diagonal.
\end{definition}
\begin{theorem}[$\ell_2$ Matrix Concentration Bound, Lemma 4 in \cite{cohen2015uniform}]\label{thm:l2chernoff}
There exists an absolute constant $C$ such that for any matrix $A \in \mathbb{R}^{n \times d}$ and any set of sampling values $p_i$ satisfying
$$
p_i \ge C \tau_i(A) \log d \varepsilon^{-2},
$$
if we generate a matrix $S$ with $N = \sum_{i=1}^n p_i$ rows, each chosen independently as the $i$-th basis vector, times $p_i^{-1/2}$ with probability $p_i / N$, then with probability at least $0.99$, for all vector $x \in \mathbb{R}^d$,
\[
(1 - \varepsilon) \|Ax\|_2 \le \|SAx\|_2 \le (1 + \varepsilon)\|Ax\|_2.
\]
\end{theorem}
\begin{theorem}[$\ell_1$ Matrix Concentration Bound, Theorem 7.1 in \cite{cohen2015p}]\label{thm:l1chernoff}
There exists an absolute constant $C$ such that for any matrix $A \in \mathbb{R}^{n \times d}$ and any set of sampling values $p_i$ satisfying
$$
p_i \ge C \overline{w}_i \log d \varepsilon^{-2},
$$
if we generate a matrix $S$ with $N = \sum_{i=1}^n p_i$ rows, each chosen independently as the $i$-th basis vector, times $p_i^{-1}$ with probability $p_i / N$, then with probability at least $0.99$, for all vectors $x \in \mathbb{R}^d$,
\[
(1 - \varepsilon)  \|Ax\|_1 \le \|SAx\|_1 \le (1 + \varepsilon)\|Ax\|_1.
\]
Here $\{\overline{w}_i\}_{i = 1}^n$ are the $\ell_1$ Lewis weights of the matrix $A$.
\end{theorem}

\section{Communication Complexity Lower Bound for Linear Systems}
\subsection{The Hard Instance}
In this section, we construct a family of matrices, which will be used to prove a communication complexity lower bound in the subsequent section. 

We first introduce {\em generalized binomial distributions}.
\begin{definition}
For any $0 \le \mu \le 1$, let $\sumsign^{(\mu)} \in \{-1, 0, 1\}$ be a random variable which takes $+1$ or $-1$ with probability $\mu / 2$, and $0$ with probability $1 - \mu$.
Let $\sumsign_t^{(\mu)}$ be a random variable with the same distribution as the sum of $t$ i.i.d. copies of $\sumsign^{(\mu)}$.
For simplicity we use $\sumsign$ and $\sumsign_t$ to denote $\sumsign^{(1)}$ and $\sumsign_t^{(1)}$, respectively.
\end{definition}

We need the following theorem on the singularity probability of discrete random matrices.
\begin{theorem}\label{thm:singularity_main}
Let $M_n \in \mathbb{R}^{n \times n}$ be a matrix whose entries are i.i.d. random variables with the same distribution as $\sumsign_t$, for sufficiently large $t$,
$$
\Pr\left[M_n \text{ is singular}\right] \le t^{-C n},
$$
where $C > 0$ is an absolute constant. 
\end{theorem}
The proof of Theorem \ref{thm:singularity_main} closely follows previous approaches for bounding the singularity probability of random $\pm 1$ matrices (see, e.g., \cite{kahn1995probability, tao2006random, tao2007singularity, bourgain2010singularity}.).
For completeness, we include a proof of Theorem \ref{thm:singularity_main} in Section~\ref{sec:singularity}.

\begin{lemma}\label{lem:prob_method}
For any $d > 0$ and sufficiently large $t$, there exists a set of matrices $\mathcal{T} \subseteq \mathbb{R}^{d \times (d - 1)}$ with integral entries in $[-t, t]$ for which $|\mathcal{T}| = t^{\Omega(d^2)}$ and
\begin{enumerate}
\item For any $T \in \mathcal{T}$, $\rank(T) = d - 1$;
\item For any $S, T \in \mathcal{T}$ such that $S \neq T$, $\spa([S~T]) = \mathbb{R}^d$.
\end{enumerate}
\end{lemma}
\begin{proof}
We use the probabilisitic method to prove existence. 
We use $\bado \subset \mathbb{R}^{d \times (d - 1)}$ to denote the set
$$
\bado = \{B \in \mathbb{R}^{d \times (d - 1)} \mid \Pr[X \in \mathrm{span}(B)] \ge t^{-Cd / 2}\text{ or } \rank(B) < d - 1\},
$$
where $X \in \mathbb{R}^d$ is a vector whose entries are i.i.d. random variables with the same distribution as $\sumsign_t$
and 
$C$ is the constant in Theorem \ref{thm:singularity_main}.

Consider a random matrix $A \in \mathbb{R}^{d \times (d - 1)}$ whose entries are i.i.d. random variables with the same distribution as $\sumsign_t$,
we have
\begin{equation}\label{pr:bad}
\Pr[A \in \bado] \le t^{-C d / 2},
\end{equation}
since otherwise, if we use $X \in \mathbb{R}^d$ to denote a vector whose entries are i.i.d. random variables with the same distribution as $\sumsign_t$, we have
\begin{align*}
&\Pr[\rank([A~X]) < d] \\
\ge& \Pr[\rank([A~X]) < d \mid  A \in \bado] \cdot \Pr[ A \in \bado]  \\
> &t^{-C d},
\end{align*}
which violates Theorem \ref{thm:singularity_main}.

For any fixed $A \in \mathbb{R}^{d \times (d - 1)} \setminus \bado$, 
consider a random matrix $B \in \mathbb{R}^{d \times (d - 1)}$ whose entries are i.i.d. random variables with the same distribution as $\sumsign_t$. We have, 
\begin{equation}\label{equ:span}
\Pr[\spa([A~B]) = \mathbb{R}^d] \ge 1 -  \Pr\left[\bigcap_{i = 1}^{d - 1} B_i \in \spa(A)\right] \ge 1 - t^{-C d(d -1 ) / 2},
\end{equation}
which follows from the definition of $\bado$ and the independence of columns of $B$.

Now we construct a multiset $\mathcal{S}$ of $|\mathcal{S}| = t^{Cd(d - 1) / 6}$ matrices chosen with replacement, each of dimension $d \times (d - 1)$ and with i.i.d. entries having the same distribution as $\sumsign_t$.
By (\ref{pr:bad}) and linearity of expectation, we have 
$$\E[|\mathcal{S} \cap \bado|] \le t^{Cd(d - 1) / 6} \cdot t^{-Cd / 2}.$$
We use $\mathcal{E}_1$ to denote the even that 
$$|\mathcal{S} \cap \bado| \le 4\E[|\mathcal{S} \cap \bado|] \le4 t^{Cd(d - 1) / 6} \cdot t^{-C d / 2},$$ 
which holds with probability at least $3 / 4$ by using Markov's inequality.

We use $\mathcal{E}_2$ to denote the event that 
$$
\forall S \in \mathcal{S} \setminus \bado, \forall T \in \mathcal{S} \setminus \{S\}, \spa([S~T]) = \mathbb{R}^d.
$$
Using a union bound and (\ref{equ:span}), the probability that $\mathcal{E}_2$ holds is at least 
$$
1 - |\mathcal{S}|^2 t^{-Cd(d - 1) / 2} = 1 - t^{-\Omega(d^2)}.
$$

Thus by a union bound, the probability that both $\mathcal{E}_1$ and $\mathcal{E}_2$ hold is strictly larger than zero, which implies there exists a set $\mathcal{S}$ such that $\mathcal{E}_1$ and $\mathcal{E}_2$ hold simultaneously. 
Now we consider $\mathcal{T} = \mathcal{S} \setminus \bado$. 
Since $\mathcal{E}_1$ holds, we have $|\mathcal{T}| \ge t^{\Omega(d^2)}$.
$\mathcal{E}_2$ implies that all elements in $\mathcal{T}$ are distinct, and furthermore for any $S, T \in \mathcal{T}$ such that $S \neq T$, we have
$$
\spa([S~T]) = \mathbb{R}^d.
$$
\end{proof}

\begin{lemma}\label{lem:hard_instance}
For any $d > 0$ and sufficiently large $t$, there exists a set of matrices $\mathcal{H} \subseteq \mathbb{R}^{d \times d}$ with integral entries in $[-t, t]$
for which $|\mathcal{H}| = t^{\Omega(d^2)}$ and
\begin{enumerate}
\item For any $T \in \mathcal{H}$, $T$ is non-singular;
\item For any $S, T \in \mathcal{H}$, $S^{-1} e_d \neq T^{-1} e_d$, where $e_d$ is the $d$-th standard basis vector. 
\end{enumerate}
\end{lemma}
\begin{proof}
Consider the matrix set $\mathcal{T}$ constructed in Lemma \ref{lem:prob_method}.
For each $T \in \mathcal{T}$, we add $[T~e_i]^T$ into $\mathcal{H}$, where $e_i$ is the $i$-th standard basis vector such that $e_i \notin \spa(T)$.
Clearly $[T~e_i]^T$ is non-singular since $\rank(T) = d - 1$ and $e_i \notin \spa(T)$.

Now suppose there exists $S, T \in \mathcal{H}$ such that $S^{-1} e_d = T^{-1} e_d$, which means there exists some $x \in \mathbb{R}^d$ such that $Sx = e_d$ and $Tx = e_d$.
This implies there exists some $S', T' \in \mathcal{T}$ such that $(S')^T x = 0$ and $(T')^T x= 0$.
However, $x$ must be $0^d$ since $\spa([S' ~ T']) = \mathbb{R}^d$, which implies $Sx = Tx = 0 \neq e_d$.
Thus for any $S, T \in \mathcal{H}$, $S^{-1} e_d \neq T^{-1} e_d$.
\end{proof}
\subsection{Deterministic Lower Bound for the \textsf{Equality} Problem} \label{sec:eq}
In this section, we prove our deterministic communication complexity lower bound for the \textsf{Equality} problem in the coordinator model, which will be used as an intermediate problem in Section \ref{sec:det_lb_ls}.
In the \textsf{Equality} problem, each server $P_i$ receives a binary string $t_i \in \{0, 1\}^n$. 
The goal is to test whether $t_1 = t_2 = \ldots = t_s$.
We will prove an $\Omega(sn)$ lower bound for deterministic communication protocols. 

The case $s = 2$ has a well-known $\Omega(n)$ lower bound.

\begin{lemma} [See, e.g., {\cite[p11]{CC}}] \label{lem:eq_2}
Any deterministic protocol for solving the \textsf{Equality} problem with $s = 2$ requires $\Omega(n)$ bits of communication. 
\end{lemma}

Our plan is to reduce the case $s = 2$ to the case $s > 2$, using the symmetrization technique 
\cite{phillips2012lower, woodruff2014optimal}.
Suppose there exists a deterministic communication protocol $\mathcal{P}$ for the \textsf{Equality} problem with $s$ servers, and the communication complexity of $\mathcal{P}$ is $C_{\mathcal{P}}(n)$ where $n$ is the length of the binary strings received by the servers.
We show how to solve the case $s = 2$ using $\mathcal{P}$.

Suppose Alice receives a binary string $x \in \{0, 1\}^{n - 1}$ and Bob receives a binary string $y \in \{0, 1\}^{n - 1}$.
We show that by using the protocol $\mathcal{P}$, they can judge whether $x = y$ or not using $O(C_{\mathcal{P}}(n) / s)$ communication.
Thus by Lemma \ref{lem:eq_2}, we must have
$$
O(C_{\mathcal{P}}(n) / s) \ge \Omega(n),
$$
which implies $C_{\mathcal{P}}(n) = \Omega(sn)$. 

Since $\mathcal{P}$ is deterministic, by averaging, there exists a fixed server $P_i$ and a fixed set $S \subseteq \{0, 1\}^n$ with size $|S| = 2^{n- 1}$, such that for any $t \in S$, when all servers have the same input $t$, the total communication complexity beteen $P_i$ and the coordinator is upper bounded by $2 C_{\mathcal{P}}(n) / s$.
Now we fix a bijection $g : S \to \{0, 1\}^{n - 1}$.
Alice plays the role of server $P_i$ in $\mathcal{P}$, and sets the input of $P_i$ to be $g(x)$.
Bob plays the role of the coordinator and all servers $P_j$ for $i \neq j$, and sets the input of $P_j$ to be $g(y)$ for all $i \neq j$.
To simulate the protocol $\mathcal{P}$, Alice and Bob need to communicate if and only if server $P_i$ needs to communicate with the coordinator.
If the total amount of communication between Alice and Bob exceeds $2C_{\mathcal{P}}(n) / s$ then they terminate and return $x \neq y$. 
Alice and Bob return $x = y$ if and only if the protocol $\mathcal{P}$ returns $g(x)= g(y)$.

Now we analyze the correctness and the efficiency of the reduction. 
When $x = y$, we have $g(x) = g(y) \in S$, and by definition of $P_i$ and $S$, we must have the total communication complexity between $P_i$ and the coordinator, and thus that between Alice and Bob, is upper bounded by $2 C_{\mathcal{P}}(n) / s$.
Also the protocol must return $g(x) = g(y)$ due to the correctness of $\mathcal{P}$.
When $x \neq y$, either the total amount of communication between Alice and Bob exceeds $2C_{\mathcal{P}}(n) / s$, in which case they will return $x \neq y$.
Otherwise $\mathcal{P}$ returns $g(x) \neq g(y)$ due to its correctness. 

Formally, we have proved the following theorem.

\begin{theorem}\label{thm:eq}
Any deterministic protocol for solving the \textsf{Equality} problem with $s$ servers in the coordinator model requires $\Omega(sn)$ bits of communication. 
\end{theorem}

\subsection{Deterministic Lower Bound for Testing Feasibility of Linear Systems} \label{sec:det_lb_ls}
In this section, we prove our deterministic communication complexity lower bound for testing the feasibility of linear systems, in the coordinator model and the blackboard model. 

\begin{theorem}\label{thm:eq_lb}
For any deterministic protocol $\mathcal{P}$,
\begin{itemize}
\item If $\mathcal{P}$ can test whether $Ax = b$ is feasible or not in the coordinator model, then
the communication complexity of $\mathcal{P}$ is $\Omega(sd^2L)$;
\item If $\mathcal{P}$ can test whether $Ax = b$ is feasible or not in the blackboard model, then
the communication complexity of $\mathcal{P}$ is $\Omega(s + d^2L)$;
\end{itemize}
\end{theorem}

\begin{proof}
Consider the set $\mathcal{H}$ constructed in Lemma \ref{lem:hard_instance} with $t = 2^L$.
In the hard instance, each server $P_i$ receives a matrix $H_i \in \mathcal{H}$.
The linear system stored on each server is just $H_i x = e_d$. 
Due to Lemma \ref{lem:hard_instance}, the entire linear system is feasible if and only if $H_1 = H_2 = \ldots = H_s$.
Since $|\mathcal{H}| = 2^{\Omega(d^2L)}$, we can reduce the \textsf{Equality} problem in Section \ref{sec:eq} to solving a linear system, with $n = \Theta(d^2L)$.
By Theorem \ref{thm:eq_lb}, this implies an $\Omega(sd^2L)$ lower bound in the coordinator model. 

In the blackboard model, the $\Omega(d^2L)$ bound follows from the case when $s = 2$. When $s = 2$, the blackboard model is essentially the same as the coordinator model, up to constants in the communication complexity. The $\Omega(s)$ lower bound follows from the fact that each server needs to communicate at least $1$ bit. 
\end{proof}

\subsection{Randomized Lower Bound for Solving Linear Systems}\label{sec:ls_rand_lb}
In this section, we prove randomized communication complexity lower bounds for solving linear systems. 
We first prove an $\Omega(d^2L)$ lower bound, which already holds for the case $s = 2$. 
When $s = 2$ the coordinator model and the blackboard model are equivalent in terms of communication complexity, and thus we shall not distinguish these two models in the remaining part of this proof. 

Consider the set $\mathcal{H}$ constructed in Lemma \ref{lem:hard_instance} with $t = 2^L$.
In the hard instance, only server $P_1$ receives a matrix $H \in \mathcal{H}$, and the goal is to let the coordinator output the solution to the linear system $Hx = e_d$.
For any two $H_1, H_2 \in \mathcal{H}$ and $H_1 \neq H_2$, we must have $H_1^{-1}e_d \neq H_2^{-1} e_d$.
Thus, by standard information-theoretic arguments, in order for the coordinator to output the solution to $Hx = e_d$, the communication complexity is at least $\Omega(\log(|\mathcal{H}|)) = \Omega(d^2L)$.

Formally, we have proved the following theorem. 
\begin{theorem}\label{thm:rand_eq_lb_s1}
Any randomized protocol that succeeds with probability at least $0.99$ for solving linear systems requires $\Omega(d^2L)$ bits of communication in the coordinator model and the blackboard model. The lower bound holds even when $s = 2$.
\end{theorem}

Now we prove another lower bound of $\widetilde{\Omega}(sd)$ for solving linear systems in the coordinator model. 
In the hard instance, the last server $P_s$ receives a vector $\hat{x} \in \{0, 1\}^d$, and the linear equations stored on server $P_s$ are simply $x = \hat{x}$, i.e., the solution vector $x$ should be exactly $\hat{x}$. This forces the solution vector to be some predefined binary vector $\hat{x}$.
The remaining $s - 1$ servers each receive a vector $a_i \in \{0, 1\}^d$, and the linear equation stored on $P_i$ is 
$$
\sum_{j = 1}^d a_j x_j = 1.
$$
Also, it is guaranteed that for each $i \in [s]$, $\langle a_i, \hat{x}\rangle = 0$ or $1$.

Here we interpret the vector $\hat{x}$ as the characteristic vector of a set $S_{\hat{x}} \subseteq [d]$, and interpret each vector $a_i$ also as the characteristic vector a set $S_{a_i}$.
Thus, testing the feasibility of the linear system is equivalent to testing whether the set $S_{\hat{x}}$ owned by the server $P_s$ is disjoint with the set owned by any other player,
which is the OR of $s - 1$ copies of the two-player set-disjointness problem. 
The communication complexity for the latter problem has been studied in \cite{phillips2012lower, WZ13}.
Combining Lemma 2.2 in \cite{phillips2012lower} with Theorem 1 in \cite{WZ13}, for any communication protocol that succeeds with probability $1 - 1 / s^3$, the communication complexity is lower bounded by $\Omega(sd)$.
By standard repetition arguments, this implies for any randomized communication protocol that succeeds with probability at least $0.99$, the communication complexity is lower bounded by $\Omega(sd / \log s)$.

Combining this lower bound and the trivial $\Omega(s)$ lower bound in the blackboard model with Theorem \ref{thm:rand_eq_lb_s1}, we have the following theorem.
\begin{theorem}\label{thm:rand_eq_lb}
Any randomized protocol that succeeds with probability at least $0.99$ for solving linear systems requires $\wt{\Omega}(sd + d^2L)$ bits of communication in the coordinator model and $\Omega(s + d^2L)$ bits of communication in the blackboard model. 
\end{theorem}

\section{Communication Protocols for Linear Systems}
\subsection{Testing Feasibility of Linear Systems}\label{sec:rand_feasible}
In this section we present a randomized communication protocol for testing feasibility of linear systems, which has communication complexity $O(sd^2\log(dL))$ in the coordinator model and $O(s + d^2\log(dL))$ in the blackboard model.
The protocol is described in Figure \ref{alg:rand_feasibility}.

\begin{figure}[!htb]
\begin{framed}
\begin{enumerate}
\item\label{step:gen_p} The coordinator generates a random prime number $p$ in $[2, \poly(dL)]$ and sends $p$ to all servers. 
\item Each server $P_i$ tests the feasibility of its {\em own} linear system. If the linear system is infeasible then $P_i$ terminates the protocol. 
\item \label{step:rand_ls_init} Each server maintains the same set of linear equations $C$. At the beginning $C$ is the empty set.
\item \label{step:rand_feas_for} For $i = 1, 2, \ldots, s$
\begin{enumerate}[(a)]
\item \label{step:check_fea}$P_i$ checks whether the linear system formed by all linear equations in $C$ and all linear equations stored on $P_i$ is feasible over the finite field $\mathbb{F}_p$. $P_i$ terminates the protocol if it is infeasible. 
\item \label{step:send_mod_p}For each linear equation stored on $P_i$ that is linearly independent with linear equations in $C$ over the finite field $\mathbb{F}_p$, $P_i$ sends that linear equation to all servers, after taking the residual of each entry modulo $p$. All servers add that linear equation into $C$.
\end{enumerate}
\end{enumerate}
\end{framed}
\caption{Randomized protocol for testing feasibility}
\label{alg:rand_feasibility}
\end{figure}

We first bound the communication complexity of the protocol in Figure \ref{alg:rand_feasibility}.
Clearly, Step \ref{step:gen_p} has communication complexity at most $O(s \log(dL))$.
During the execution of the whole protocol, at most $d$ linear equations will be added into $C$.
The communication complexity for sending each linear equation is $O(sdL \log p)$ in the coordinator model, and $O(dL \log p)$ in the blackboard model.
Thus, the total communication complexity is $O(sd^2 \log (dL))$ in the coordinator model, and $O(s + d^2\log(dL))$ in the blackboard model. 

To prove the correctness of this protocol, we need the following lemma. 
\begin{lemma}\label{lem:mod_rank}
Given a matrix $A \in \mathbb{R}^{m \times n}$ where each entry is an integer in $[-2^L, 2^L]$ and $\rank(A) = r$. 
Suppose $p$ is chose uniformly at random from all primes numbers in $[2, \poly(rL)]$.
\begin{enumerate}[(i)]
\item $\rank_p(A) \le \rank(A)$;
\item With probability at least $0.99$, $\rank_p(A) = \rank(A)$.
\end{enumerate}
\end{lemma}
\begin{proof}
The point (i) is immediate.
For point (ii), there exists a square submatrix $A'$ of $A$ with size $r \times r$ which is non-singular over real numbers, which implies the determinant of $A'$ is non-zero as a real number.
Since all entries of $A'$ are integers in $[-2^L, 2^L]$, the determinant of $A'$ as a real number is an integer in $[-r! 2^{rL}, r! 2^{rL}]$.
Thus, the determinant of $A'$ has at most $\poly(rL)$ prime factors. 
According to the Prime Number Theorem, there are at least $n$ distinct prime numbers in the range $[2, n^2]$, for sufficiently large $n$.
Thus, by adjusting constants, $p$ is not a prime factor of the determinant of $A'$ with probability at least $0.99$, in which case $\rank(A) = \rank_p(A)$.
\end{proof}

Notice that the protocol in Figure \ref{alg:rand_feasibility} is basically testing the feasibility of the linear system over the finite field $\mathbb{F}_p$, for a randomly chosen prime number $p$.
Before the execution of the $i$-th loop of Step \ref{step:rand_feas_for}, the set $C$ is a maximal set of linearly independent equations for all linear equations stored on the first $i - 1$ servers $P_1, P_2, \ldots, P_{i - 1}$.
Here the linear independence is defined over the finite field $\mathbb{F}_p$. 
During the execution of the $i$-th loop of Step \ref{step:rand_feas_for}, server $P_i$ considers each linear equation stored on itself one by one, sends the linear equation to all other servers and adds the linear equation to set $C$ if that linear equation is linearly independent with all existing linear equations in $C$.
If server $P_i$ finds that the set $C$ becomes infeasible after adding linear equations stored on $P_i$, then $P_i$ terminates the protocol. 

Consider a linear system $Ax = b$ where $A \in \mathbb{R}^{n \times d}$ and $b \in \mathbb{R}^n$ and all entries of $A$ and $b$ are integers in the range $[-2^L, 2^L]$.
If $Ax = b$ is feasible over the real numbers, then it will also be feasible over the finite field $\mathbb{F}_p$.
If $Ax = b$ is infeasible, then we have $\rank([A~b]) > \rank(A)$.
By Lemma \ref{lem:mod_rank}, $\rank_p(A) \le \rank(A)$, and since $\rank([A~b]) \le d + 1$, with probability at least $0.99$, $\rank_p([A~b]) = \rank([A~b])$, which implies with probability $0.99$, $Ax = b$ is still infeasible over the finite field $\mathbb{F}_p$.
Since the protocol in Figure \ref{alg:rand_feasibility} tests the feasibility of the linear system over the finite field $\mathbb{F}_p$, the correctness follows. 

Formally, we have proved the following theorem.

\begin{theorem}\label{thm:rand_feasibility}
The protocol in Figure \ref{alg:rand_feasibility} is a randomized protocol for testing feasibility of linear systems and has communication complexity $O(sd^2\log(dL))$ in the coordinator model and $O(s + d^2\log(dL))$ in the blackboard model. The protocol succeeds with probability at least $0.99$.
\end{theorem}
\subsection{Solving Linear Systems}\label{sec:rand_ls_solve}
In this section we present communication protocols for solving linear systems.
We start with deterministic protocols, in which case we can get a protocol with communication complexity $O(sd^2L)$ in the coordinator model and $O(s + d^2L)$ in the blackboard model. 

In order to solve linear systems, we can still use the protocol in Figure \ref{alg:rand_feasibility}, but we don't use the prime number $p$ any more.
In Step \ref{step:check_fea} of the protocol, we no longer check the feasibility over the finite field. 
In Step \ref{step:send_mod_p} of the protocol, we no longer takes the residual modulo $p$ before sending the linear equations. 
At the end of the protocol, each server can use the set of linear equations $C$, which is a maximal set of linear equations of the original linear system, to solve the linear system.
The communication complexity is $O(sd^2L)$ in the coordinator model and $O(s + d^2L)$ in the blackboard model since at most $d$ linear equations will be added into the set $C$, and each linear equation requires $O(dL)$ bits to describe. 

Formally, we have proved the following theorem.
\begin{theorem}
There exists a deterministic protocol for solving linear systems which has communication complexity $O(sd^2L)$ in the coordinator model and $O(s + d^2L)$ in the blackboard model.
\end{theorem}

Now we turn to randomized protocols. We describe a protocol for solving linear systems with communication complexity $\widetilde{O}(d^2L + sd)$ in the coordinator model.
The description is given in Figure \ref{alg:rand_ls}.

\begin{figure}[!htb]
\begin{framed}
\begin{enumerate}
\item Each server $P_i$ tests the feasibility of their {\em own} linear system. If the linear system is infeasible then $P_i$ terminates the protocol. Otherwise, each server $P_i$ finds a maximal set of linearly independent linear equations, say $S_i$. 
\item The coordinator maintains a set of linear equations $C$. At the beginning $C$ is the empty set.
\item \label{step:rand_ls_for} For $i = 1, 2, \ldots, s$
\begin{enumerate}[(a)]
\item \label{step:rand_ls_rep} Repeat the followings for $O(\log d)$ times
\begin{enumerate}
\item \label{step:rand_lin_comb} Server $P_i$ calculates a linear equation $c = \sum_{t \in S_i} r_t \cdot t$, here $\{r_t\}_{t \in S_i}$ is a set of i.i.d. random signs. $P_i$ sends the linear equation $c$ to the coordinator.
\item \label{step:rand_add_or_term} The coordinator terminates the protocol if $C \cup \{c\}$ is infeasible. Otherwise if $c$ is not a linear combination of those linear equations in $C$, then the coordinator adds $c$ into $C$, and then goes to Step \ref{step:rand_ls_rep}, i.e., repeats another $O(\log d)$ times. 
\end{enumerate}
\end{enumerate}
\item The coordinator obtains the solution by solving all equations in $C$.
\end{enumerate}
\end{framed}
\caption{Randomized protocol for solving linear systems in the coordinator model}
\label{alg:rand_ls}
\end{figure}

Now we prove the correctness of the protocol. 
We first note a few simple properties of the protocol.
For each $i \in [s]$, after executing the $i$-th loop of Step \ref{step:rand_ls_for}, we have $C \subseteq \spa \left(\bigcup_{j \le i} S_j\right)$.
Furthermore, at Step \ref{step:rand_add_or_term}, we must have $c \in \spa(S_i)$.
This means if $C \cup \{c\}$ is infeasible, then the original linear system must be infeasible. 

Thus, it suffices to show that for each $i \in [s]$, if there exists $s \in S_i$ such that $\left(\bigcup_{j < i} S_j \right) \cup\{s\}$ is infeasible, then the protocol is terminated, and otherwise after executing the $i$-th loop of Step \ref{step:rand_ls_for}, we have $\spa(C) = \spa\left(\bigcup_{j \le i} S_j\right)$.

Suppose before the execution of the $i$-th loop of Step \ref{step:rand_ls_for} we have $\spa(C) = \spa\left(\bigcup_{j < i} S_j\right)$ and $\spa\left(\bigcup_{j < i} S_j\right)$ is feasible, and the protocol is executing the $i$-th loop of Step \ref{step:rand_ls_for}. There are two cases here.
\begin{description}
\item[Case 1: $\spa\left(\bigcup_{j\le i} S_j \right) = \spa\left(\bigcup_{j < i} S_j\right) = \spa(C) $.] In this case, $C$ will remain unchanged during the $i$-th loop of  Step \ref{step:rand_ls_for}. 
\item[Case 2: There exists $s \in S_i$, $s \notin \spa(C)$.] In this case, we claim that with probability at least $1/2$, the linear equation $c$ calculated at Step \ref{step:rand_lin_comb} satisfies $c \notin \spa(C)$. This can be seen since for any linear combination $c = \sum_{t \in S_i} r_t \cdot t$, if we flip the sign of $r_s$ and obtain $\hat{c}$, then either $c \notin \spa(C)$ or $\hat{c} \notin \spa(C)$, since $c - \hat{c} = \pm 2s$ and $s \notin \spa(C)$.
\end{description}
Thus, if there exists $s \in S_i$ such that $s \notin \spa(C)$, then with probability $1 - 1 / \poly(d)$, at least one of the linear equations $c$ calculated at Step \ref{step:rand_lin_comb} satisfies $c \notin \spa(C)$, in which case the protocol terminates if $C \cup \{c\}$ is infeasible, or $c$ is added into $C$ otherwise. 
Thus, if $ \spa\left(\bigcup_{j < i} S_j\right) \neq  \spa\left(\bigcup_{j \le i} S_j\right)$, then after the execution of the $i$-th loop of Step \ref{step:rand_ls_for}, with probability at least $1 - 1 / \poly(d)$, either the protocol (correctly) terminates, or we have $\spa(C) =  \spa\left(\bigcup_{j \le i} S_j\right)$.
The correctness of the protocol just follows by applying a union bound over all $i$ such that $ \spa\left(\bigcup_{j < i} S_j\right) \neq  \spa\left(\bigcup_{j \le i} S_j\right)$.
Notice that there are at most $d$ such $i$ we need to apply a union bound over.

Now we analyze the communication complexity of the protocol. 
Notice that at most $d$ linear equations will be added into $C$, and thus the total communication complexity associated with sending $c$ when $c$ is added into $C$ is upper bounded by $O(d^2L)$.
Furthermore, if $C \cup \{c\}$ is infeasible at Step \ref{step:rand_add_or_term}, then the protocol terminates and thus the communication complexity for sending $c$ associated such that case is upper bounded by $O(dL)$.
Furthermore, for each $i$, server $P_i$ will send $O(\log d)$ different linear equations $c$ to the coordinator, and if we implement the protocol na\"ively, then the total communication complexity is upper bounded by $\widetilde{O}(sdL)$.
Thus, the total communication complexity of the whole protocol is upper bounded by $\widetilde{O}(sdL + d^2L)$.

However, using Lemma \ref{lem:mod_rank} and the same argument as in Section \ref{sec:rand_feasible}, to implement Step \ref{step:rand_add_or_term}, it suffices to check if $C \cup \{c\}$ is feasible and if $c$ is a linear combination of existing linear equations in $C$, over the finite field $\mathbb{F}_p$, for a random prime number $p \in [2, \poly(dL)]$.
The correctness still follows since this check fails with probability at most $0.01$.
After this modification, the communication complexity is now upper bounded by $\widetilde{O}(sd + d^2L)$.

Formally, we have proved the following theorem.

\begin{theorem}\label{thm:rand_solve}
The protocol described in Figure \ref{alg:rand_ls} is a randomized protocol for solving linear systems which has communication complexity $\widetilde{O}(sd + d^2L)$ in the coordinator model. Here the $\widetilde{O}(\cdot)$ notation hides only $\polylog(dL)$ factors. 
The protocol succeeds with probability at least $0.99$.
\end{theorem}

\section{Communication Complexity Lower Bounds for Linear Regressions in the Blackboard Model} \label{sec:bb_reg_lb}
In this section, we prove communication complexity lower bounds for linear regression in the blackboard model. 

We first define the $k$-XOR problem and the $k$-MAJ problem.
In the blackboard model, each server $P_i$ receives a binary string $x_i \in \{0, 1\}^d$.
In the $k$-XOR problem, at the end of a communication protocol, the coordinator correctly outputs the coordinate-wise XOR of these vectors, for at least $0.99d$ coordinates.
In the $k$-MAJ problem, at the end of a communication protocol, the coordinator correctly outputs the coordinate-wise majority of these vectors, for at least $0.99d$ coordinates.

We need the following lemma for our lower bound proof.
\begin{lemma}
Any randomized communication protocol that solves the $k$-XOR problem or the $k$-MAJ problem and succeeds with probability at least $0.99$ has communication complexity $\Omega(dk)$.
\end{lemma}

\begin{proof}
The lower bound for $k$-XOR directly follows from \cite[Theorem 1.1]{phillips2012lower}.
Now we prove the lower bound for $k$-MAJ. 

First, consider a communication problem with two players.
Alice receives a binary string $x \in \{0, 1\}^d$ and $2k - 1$ binary strings $z_1, z_2, \ldots, z_{2k - 1} \in \{0, 1\}^d$.
Bob receives a binary string $y \in \{0, 1\}^d$ and the same $2k - 1$ binary strings $z_1, z_2, \ldots, z_{2k - 1} \in \{0, 1\}^d$.
These $2k + 1$ binary strings $x, y, z_1, z_2, \ldots, z_{2k - 1}$ are generated uniformly at random conditioned on the following constraint: for each coordinate $i \in [d]$, the $i$-th coordinate of $x, y, z_1, z_2, \ldots, z_{2k - 1}$ contains either $k$ zeros or $k + 1$ zeros.  
For each coordiante $i \in [d]$, whether the $i$-th coordinate of $x, y, z_1, z_2, \ldots, z_{2k - 1}$ contains $k$ zeros or $k + 1$ zeros is also chosen uniformly at random. 
In this communication problem, the goal of Alice is to output the vector $y$.

Now we prove a lower bound the communication problem defined above. 
Notice that for each coordinate $i \in [d]$, if $x, z_1, z_2, \ldots, z_{2k - 1}$ contains exactly $k$ zeros and $k$ ones at the $i$-th coordinate, then the $i$-th coordinate of $y$ will be uniformly at random.
By a Chernoff bound, with high probability, there exists a set $S \subseteq [d]$ with size $|S| \ge d / 10$, such that for each $i \in S$, the $i$-th coordinate of $x, z_1, z_2, \ldots, z_{2k - 1}$ contains exactly $k$ zeros and $k$ ones. 
By standard information-theoretic arguments, if at the end of the communication protocol, with constant probability, Alice correctly outputs the value of $y_i$ for at least $9/10$ fraction of $i \in S$, then the expected communication complexity is lower bounded by $\Omega(d)$, even with public randomness. See, e.g., Lemma 2.1 in \cite{phillips2012lower} for a formal proof. 

Now we reduce the problem mentioned above to $(2k + 1)$-MAJ and prove an $\Omega(dk)$ lower bound. 
Given any protocol $\mathcal{P}$ for the $(2k + 1)$-MAJ problem with expected communication complexity $C_{\mathcal{P}}$ on the distribution $x, y, z_1, z_2, \ldots, z_{2k - 1}$ mentioned above,
Alice and Bob first use public randomness to choose two distinct servers $s_1$ and $s_2$ uniformly at random, and then Alice and Bob simulate the protocol $\mathcal{P}$.
To simulate $\mathcal{P}$, Alice plays the role of server $s_1$ and Bob plays the role of server $s_2$. They both play the roles of all other players. 
Alice sets the input of $s_1$ to be $x$, and Bob sets the input of $s_2$ to be $y$.
The inputs of the other $2k - 1$ servers are set to be  $z_1, z_2, \ldots, z_{2k - 1}$.
To simulate $\mathcal{P}$, Alice and Bob need to communicate if and only if server $s_1$ or server $s_2$ needs to communicate with the coordinator since all other communication can be simulated by Alice and Bob themselves. 

By symmetry, the expected communication complexity between Alice and Bob is upper bounded by $2 C_{\mathcal{P}} / k$.
Furthermore, at the end of the protocol, Alice and Bob have the coordinate-wise majority of $x, y, z_1, z_2, \ldots, z_{2k - 1}$, for at least $0.99d$ coordinates. 
Thus, for at least $9/10$ fraction of $i \in S$, Alice knows the majority of the $i$-th coordinates of $x, y, z_1, z_2, \ldots, z_{2k - 1}$.
However, by definition of $S$, the majority of the $i$-th coordinates of $x, y, z_1, z_2, \ldots, z_{2k - 1}$ is exactly $y_i$.
Thus, by the $\Omega(d)$ lower bound mentioned above, we must have $2 C_{\mathcal{P}} / k = \Omega(d)$, which implies an $\Omega(dk)$ lower bound. 
\end{proof}

Now we give a reduction from $k$-MAJ to $(1 + \varepsilon)$-approximate $\ell_1$ regression in the blackboard model, and prove an $\Omega(d / \varepsilon)$ lower bound when $s > \Omega(1 / \eps)$.
In the hard case we assume $s = \Theta(1 / \eps)$, and we simply ignore all other servers if $s > \Theta(1 / \eps)$.
For each server $P_i$, its matrix $A^{(i)}$ is set to be the identity matrix $I_d \in \mathbb{R}^{d \times d}$, and $b^{(i)} \in \{0, 1\}^d$.
Notice that in such case, we can calculate the $\ell_1$ regression value separately for each coordinate $x_j$.
The optimal solution can be achieved by taking $x_j$ to be $m_j$, where $m_j$ is the {\em majority} of $b^{(1)}_j, b^{(2)}_j, \ldots, b^{(s)}_j$.
Notice that the $\ell_1$ regression value associated with the $j$-th coordinate in the optimal solution is upper bounded by $s = \Theta(1 / \eps)$, and thus the total $\ell_1$ regression value is upper bounded by $sd = \Theta(d / \eps)$ in the optimal solution.
Furthermore, if $|x_j - m_j| \ge 0.1$, then the $\ell_1$ regression value associated with the $j$-th coordinate by using $x_j$ will be at least $0.1$ larger than the $\ell_1$ regression value associated with the $j$-th coordinate in the optimal solution.

Now consider a $(1 + \varepsilon)$-approximate solution $x$. 
We claim that for at least $0.99d$ coordinates $x_i$ of $x$, we have $|x_i - m_i| \le 0.1$.
The claim follows since otherwise, the total $\ell_1$ regression value of $x$ would be at least $0.099d$ larger than the optimal $\ell_1$ regression value, which would again be larger than $1 + \varepsilon$ times the optimal $\ell_1$ regression value, by adjusting the constant in $s = \Theta(1 / \eps)$.
Thus, from a $(1 + \varepsilon)$-approximate solution $x$ to the $\ell_1$ regression problem, we can solve the $k$-MAJ problem with $k = \Theta(1 / \varepsilon)$, which implies an $\Omega(d / \eps)$ lower bound.

Formally, we have proved the following theorem.
\begin{theorem}\label{thm:l1_bb_lb}
When $s > \Omega(1 / \eps)$, any randomized protocol that succeeds with probability at least $0.99$ for solving $(1 + \varepsilon)$-approximate $\ell_1$ regression requires $\Omega(d / \eps)$ bits of communication in the blackboard model.
\end{theorem}

Now we give a reduction from $k$-XOR to $(1 + \varepsilon)$-approximate $\ell_2$ regression in the blackboard model, and prove an $\Omega(d / \sqrt{\varepsilon})$ lower bound when $s > \Omega(1 / \sqrt{\eps})$.
In the hard case we assume $s = \Theta(1 / \sqrt{\eps})$, and we simply ignore all other servers if $s >  \Theta(1 / \sqrt{\eps})$.
For each server $P_i$, its matrix $A^{(i)}$ is set to be the identity matrix $I_d \in \mathbb{R}^{d \times d}$, and $b^{(i)} \in \{0, 1\}^d$.
For $\ell_2$ regression, the optimal solution can be achieved by taking $x_j$ to be $a_j$, where $a_j$ is the {\em average} of $b^{(1)}_j, b^{(2)}_j, \ldots, b^{(s)}_j$.
Notice that the squared $\ell_2$ regression value associated with the $j$-th coordinate in the optimal solution is upper bounded by $s = \Theta(1 / \sqrt{\eps})$, and thus the total squared $\ell_2$ regression value is upper bounded by $sd = \Theta(d / \sqrt{\eps})$.
Furthermore, if $|x_j - a_j| \ge \Omega(\sqrt{\eps})$, then the squared $\ell_2$ regression value associated with the $j$-th coordinate by using $x_j$ will be at least $\Theta(s \cdot |x_j - a_j|^2) = \Theta(\sqrt{\eps})$ larger than the squared $\ell_2$ regression value associated with the $j$-th coordinate in the optimal solution. 

Now consider a $(1 + \varepsilon)$-approximate solution $x$. 
We claim that for at least $0.99d$ coordinates $x_i$ of $x$, we have $|x_i - a_i| \le O(\sqrt{\eps})$.
The claim follows since otherwise, the total squared $\ell_2$ regression value of $x$ would be at least $\Omega(d\sqrt{\eps})$ larger than the optimal squared $\ell_2$ regression value, which would again be larger than $1 + \varepsilon$ times the optimal squared $\ell_2$ regression value, by adjusting the constant in $s = \Theta(1 / \sqrt{\eps})$.
Notice that for those coordinates with $|x_i - a_i| \le O(\sqrt{\eps})$, we can exactly recover $a_ i$ from $x_i$, since $a_i$ is the average of $b^{(1)}_i, b^{(2)}_i, \ldots, b^{(s)}_i$ and thus $a_i$  is an integer multiple of $1 / s = \Theta(\sqrt{\varepsilon})$.
This also implies we can recover the XOR of $b^{(1)}_i, b^{(2)}_i, \ldots, b^{(s)}_i$.
Thus, from a $(1 + \varepsilon)$-approximate solution $x$ to the $\ell_2$ regression problem, we can solve the $k$-XOR problem with $k = \Theta(1 / \sqrt{\varepsilon})$, which implies an $\Omega(d / \sqrt{\eps})$ lower bound.

Formally, we have proved the following theorem.
\begin{theorem}\label{thm:l2_bb_lb}
When $s > \Omega(1 / \sqrt{\eps})$, any randomized protocol that succeeds with probability at least $0.99$ for solving $(1 + \varepsilon)$-approximate $\ell_2$ regression requires $\Omega(d / \sqrt{\eps})$ bits of communication in the blackboard model.
\end{theorem}

\section{Communication Protocols for $\ell_2$ Regression}\label{sec:l2regression}
In this section, we design distributed protocols for solving the $\ell_2$ regression problem.
\subsection{A Deterministic Protocol} 
In this section, we design a simple deterministic protocol for $\ell_2$ regression in the distributed setting with communication complexity $\widetilde{O}(sd^2L)$ in the coordinator model. 

According to the normal equations, the optimal solution to the $\ell_2$ regression problem $\min_x \|Ax-b\|_2$ can be attained by setting $x^* =  (A^TA)^{\dagger} A^T b$.
In Figure \ref{alg:l2alg1}, we show how to calculate $A^TA$ and $A^Tb$ in the distributed model.

\begin{figure}[!htb]
\begin{framed}
\begin{enumerate}
\item Each server $P_i$ calculates $(A^{(i)})^TA^{(i)}$ and $(A^{(i)})^T b^{(i)}$, and then sends them to the coordinator.
\item The coordinator calculates $A^TA = \sum_{i=1}^s (A^{(i)})^T A^{(i)}$ and $A^T b = \sum_{i=1}^s (A^{(i)})^T b^{(i)}$, and then calculates $x =  (A^TA)^{\dagger} \cdot A^T b$.
\end{enumerate}
\end{framed}
\caption{Protocol for $\ell_2$ regression in the coordinator model}
\label{alg:l2alg1}
\end{figure}

Notice that the bit complexity of entries in $A^TA$ and $A^T b$ is $O(L + \log n)$ since the bit complexity of entries in $A$ and $b$ is $L$,  which implies the communication complexity of the protocol in Figure \ref{alg:l2alg1} is $O(sd^2(L + \log n))$, in both the coordinator model and the blackboard model. 
\begin{theorem}\label{thm:l2alg1}
The protocol in Figure \ref{alg:l2alg1} is a deterministic protocol which exactly solves $\ell_2$ regression, and the communication complexity is $O(sd^2(L + \log n))$, in both the coordinator model and the blackboard model. 
\end{theorem}

\subsection{A Protocol in the Blackboard Model} \label{sec:l2_blackboard}
In this section, we design a recursive protocol for obtaining constant approximations to leverage scores in the distributed setting, which is described in Figure \ref{alg:leverage}.
We then show how to  solve $\ell_2$ regression by using this protocol. 
\begin{figure}[!htb]
\begin{framed}
Input: An $n \times d$ matrix $
A = \left[ 
\begin{array}{c}
A^{(1)} \\
A^{(2)} \\
\vdots \\
A^{(s)} 
\end{array} 
\right]
$, where $A^{(i)}$ is stored on server $P_i$.

Output: An $O(d \log d) \times d$ matrix $\tilde{A}$ such that $$\Omega(1) \|Ax\|_2 \le \|\tilde{A} x\|_2 \le O(1)\|Ax\|_2$$ for all $x \in \mathbb{R}^d$, which is stored on the coordinator and all servers.
\begin{enumerate}
\item \label{step:leverage_base} If $n \le O(d \log d)$, then each server $P_i$ sends $A^{(i)}$ to the coordinator, and then the coordinator sends $A$ to each server, and returns. 
\item \label{step:leverage_sample} Each server $P_i$ locally uniformly samples half of the rows from $A^{(i)}$ to form $\left(A^{(i)}\right)'$. Then each server takes $\left(A^{(i)}\right)'$ as input and invokes the protocol recursively to compute $\tilde{A'}$ for $A' = \left[ 
\begin{array}{c}
\left(A^{(1)}\right)' \\
\left(A^{(2)}\right)' \\
\vdots \\
\left(A^{(s)}\right)'
\end{array} 
\right]$ such that
$\Omega(1) \|A'x\|_2 \le \|\tilde{A'} x\|_2 \le O(1)\|A'x\|_2$ for all $x \in \mathbb{R}^d$.
\item Using $\tilde{A'}$, each server $P_i$ calculates generalized leverage scores (Definition \ref{def:leverage_score}) of all rows in $A^{(i)}$ with respect to $\tilde{A'}$.
\item \label{step:leverage_psample} The coordinator and all servers obtain $\tilde{A}$ by sampling and rescaling $O(d \log d)$ rows using Theorem \ref{thm:l2chernoff}, by setting $p_i \ge C \tilde{\tau}_i \log d $, where
\[
\tilde{\tau}_i = \begin{cases}
\tau_i^{\tilde{A}}(A) & \text{if row $A^i$ is sampled in Step \ref{step:leverage_sample}},\\
\frac{1}{1 + \frac{1}{\tau_i^{\tilde{A}}(A)}} & \text{otherwise}.
\end{cases}
\]
\end{enumerate}
\end{framed}
\caption{Protocol for approximating leverage scores in the distributed setting}
\label{alg:leverage}
\end{figure}

The protocol described in Figure \ref{alg:leverage} is basically Algorithm 2 in \cite{cohen2015uniform} for approximating leverage scores, implemented in the distributed setting. 
Using Lemma 8 in \cite{cohen2015uniform}, the protocol returns an $O(d \log d) \times d$ matrix $\tilde{A}$ such that 
$$\Omega(1) \|Ax\|_2 \le \|\tilde{A} x\|_2 \le O(1)\|Ax\|_2$$ for all $x \in \mathbb{R}^d$, with constant probability. 
Each server $P_i$ can then obtain constant approximations to leverage scores of all rows in $A^{(i)}$ by calculating $\tau_i^{\tilde{A}}(A)$.

Now we analyze the communication complexity of the protocol.
Notice that this recursive algorithm has $O(\log(n / d))$ levels of recursion. 
Step \ref{step:leverage_base} has communication complexity $O(sd^2  L  \log d)$ in the coordinator model and $O(d^2 L \log d )$ in the blackboard model, and will be executed at most once during the whole protocol. 
At Step \ref{step:leverage_psample}, we can assume each $p_i^{-1/2}$ is a power of two between $1$ and $\poly(n)$, since we can discard all rows whose $p_i < 1 / \poly(n)$ and increase each $p_i$ by a constant factor.
In order to implement the sampling process in Theorem \ref{thm:l2chernoff} in the distributed setting, each server $P_i$ sends the summation of $p_i$ for all rows in $A^{(i)}$ to the coordinator.
After receiving all these summations, the coordinator decides the number of rows to be sampled from each $A^{(i)}$ and sends these numbers back to each server. 
The communication complexity of this step is at most $O(s \log n)$.
Each server $P_i$ samples and rescales the rows accordingly, and then sends these sampled rows to the coordinator, and the coordinator sends all sampled rows back to all servers. 
Notice that the bit complexity of all entries in the sampled rows is at most $O(L + \log n)$ since $p_i^{-1/2}$ is an integer between $1$ and $\poly(n)$.
Thus, in the blackboard model, the total communication complexity at each recursive level of the protocol is upper bounded by $O(d^2 \log d(L + \log n) + s \log n)$, which implies the communication complexity of the whole protocol is at most $O((d^2 \log d(L + \log n) + s \log n)\cdot \log(n / d))$ in the blackboard model. 
A similar analysis shows that the communication complexity of the whole protocol is $\widetilde{O}(sd^2L)$ in the coordinator model.

\begin{lemma}\label{lem:leverage}
The protocol described in Figure \ref{alg:leverage} is a randomized protocol with communication complexity $\widetilde{O}(sd^2L)$ in the coordinator model and $\widetilde{O}(s + d^2L)$ in the blackboard model, such that with constant probability, upon termination of the protocol, each server $P_i$ has constant approximations to leverage scores of all rows in $A^{(i)}$. 
\end{lemma}

Our protocol for solving the $\ell_2$ regression problem in the blackboard model is described in Figure \ref{alg:l2alg2}.
It is guaranteed that the vector $x$ calculated at Step \ref{step:terminate} is a $(1 + \varepsilon)$-approximate solution, with constant probability. 
A na\"ive approach for obtaining $(1 + \varepsilon)$-approximate solution to the $\ell_2$ regression problem will be using Theorem \ref{thm:l2chernoff} to obtain $SA$ and $Sb$ such that with constant probability, for all $x \in \R^d$,
\[
(1 - \varepsilon)\|Ax - b\|_2 \le \|S(Ax - b)\|_2 \le (1 + \varepsilon) \|Ax -  b\|_2.
\]
By doing so, the number of sampled rows should be $O(d \log d / \varepsilon^2)$ according to Theorem \ref{thm:l2chernoff}. 
However, as shown in Theorem 36 of \cite{clarkson2013low}, in order to obtain a $(1 + \varepsilon)$-approximate solution to the $\ell_2$ regression problem (instead of obtaining a  $(1 + \varepsilon)$ subspace embedding), it suffices to sample $O(d \log d + d / \varepsilon)$ rows from $[A ~ b]$.

Now we analyze the communication complexity of the protocol in Figure \ref{alg:l2alg2} in the blackboard model. 
By Lemma \ref{lem:leverage}, the communication complexity of Step \ref{step:approx_leverage} is upper bounded by $\widetilde{O}(d^2L + s)$.
Similar to Step \ref{step:leverage_psample} of the protocol described in Figure \ref{alg:leverage}, the sampling process in Step \ref{step:do_sample} can be implemented with communication complexity $\widetilde{O}(s + d^2L / \varepsilon)$.
Thus, the total communication complexity is $\widetilde{O}(s + d^2L / \varepsilon)$ in the blackboard model.
\begin{figure}[!htb]
\begin{framed}
\begin{enumerate}
\item \label{step:approx_leverage} Use the protocol in Figure \ref{alg:leverage} to approximate leverage scores of $A$. 
\item \label{step:do_sample} The coordinator obtains $SA$ and $Sb$ by sampling and rescaling $O(d / \varepsilon + d \log d)$ rows of $[A~b]$, using the sampling process in Theorem \ref{thm:l2chernoff}.
\item \label{step:terminate} The coordinator calculates $x =  \min_x \|SAx-Sb\|_2$.
\end{enumerate}
\end{framed}
\caption{Protocol for $\ell_2$ regression in the blackboard model}
\label{alg:l2alg2}
\end{figure}

\begin{theorem}\label{thm:l2alg2}
The protocol described in Figure \ref{alg:l2alg2} is a randomized protocol which returns a $(1 + \varepsilon)$-approximate solution to $\ell_2$ regression with constant probability, and the communication complexity is $\widetilde{O}(s + d^2L / \varepsilon)$ in the blackboard model . 
\end{theorem}

\section{Communication Protocols for $\ell_1$ Regression}\label{sec:l1regression}
In this section, we design distributed protocols for solving the $\ell_1$ regression problem.
\subsection{A Simple Protocol} \label{sec:l1alg1}
In this section, we design a simple protocol for obtaining a $(1 + \varepsilon)$-approximate solution to the $\ell_1$ regression problem in the distributed setting.
The protocol is described in Figure \ref{alg:l1alg1}.
\begin{figure}[!htb]
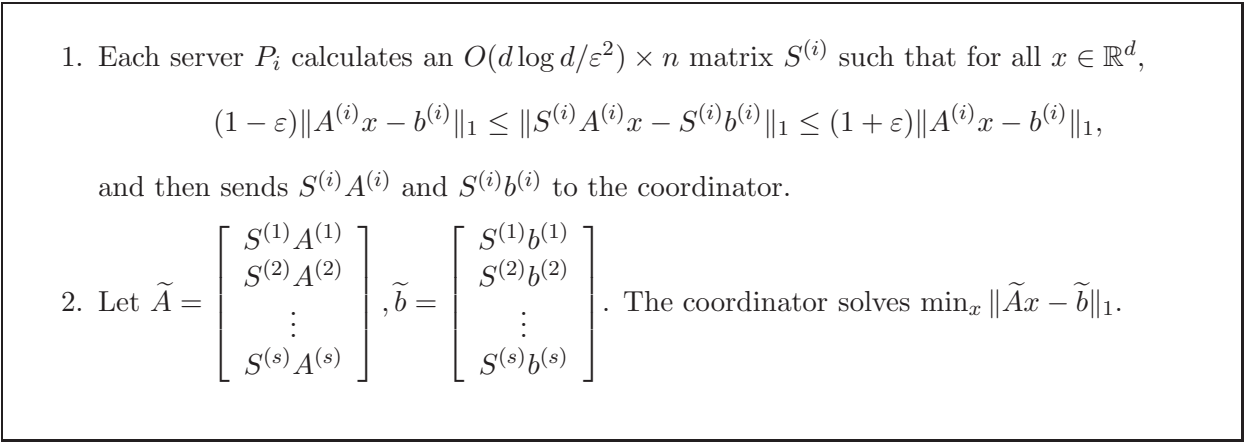

\begin{framed}
\begin{enumerate}
\item \label{step:l1alg1_lewis1} Each server $P_i$ calculates an $O(d \log d / \varepsilon^2) \times n$ matrix $S^{(i)}$ such that for all $x \in \mathbb{R}^d$,
\[
(1 - \varepsilon)\|A^{(i)}x - b^{(i)}\|_1 \le \|S^{(i)}A^{(i)}x - S^{(i)} b^{(i)}\|_1 \le (1 + \varepsilon)\|A^{(i)}x - b^{(i)}\|_1,
\]
and then sends $S^{(i)} A^{(i)}$ and $S^{(i)} b^{(i)}$ to the coordinator.
\item \label{step:l1alg1_sol}
Let 
$
\tilde{A} = \left[ \begin{array}{c} S^{(1)} A^{(1)} \\ S^{(2)} A^{(2)} \\ \vdots \\ S^{(s)} A^{(s)}\end{array} \right], 
\tilde{b} = \left[ 
\begin{array}{c}
S^{(1)} b^{(1)} \\
S^{(2)} b^{(2)} \\
\vdots \\
S^{(s)} b^{(s)}
\end{array}
\right]
$.
The coordinator solves
$
\min_{x} \|\tilde{A} x - \tilde{b}\|_1
$.
\end{enumerate}
\end{framed}
\caption{Protocol for $\ell_1$ regression in the coordinator model}
\label{alg:l1alg1}
\end{figure}

To implement Step \ref{step:l1alg1_lewis1}, each server $P_i$ calculates the $\ell_1$ Lewis weights of $[A^{(i)}~b^{(i)}]$ and uses Theorem \ref{thm:l1chernoff} to randomly generate a matrix $S^{(i)}$.
$P_i$ then checks whether for all $x \in \mathbb{R}^d$
\begin{equation}\label{equ:suc_se}
(1 - \varepsilon)\|A^{(i)}x - b^{(i)}\|_1 \le \|S^{(i)}A^{(i)}x - S^{(i)} b^{(i)}\|_1 \le (1 + \varepsilon)\|A^{(i)}x - b^{(i)}\|_1.
\end{equation}
If not, $P_i$ randomly generates another $S^{(i)}$ until (\ref{equ:suc_se}) is satisfied for all $x \in \mathbb{R}^d$.
Since \eqref{equ:suc_se} is satisfied with constant probability, the number of independent trials is $O(1)$ in expectation. 
Furthermore, each server can locally check whether \eqref{equ:suc_se} holds or not by e.g., verifying on an $\varepsilon$-net. 
Notice that the use of randomness is not critical here, since each server $P_i$ can locally enumerate all possible $S^{(i)}$ up to a specific precision, instead of using Theorem \ref{thm:l1chernoff} to randomly generate a matrix $S^{(i)}$.
Each server $P_i$ will eventually find a matrix $S^{(i)}$ which satisfies \eqref{equ:suc_se}, whose existence is guaranteed by Theorem \ref{thm:l1chernoff}.

Now we prove the correctness of the protocol. 
Notice that it is guaranteed that for any $x \in \mathbb{R}^d$ and any $i \in [s]$,
\[
(1 - \varepsilon) \|A^{(i)}x - b^{(i)}\|_1 \le \|S^{(i)}A^{(i)}x - S^{(i)} b^{(i)}\|_1 \le (1 + \varepsilon)\|A^{(i)}x - b^{(i)}\|_1.
\]
It implies that 
\[
 \|\tilde{A} x - \tilde{b}\|_1 = \sum_{i=1}^s \|S^{(i)}A^{(i)}x - S^{(i)} b^{(i)}\|_1 \ge \sum_{i=1}^s (1 - \varepsilon) \|A^{(i)}x - b^{(i)}\|_1  =  (1 - \varepsilon) \|Ax - b\|_1
\]
and
\[
\|\tilde{A} x - \tilde{b}\|_1  = \sum_{i=1}^s \|S^{(i)}A^{(i)}x - S^{(i)} b^{(i)}\|_1 \le (1 + \varepsilon )\sum_{i=1}^s \|A^{(i)}x - b^{(i)}\|_1  = (1 + \varepsilon )\|Ax - b\|_1.
\]
Thus, the vector $x$ calculated at Step \ref{step:l1alg1_sol} is a $(1 + \varepsilon)$-approximate solution to the $\ell_1$ regression problem.

Finally, we analyze the communication complexity of the protocol. 
Similar to the analysis in Section \ref{sec:l2_blackboard}, we may assume all $1 / p_i$ in the sampling process of Theorem \ref{thm:l1chernoff} are integers between $1$ and $\poly(n)$.
Thus, the bit complexity of all entries in $S^{(i)} A^{(i)}$ and $S^{(i)} b^{(i)}$ is at most $O(L + \log n)$, which implies the communication complexity of Step \ref{step:l1alg1_lewis1} is $O(sd^2 \log d \cdot  (L + \log n) / \varepsilon^2)$ in both the coordinator model and the blackboard model. 

\begin{theorem}\label{thm:l1alg1}
The protocol described in Figure \ref{alg:l1alg1} is a deterministic protocol which returns a $(1 + \varepsilon)$-approximate solution to the $\ell_1$ regression problem, and the communication complexity is $\widetilde{O}(sd^2L / \varepsilon^2)$ in both the coordinator model and the blackboard model. 
\end{theorem}
\subsection{A Protocol Based on $\ell_1$ Lewis Weights Sampling} \label{sec:l1alg2}
In this section, we first design a protocol for obtaining constant approximations to $\ell_1$ Lewis weights in the distributed setting, which is described in Figure \ref{alg:lewis},
and then solves the $\ell_1$ regression problem based on this protocol. 

\begin{figure}[!htb]
\begin{framed}
\begin{enumerate}
\item Each server $P_i$ initializes $w_i = 1$ for all rows in $A^{(i)}$.
\item For $t = 1, 2, \ldots, T$
\begin{enumerate}
\item Each server $P_i$ obtains constant approximations to leverage scores of $W^{-1 / 2}A$, for all rows stored on $P_i$ using the protocol in Lemma \ref{lem:leverage}, where $W$ is the diagonal matrix formed by putting the elements of $w$ on the diagonal.
\item Set $w_i = (w_i \tau_i(W^{-1/2}A))^{1 / 2}$.
\end{enumerate}
\end{enumerate}
\end{framed}
\caption{Protocol for approximating $\ell_1$ Lewis weights}
\label{alg:lewis}
\end{figure}

The protocol described in Figure \ref{alg:lewis} is basically the algorithm in Section 3 of \cite{cohen2015p} for approximating $\ell_1$ Lewis weights, implemented in the distributed setting.
Using the same analysis, by setting $T = O(\log \log n)$,  we can show $w_i$ are constant approximations to the $\ell_1$ Lewis weights of $A$.
Now we show that we can assume all $w_i^{-1/2}$ are integers between $1$ and $2^{\widetilde{O}(L)}$.

Without loss of generality we assume each row of $A$ contains at least one non-zero entry. 
Since our goal here is to calculate constant approximations to the $\ell_1$ Lewis weights, using the analysis in Section 3 in \cite{cohen2015p}, we only need constant approximations to $w_i$ during the execution of the algorithm. 
Furthermore, since leverage scores $\tau_i(W^{-1/2}A)$ are at most $1$ (see, e.g., Section 2.4 in \cite{woodruff2014sketching}), we can prove by induction that $w_i \le 1$ during the execution of the algorithm. 
Thus, we may assume $w_i^{-1/2} \ge 1$ and $w_i^{-1 / 2}$ are integers.

Now we show that $w_i \ge 2^{-\widetilde{O}(L)}$. We prove this claim by induction. At the beginning of the algorithm, $w_i = 1$ for all $1 \le i \le n$.
Assume $w_i \ge 2^{-\widetilde{O}(L)}$ by the induction hypothesis, we know all entries in $W^{-1/2}A$ are integers between $1$ and $2^{\widetilde{O}(L)}$.
Using Lemma 2 in \cite{cohen2015uniform}, we know
$$
\tau_i(W^{-1/2}A) = \min_{(W^{-1/2}A)^T x = (W^{-1/2} A)^i} \|x\|_2^2.
$$
By the Cauchy-Schwarz inequality, in order that $(W^{-1/2}A)^T x = (W^{-1/2} A)^i$, we must have
\[
\|x\|_2 \ge 2^{-\widetilde{O}(L)} / \sqrt{n}
\]
since otherwise all entries in $(W^{-1/2}A)^T x$ are less than $1$, which violates the assumption that all entries in $W^{-1/2}A$ are integers and each row of $A$ contains at least one non-zero entry.
Furthermore, the number of iterations is at most $O(\log \log n)$, which implies $w_i \ge 2^{-\widetilde{O}(L)}$ by induction.

Thus, all entries in $W^{-1/2}A$ have bit complexity $\wt{O}(L)$ during the execution of the algorithm.
Using Lemma \ref{lem:leverage}, the communication complexity is $\widetilde{O}(s + d^2L)$ in the blackboard model, and $\widetilde{O}(sd^2L)$ in the coordinator model.

\begin{lemma}\label{lem:lewis}
The protocol described in Figure \ref{alg:lewis} is a randomized protocol with communication complexity $\widetilde{O}(s + d^2L)$ in the blackboard model and $\widetilde{O}(sd^2L)$ in the blackboard model, such that with constant probability, upon termination of the protocol, each server $P_i$ has constant approximations to the $\ell_1$ Lewis weights of all rows in $A^{(i)}$. 
\end{lemma}

Our protocol for solving the $\ell_1$ regression problem in the blackboard model is described in Figure \ref{alg:l1alg2}.
By Theorem \ref{thm:l1chernoff}, it is guaranteed that the vector $x$ calculated at Step \ref{step:terminate} is a $(1 + \varepsilon)$-approximate solution, with constant probability. 

Now we analyze the communication complexity of the protocol in Figure \ref{alg:l1alg2}.
By Lemma \ref{lem:lewis}, the communication complexity of Step \ref{step:approx_leverage} is upper bounded by $\widetilde{O}(d^2L + s)$ in the blackboard model and $\widetilde{O}(sd^2)$ in the coordinator model.
Similar to Step \ref{step:leverage_psample} of the protocol described in Figure \ref{alg:leverage}, the sampling process in Step \ref{step:do_sample} can be implemented with communication complexity $\widetilde{O}(s + d^2L / \varepsilon^2)$ in both models.
Thus, the total communication complexity is $\widetilde{O}(s + d^2L / \varepsilon^2)$ in the blackboard model.
In the coordinator model, the total communication complexity is $\widetilde{O}(sd^2L + d^2 L / \varepsilon^2)$.
\begin{figure}[!htb]
\begin{framed}
\begin{enumerate}
\item \label{step:approx_leverage} Use the protocol in Figure \ref{alg:lewis} to approximate $\ell_1$ Lewis weights of $[A~b]$. 
\item \label{step:do_sample} The coordinator obtains $SA$ and $Sb$ by sampling and rescaling $O(d \log d / \varepsilon^2)$ rows of $[A~b]$, using the sampling process in Theorem \ref{thm:l1chernoff}.
\item \label{step:terminate} The coordinator calculates $x =  \min_x \|SAx-Sb\|_1$. 
\end{enumerate}
\end{framed}
\caption{Protocol for $\ell_1$ regression in the blackboard model}
\label{alg:l1alg2}
\end{figure}

\begin{theorem}\label{thm:l1alg2}
The protocol described in Figure \ref{alg:l1alg2} is a randomized protocol which returns a $(1 + \varepsilon)$-approximate solution to the $\ell_1$ regression problem with constant probability, and the communication complexity is $\widetilde{O}(s + d^2L / \varepsilon)$ in the blackboard model and $\widetilde{O}(sd^2L + d^2 L / \varepsilon^2)$ in the coordinator model. 
\end{theorem}
\subsection{A Protocol Based on Accelerated Gradient Descent}\label{sec:gd}
In this section, we present a protocol for the $\ell_1$ regression problem in the coordinator model, whose communication complexity is $\widetilde{O}(sd^3L / \varepsilon)$.

We need the following definition in \cite{durfee2017ell_1}.
\begin{definition}[\cite{durfee2017ell_1}]\label{irb}
Suppose $n \ge d$. 
A matrix $A \in \mathbb{R}^{n \times d}$ is {\em approximately isotropic row-bounded} if the following hold:
\begin{enumerate}
\item $A^T A \approx_{O(1)} I$;
\item For all rows of $A$, $\|A^i\|_2^2 \le O(d / n)$.
\end{enumerate}
\end{definition}

Before presenting the protocol, we first present a preconditioning procedure in Figure \ref{alg:precondition}, which will later be used in the protocol for $\ell_1$ regression. 
\begin{figure}[!htb]
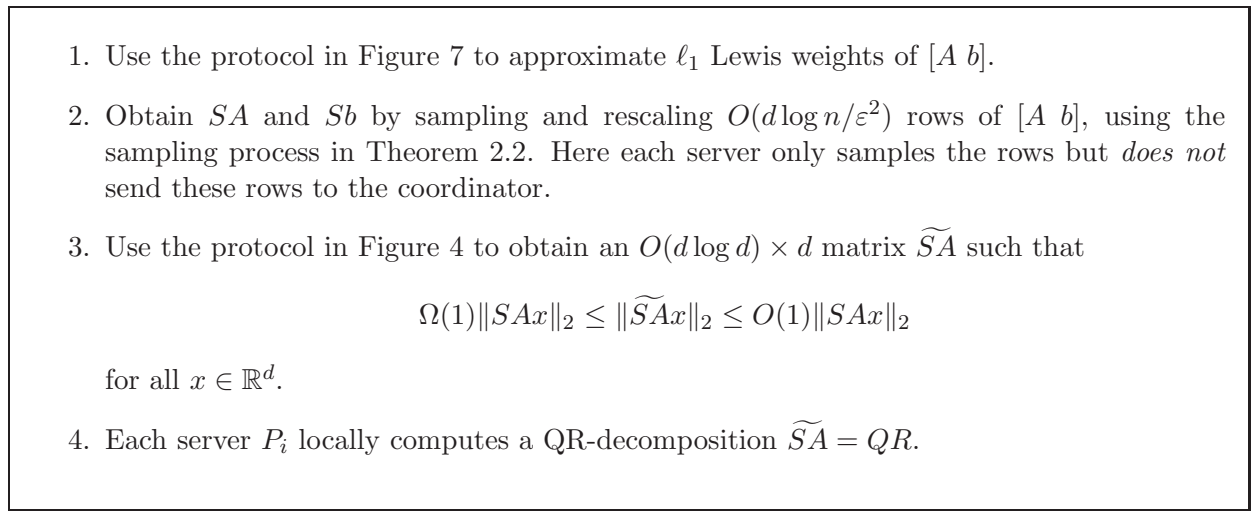

\begin{framed}
\begin{enumerate}
\item Use the protocol in Figure \ref{alg:lewis} to approximate $\ell_1$ Lewis weights of $[A~b]$. 
\item Obtain $SA$ and $Sb$ by sampling and rescaling $O(d \log n / \varepsilon^2)$ rows of $[A~b]$, using the sampling process in Theorem \ref{thm:l1chernoff}. Here each server only samples the rows but {\em does not} send these rows to the coordinator. 
\item Use the protocol in Figure \ref{alg:leverage} to obtain an $O(d \log d) \times d$ matrix $\widetilde{SA}$ such that $$\Omega(1) \|SAx\|_2 \le \|\widetilde{SA}x\|_2 \le O(1) \|SAx\|_2$$ for all $x \in \mathbb{R}^d$.
\item Each server $P_i$ locally computes a QR-decomposition $\widetilde{SA} = QR$.
\end{enumerate}
\end{framed}
\caption{A preconditioning procedure for $\ell_1$ regression} \label{alg:precondition}
\end{figure}
The communication complexity of the this protocol in the coordinator model is $\widetilde{O}(sd^2L)$ by Lemma \ref{lem:leverage} and Lemma \ref{lem:lewis}.
Similar to the analysis in Section \ref{sec:l2_blackboard}, we can also assume the bit complexity of all entries in $SA$ and $Sb$ is $O(L + \log n)$.
Furthermore, by Theorem \ref{thm:l1chernoff}, a $(1 + \varepsilon)$-approximate solution to the $\ell_1$ regression problem
\begin{equation}\label{equ:preconditioned_l1}
\min_x \|SAR^{-1}x - Sb\|_1
\end{equation}
is a $(1 + O(\varepsilon))$-approximate solution to the original $\ell_1$ regression problem
$$
\min_x \|Ax - b\|_1.
$$
Thus, we will focus on the $\ell_1$ regression problem in \eqref{equ:preconditioned_l1} in the remaining part of this section

Now we show $SAR^{-1}$ is approximately isotropic row-bounded as in Definition \ref{irb}.
We only need to show $(SAR^{-1})^T SAR^{-1} \approx_{O(1)} I$ and all rows of $SAR^{-1}$ satisfy $\|(SAR^{-1})^i\|_2^2 = O(d / N)$ where $N = O(d \log n / \varepsilon^2)$ is the number of rows of $SAR^{-1}$.

To show $(SAR^{-1})^T SAR^{-1} \approx_{O(1)} I$, it is equivalent to show that for all $x \in \mathbb{R}^d$,
\begin{equation}\label{equ:irb}
\Omega(1) \|x\|_2^2 \le \|SAR^{-1}x\|_2^2 \le O(1) \|x\|_2^2.
\end{equation}
Notice that $\widetilde{SA}R^{-1} $ is an orthogonal matrix $Q$, which implies for all $x \in \mathbb{R}^d$
\begin{equation}\label{equ:qr}
\|\widetilde{SA}R^{-1} x\|_2 = \|x\|_2.
\end{equation}
Combining (\ref{equ:qr}) and the fact that 
$$
\Omega(1) \|SAx\|_2 \le \|\widetilde{SA}x\|_2 \le O(1) \|SAx\|_2
$$ for all $x \in \mathbb{R}^d$, we can prove \eqref{equ:irb}, which implies $(SAR^{-1})^T SAR^{-1} \approx_{O(1)} I$.

To show $\|(SA)^i R^{-1}\|_2^2 = O(d / N)$, we use Lemma 29 in \cite{durfee2017ell_1}, which states that with constant probability, the leverage scores of $SA$ satisfy $\tau_i(SA) = O(d / N)$ for all $i$.
Since leverage scores are invariant under change of basis (see, e.g., Section 2.4 in \cite{woodruff2014sketching}), we have for all $i$,
$$
 (SA)^i R^{-1} (SAR^{-1})^T SAR^{-1} ((SA)^i R^{-1})^T  = O(d / N).
$$
Since $(SAR^{-1})^T SAR^{-1} \approx_{O(1)} I$, we have
$$
\|(SA)^i R^{-1} \|_2^2 = O(d / N).
$$
Thus, $SAR^{-1}$ is approximately isotropic row-bounded.

Now we describe our protocol for $\ell_1$ regression in Figure \ref{alg:gd}. 
Our protocol first uses the preconditioning procedure in Figure \ref{alg:precondition} and then uses Nesterov's accelerated gradient descent \cite{nesterov1983method} to solve the $\ell_1$ regression problem 
$$\min_x \|SAR^{-1}x - Sb\|_1.$$
Furthermore, we invoke a smoothing reduction \textsf{JointAdaptRegSmooth} in \cite{allen2016optimal} to obtain better dependence on $\varepsilon$.
\begin{figure}[!htb]
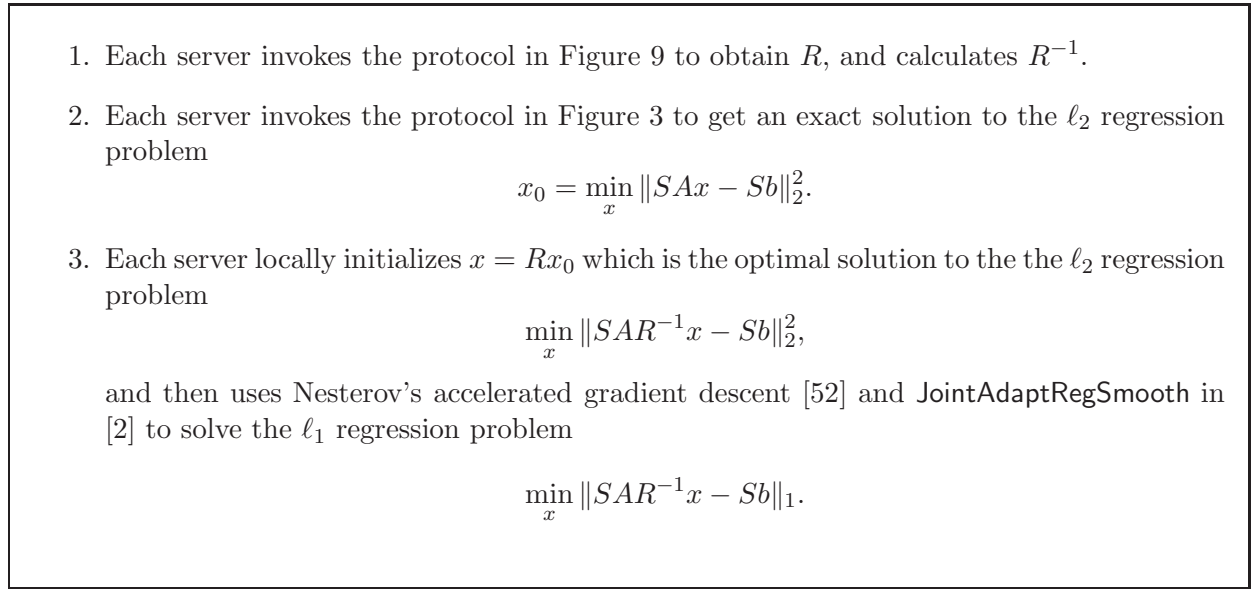

\begin{framed}
\begin{enumerate}
\item Each server invokes the protocol in Figure \ref{alg:precondition} to obtain $R$, and calculates $R^{-1}$.
\item Each server invokes the protocol in Figure \ref{alg:l2alg1} to get an exact solution to the $\ell_2$ regression problem
\[
x_0 = \min_{x} \|SAx - Sb\|_2^2.
\]
\item Each server locally initializes $x = Rx_0$ which is the optimal solution to the the $\ell_2$ regression problem $$ \min_{x} \|SAR^{-1}x - Sb\|_2^2,$$ and then uses Nesterov's accelerated gradient descent \cite{nesterov1983method} and \textsf{JointAdaptRegSmooth} in \cite{allen2016optimal} to solve the $\ell_1$ regression problem
\[
\min_{x} \|SAR^{-1} x - Sb\|_1.
\]
\end{enumerate}
\end{framed}
\caption{Accelerated gradient descent for $\ell_1$ regression} \label{alg:gd}
\end{figure}

In order to implement Nesterov's accelerated gradient descent in the distributed setting, each server $P_i$ maintains the current solution $x$.
In each round, servers communicate to calculate the current gradient vector. 
Once all servers receive the gradient vector, they can update their current solution $x$ locally and proceed to the next round.
Analysis in \cite{allen2016optimal} (Example C.3) shows that 
when \textsf{JointAdaptRegSmooth} is applied to Nesterov's accelerated gradient descent,
after $O(G \sqrt{\Theta} / \delta )$ full gradient calculations, the algorithm will output a vector $x$ such that 
$$
\|Ax-b\|_1 \le \min_x\| Ax - b\|_1 + O(\delta),
$$
where we assume $\|SAR^{-1}x - Sb\|_1$ is $G$-Lipschitz continuous and the initial solution $x$ satisfies $\|x - x^{*}\|_2^2 \le \Theta$. 
Since $SAR^{-1}$ is approximately isotropic row-bounded and the initial vector $x$ is the optimal solution to the $\ell_2$ regression problem $ \min_{x} \|SAR^{-1}x - Sb\|_2^2$, Lemma 19 in \cite{durfee2017ell_1} shows that 
$\|x - x^*\|_2 \le \sqrt{d / n} \|Ax^* - b\|_1$.
Furthermore, Lemma 15 in \cite{durfee2017ell_1} shows that 
$
G \le \sqrt{nd}.
$
By setting $\delta = \varepsilon\|Ax^* - b\|_1$, we can calculate a $(1 + \varepsilon)$-approximate solution to the $\ell_1$ regression problem using $O(d / \varepsilon)$ full gradient calculations. 

Both Nesterov's accelerated gradient descent \cite{nesterov1983method} and \textsf{JointAdaptRegSmooth} in \cite{allen2016optimal} require an estimation (up to a constant factor) of $\|Ax^* - b\|_1$, which be can be obtained by using the algorithm in Section \ref{sec:l1alg1} to obtain an $O(1)$-approximate solution $\hat{x}$ and then calculating $\|A\hat{x} - b\|_1$.

It remains to design an protocol to calculate the gradient vector of the smoothed objective function for $\|SAR^{-1}x - Sb\|_1$, in the distributed setting.
We show this can be done with communication complexity $\wt{O}(sd^2L)$.
By using \textsf{JointAdaptRegSmooth} in \cite{allen2016optimal}, the new objective function will be 
\begin{equation}\label{equ:gradient}
\sum_{i=1}^n f_{\lambda_t} (\langle (SA)^i R^{-1}, x \rangle - (Sb)^i) + \frac{\sigma_t}{2} \|x - x_0\|_2^2
\end{equation}
where 
$$
f_{\lambda_t}(x) = 
\begin{cases}
\frac{x^2}{2 \lambda_t} & \text{if } |x| \le \lambda_t\\
|x| - \frac{\lambda_t}{2} & \text{otherwise}
\end{cases}, 
$$
for some $\lambda_t, \sigma_t \in \mathbb{R}$ and $x_0 \in \mathbb{R}^{d}$ known to each server.

Each server can locally calculate the gradient vector of the $\frac{\sigma_t}{2} \|x - x_0\|_2^2$ term, since $\sigma_t$ and $x_0$ is known to each server.
In the remaining part of this section, we focus on designing an algorithm for calculating the gradient vector of the first term in \eqref{equ:gradient}.

For the first term in (\ref{equ:gradient}), we have
\begin{equation}\label{equ:each_gradient}
\nabla  f_{\lambda_t} (\langle (SA)^i R^{-1}, x \rangle - (Sb)^i) = \begin{cases}
(\langle (SA)^i R^{-1}, x \rangle - (Sb)^i) / \lambda_t \cdot (SA)^{i}R^{-1} & \text{if } |\langle (SA)^i R^{-1}, x \rangle - (Sb)^i| \le \lambda_t\\
\sign(\langle (SA)^i R^{-1}, x \rangle - (Sb)^i) \cdot (SA)^{i} R^{-1} & \text{otherwise}
\end{cases}.
\end{equation}
Notice that we cannot directly let each server $P_i$ calculate the gradient vectors using \eqref{equ:each_gradient}, send the gradient vectors to the coordinator and calculate the summation, since the bit complexity of $R^{-1}$ can be unbounded.
Instead, we deal the two cases in (\ref{equ:gradient}) by using two different approaches. 

When $|\langle (SA)^i R^{-1}, x \rangle - (Sb)^i| > \lambda_t$, notice that although the bit complexity of $\sign(\langle (SA)^i R^{-1}, x \rangle - (Sb)^i) \cdot (SA)^{i} R^{-1}$ can be unbounded, 
all entries in the vector $\sign(\langle (SA)^i R^{-1}, x \rangle - (Sb)^i) \cdot (SA)^{i}$ have bit complexity at most $\widetilde{O}(L)$ and
 $R^{-1}$ is a matrix known to each server. 
Thus, for each server $P$, it sends 
\[
\sum \sign(\langle (SA)^i R^{-1}, x \rangle - (Sb)^i) \cdot (SA)^{i} 
\]
to the coordinator, 
for each row $(SA)^{i}$ which is stored on $P$ and satisfies 
$|\langle (SA)^i R^{-1}, x \rangle - (Sb)^i| > \lambda_t$.
After receiving from each server, the coordinator calculates
\[
\sum \sign(\langle (SA)^i R^{-1}, x \rangle - (Sb)^i) \cdot (SA)^{i} 
\]
for all rows $(SA)^{i}$ that satisfy $|\langle (SA)^i R^{-1}, x \rangle - (Sb)^i| > \lambda_t$,
and sends it to each server.
All servers can then recover the gradient vector.
The total communication for this case is at most $\widetilde{O}(sdL)$.

When $|\langle (SA)^i R^{-1}, x \rangle - (Sb)^i| \le \lambda_t$, 
\begin{align*}
& (\langle (SA)^i R^{-1}, x \rangle - (Sb)^i) / \lambda_t \cdot (SA)^{i} R^{-1} \\
= &  x^T \left(R^{-1}\right)^T (((SA)^i)^T (SA)^i) R^{-1} / \lambda_t - (Sb)^{i} (SA)^i R^{-1} / \lambda_t.
\end{align*}
Thus, for each server $P$, it sends
$$
\sum ((SA)^i)^T (SA)^i
$$
and
$$
\sum (Sb)^{i} (SA)^i
$$
to the coordinator, 
for each row $(SA)^{i}$ which is stored on $P$ and satisfies 
$|\langle (SA)^i R^{-1}, x \rangle - (Sb)^i| \le \lambda_t$.
After receiving from each server, the coordinator calculates
$$
\sum ((SA)^i)^T (SA)^i
$$
and
$$
\sum (Sb)^{i} (SA)^i
$$
for all rows $(SA)^{i}$ that satisfy $|\langle (SA)^i R^{-1}, x \rangle - (Sb)^i| \le \lambda_t$,
and sends it to each server.
All servers can then recover the gradient vector.
The total communication for this case is at most $\widetilde{O}(sd^2L)$.

Thus, the total communication complexity of the protocol in Figure \ref{alg:gd} is $\widetilde{O}(sd^3L / \varepsilon)$.
\begin{theorem}\label{thm:gd}
The protocol described in Figure \ref{alg:gd} is a randomized protocol which returns a $(1 + \varepsilon)$-approximate solution to the $\ell_1$ regression problem with constant probability, and the communication complexity is $\widetilde{O}(sd^3L / \varepsilon)$ in the coordinator model. 
\end{theorem}


\section{Communication Protocols for $\ell_p$ Regression}\label{sec:reg_lp}
In this section, we design distributed protocols for solving the $\ell_p$ regression problem, including $p = \infty$.
\subsection{Communication Protocols for $\ell_{\infty}$ Regression} \label{sec:l_infty}
Any $\ell_{\infty}$ regression instance $\min_x \|Ax-b\|_{\infty}$ can be formulated as the following linear program, 
\begin{align*}
\textnormal{minimize } & v &&  \\
\textnormal{subject to } 
&  \langle A^i, x \rangle - b^i \le v,\\ 
&  \langle A^i, x \rangle - b^i \ge -v,\\
\end{align*}
which has $2n$ constraints and $d + 1$ variables.
Thus, any linear programming protocol implies a protocol for solving the $\ell_{\infty}$ regression problem, with the same communication complexity.
Using the linear program solvers in Section \ref{sec:clarkson} and Section \ref{sec:cog}, we have the following theorem.
\begin{theorem}\label{thm:l_infty}
$\ell_{\infty}$ regression can be solved deterministically and exactly with communication complexity $\widetilde{O}(sd^3L)$ in the coordinator model, and randomly and exactly with communication complexity $\widetilde{O}(\min\{sd + d^4L, sd^3L\})$ in the blackboard model.
\end{theorem}
\subsection{Communication Protocols for $\ell_p$ Regression When $p > 2$}\label{sec:lpregression}
In this section, we introduce an approach that reduces $(1 + \varepsilon)$-approximate $\ell_p$ regression to linear programs with $\widetilde{O}(d / \varepsilon^2)$ variables.
Our main idea is to use the max-stability of exponential random variables \cite{andoni2017high} to embed $\ell_p$ into $\ell_{\infty}$.
Such idea was previously used to construct subspace embeddings for the $\ell_p$ norm~\cite{wz13b}.
However, since our goal here is to solve linear regression instead of providing an embedding for the whole subspace, we can achieve a much better approximation ratio than previous work~\cite{wz13b}.


\begin{theorem}\label{thm:p_red}
For any matrix $A \in \mathbb{R}^{n \times d}$ and constant $p > 2$, let $D^{(1)}, D^{(2)}, \ldots, D^{(R)}$ be $n \times n$ random diagonal matrices, whose diagonal entries are i.i.d. random variables with the same distribution as $E^{-1/p}$, where $E$ is an exponential random variable. 
If $R = O(d \log(d / \varepsilon) / \varepsilon^2)$, then with constant probability, the following holds:
\begin{enumerate}
\item \label{item:dilation}
\[
\sum_{i=1}^R \|D^{(i)} (Ax^* - b)\|_{\infty} \le (1 + \varepsilon) C_p R \|Ax^* - b\|_p;
\]
\item \label{item:contraction}For all $x \in \mathbb{R}^d$,
\[
\sum_{i=1}^R \|D^{(i)} (Ax - b)\|_{\infty} \ge (1 - \varepsilon) C_p R \|Ax - b\|_p.
\]
\end{enumerate}
Here $x^* \in \mathbb{R}^d$ is the optimal solution to the $\ell_p$ regression problem $\min_x \|Ax - b\|_p$ and $C_p$ is a constant which is the expectation of $E^{-1 / p}$ for an exponential random variable $E$.
\end{theorem}
The proof of Theorem \ref{thm:p_red} can be found in Section~\ref{sec:proof_p_red}.

Now we prove with constant probability, the optimal solution to the optimization problem
$$
\hat{x} = \argmin_x \sum_{i=1}^R \|D^{(i)} (Ax - b)\|_{\infty} 
$$
satisfies 
$$
\|A\hat{x} - b\|_p \le (1 + 3\varepsilon) \|Ax^* - b\|_p,
$$
where $x^* \in \mathbb{R}^d$ is the optimal solution to the $\ell_p$ regression problem $\min_x \|Ax - b\|_p$.

Notice that with constant probability, 
\begin{align*}
&\|A\hat{x} - b\|_p \le \frac{\sum_{i=1}^R \|D^{(i)} (A\hat{x} - b)\|_{\infty}}{ (1 - \varepsilon) C_p R} \le  \frac{\sum_{i=1}^R \|D^{(i)} (Ax^* - b)\|_{\infty}}{ (1 - \varepsilon) C_p R} \\
\le& \frac{1 + \varepsilon}{1 - \varepsilon} \|Ax^* - b\|_p   \le (1 + 3\varepsilon) \|Ax^* - b\|_p.
\end{align*}

Thus, we have reduced $(1 + \varepsilon)$-approxiamte $\ell_p$ regression to 
\begin{equation} \label{equ:red_p}
\argmin_x \sum_{i=1}^R \|D^{(i)} (Ax - b)\|_{\infty} .
\end{equation}

The optimzation problem in \eqref{equ:red_p} can be written as a linear program with $R + d = \widetilde{O}(d / \varepsilon^2)$ variables.
For each $i \in [R]$, we use $v_i$ to represent the value of $\|D^{(i)} (Ax - b)\|_{\infty}$ as in Section \ref{sec:l_infty}, and the goal is to minimize
$
\sum_{i = 1}^R v_i.
$
Furthermore this reduction can be easily implemented in the distributed setting since each server can independently generate random variables in $D^{(i)}$ associated with its own input rows in $[A~b]$.
We can round each entry in $D^{(i)}$ to its nearest integer mutiple of $\poly(\varepsilon / d)$, which is enough for the correctness of Theorem \ref{thm:p_red}, but increases the bit complexity of each entry by at most an $O(\log(d / \varepsilon))$ factor. 

Using the linear program solvers in Section \ref{sec:clarkson} and Section \ref{sec:cog}, we have the following theorem.
\begin{theorem}\label{thm:l_p}
$(1 + \varepsilon)$-approximate $\ell_{p}$ regression can be solved by a randomized protocol with communication complexity $\widetilde{O}(sd^3L / \varepsilon^6)$ in the coordinator model, or  by a randomized protocol with communication complexity  $\widetilde{O}(\min\{sd^3L / \varepsilon^6, sd / \varepsilon^2 + d^4L / \varepsilon^8\})$ in the blackboard model. 
\end{theorem}

\subsection{Proof of Theorem \ref{thm:p_red}} \label{sec:proof_p_red}
We need the following Bernstein-type lower tail inequality which is due to Maurer \cite{maurer2003bound}.
\begin{lemma}[\cite{maurer2003bound}]\label{lem:maurer}
Suppose $X_1, X_2, \ldots, X_n$ are independent positive random variables that satisfy $\E[X_i^2] < \infty$.
Let $X = \sum_{i=1}^n X_i$. For any $t > 0$ we have
$$
\Pr[X \le \E[X] - t] \le \exp\left(-\frac{t^2}{2 \sum_{i=1}^n \E[X_i^2]}\right).
$$
\end{lemma}

We use the standard $\varepsilon$-net construction of a subspace  in \cite{bourgain1989approximation}.
\begin{definition}
For any $p \ge 1$, for a given $A \in \mathbb{R}^{n \times d}$, let $B = \{Ax \mid x \in \mathbb{R}^d, \|Ax\|_p = 1\}$. We say $\mathcal{N} \subseteq B$ is an $\varepsilon$-net of $B$ if for any $y \in B$, there exists a $\hat{y} \in \mathcal{N}$ such that $\|y - \hat{y}\|_p \le \varepsilon$.
\end{definition}
\begin{lemma}[\cite{bourgain1989approximation}]\label{lem:eps_net}
For a given $A \in \mathbb{R}^{n \times d}$, there exists an $\varepsilon$-net $\mathcal{N} \subseteq B = \{Ax \mid x \in \mathbb{R}^d, \|Ax\|_p = 1\}$ with size $|\mathcal{N}| \le (3 / \varepsilon)^d$.
\end{lemma}

\begin{lemma}[Auerbach basis \cite{auerbach1930area}]
For any matrix $A \in \mathbb{R}^{n \times d}$ and $p \ge 1$, there exists a basis matrix $U$ of the column space of $A$, such that $\|U_i\|_p = 1$ for all $i \in [d]$, and for any vector $x \in \mathbb{R}^d$,
$$
\|x\|_{\infty} \le \left\| Ux\right\|_p.
$$
\end{lemma}

Now we give the proof of Theorem~\ref{thm:p_red}.
Notice that for any fixed vector $y \in \mathbb{R}^n$, 
\[
\|D^{(i)} y\|_{\infty} \sim \|y\|_p E^{-1/ p}
\]
where $E$ is an exponential random variable. 
Moreover, when $p > 2$, both $\E[(E^{-1/ p})^2]$ and $\Var[E^{-1 / p}]$ are bounded by a constant. 

We have $\E[\|D^{(i)} y\|_{\infty}] = C_p \|y\|_p$.
By linearity of expectation, we also have
\[
\E \left[ \sum_{i=1}^R \|D^{(i)} y\|_{\infty} \right] = C_p R \|y\|_p.
\]

We use $U \in \mathbb{R}^{n \times (d + 1)}$ to denote an Auerbach basis of the column space of $\tilde{A} = [A~b]$. 
We create three events $\mathcal{E}_1$, $\mathcal{E}_2$, $\mathcal{E}_3$.
Here $C$ is an absolute constant. 
\begin{itemize}
\item $\mathcal{E}_1$: $\|D^{(i)} (Ax^* - b)\|_{\infty} \le C \cdot R^{1 / p} \cdot \|Ax^* - b\|_p$ for all $i \in [R]$.
\item $\mathcal{E}_2$: $\|D^{(i)} U_j\|_{\infty} \le C (R \cdot d)^{1 / p}$ for all $i \in [R]$ and $j \in [d + 1]$.
\item $\mathcal{E}_3$: for all $y \in \mathcal{N}$, 
\[
\sum_{i=1}^R \|D^iy\|_{\infty} \ge (1 - \varepsilon / 3)C_p R \|y\|_p,
\]
where $\mathcal{N}$ is a $(\poly(\varepsilon / d))$-net of $\{\tilde{A} x \mid x \in \mathbb{R}^{d + 1}, \| \tilde{A} x\|_p = 1\}$.
By Lemma \ref{lem:eps_net} we have $|\mathcal{N}| \le (d / \varepsilon)^{O(d)}$.
\end{itemize}

According to the cumulative density function of $E^{-1 / p}$ for an exponential random variable $E$, and a union bound over $i \in [R]$, $\mathcal{E}_1$ holds with constant probability.
Similarly, $\mathcal{E}_2$ also holds with constant probability.
For each $y \in \mathcal{N}$, using Maurer's inequality in Lemma \ref{lem:maurer}, we have
\[
\Pr \left[ \sum_{i=1}^R \|D^iy\|_{\infty} < (1 - \varepsilon / 3)C_p R \|y\|_p  \right] \le \exp\left( -\frac{(\varepsilon \cdot C_p R)^2}{O(R) }\right).
\]
Thus for $R = O(d \log(d / \varepsilon) / \varepsilon^2)$, by using a union bound for all $y \in \mathcal{N}$, with constant probability $\mathcal{E}_3$ holds.

Conditioned on $\mathcal{E}_1$, using Bernstein's inequality, we have
\[
\Pr \left[ \sum_{i=1}^R \|D^{(i)} (Ax^* - b)\|_{\infty} > (1 + \varepsilon)C_p R \|Ax^* - b\|_p  \right] \le \exp\left( -\frac{(\varepsilon \cdot C_p R)^2}{O(R) + O( \varepsilon \cdot C_p R^{1 + 1 / p}) }\right).
\]
Thus, for $R = O(d \log(d / \varepsilon) / \varepsilon^2)$ and $p > 2$, 
\[
\sum_{i=1}^R \|D^{(i)} (Ax^* - b)\|_{\infty} \le (1 + \varepsilon)C_p R \|Ax^* - b\|_p
\]
holds with constant probability, which implies Part 1 of Theorem \ref{thm:p_red}.

Now for any $y = U x$ with $\|y\|_p = 1$, by definition of the Auerbach basis we have $\|x\|_{\infty} \le \|y\|_p \le 1$.
Conditioned on $\mathcal{E}_2$, we have, 
$$
\sum_{i=1}^R \|D^{(i)} y\|_{\infty} \le \sum_{i=1}^R \sum_{j = 1}^{d + 1} \|D^{(i)} U_j \cdot x_j\|_{\infty} \le \sum_{i=1}^R \sum_{j = 1}^{d + 1} \|D^{(i)} U_j\|_{\infty} \le \poly(d / \varepsilon).
$$

Consider any $y = \tilde{A} x$ with $\|y\|_p = 1$. We claim $y$ can be written as
$$
y = \sum_{j = 0}^{\infty} y^j,
$$
where for any $j \ge 0$ we have (i) $\frac{y^j}{\|y_j\|_p} \in \mathcal{N}$ and (ii) $\|y^j\|_p \le \poly(\varepsilon / d)^j$.

According to the definition of a $(\poly(\varepsilon / d))$-net, there exists a vector $y^0 \in \mathcal{N}$ for which $\|y - y^0\|_p \le \poly(\varepsilon / d)$ and $\|y^0\|_p = 1$.
If $y = y_0$ then we stop. 
Otherwise we consider the vector $\frac{y - y^0}{\|y-y^0\|_p}$. 
Again we can find a vector $\hat{y}^1 \in \mathcal{N}$ such that $\left\| \frac{y - y^0}{\|y-y^0\|_p} - \hat{y}^1\right\|_p \le \poly(\varepsilon / d)$ and $\|\hat{y}^1\|_p = 1$.
Here we set $y^1 = \|y-y^0\|_p \cdot \hat{y}^1$ and continue this process inductively.

Thus, conditioned on $\mathcal{E}_2$ and $\mathcal{E}_3$, we have for any  $y = \tilde{A} x$ with $\|y\|_p = 1$, 
$$
\sum_{i = 1}^R \|D^{(i)}y\|_{\infty} \ge \sum_{i=1}^R \|D^{(i)}y^{0}\|_{\infty} - \sum_{i = 1}^R \sum_{j = 1}^{\infty} \|D^{(i)}y^j\| \ge (1 - \varepsilon)C_p R.
$$
For any $y = Ax - b$, by homogeneity, we still have
$$
\sum_{i=1}^R \|D^{(i)} y\|_{\infty} \ge (1 - \varepsilon) C_p R \|Ax - b\|_p,
$$
which implies Part 2 of Theorem \ref{thm:p_red}.

\section{Communication Complexity Lower Bound for Linear Programming}
In this section, we prove a communication complexity lower bound for testing feasibility of linear programs.

We need the following lemma to construct our hard instance. 
\begin{lemma}\label{lem:dot}
Let $L$ be a sufficiently large integer. We use $m_i \in \R^2$ to denote the vector 
$$m_i = \left(\frac{i}{2^L}, 1- \frac{i^2}{2 \cdot 4^L}\right).$$
For any $1 \le i, j \le 2^{L / 100}$, we have
\begin{enumerate}
\item $\|m_i\|_2^2 \ge  1 + \frac{1}{2^{4L+ 2}}$;
\item For any $i \neq j$, $\langle m_i, m_j \rangle \le 1$.
\end{enumerate}
\end{lemma}
\begin{proof}
For any $1 \le i \le 2^{L / 100}$, we have
$$
\|m_i\|_2^2  = \frac{i^2}{4^L} + \left(1- \frac{i^2}{2 \cdot 4^L} \right)^2 = 1 + \frac{i^4}{2^{4L+ 2}}\ge 1 + \frac{1}{2^{4L+ 2}}.
$$

For any $1 \le i, j \le 2^{L / 100}$ and $i \neq j$, we have
$$
\langle m_i, m_j \rangle = \frac{ij}{4^L} +  \left(1- \frac{i^2}{2 \cdot 4^L}\right) \cdot \left(1- \frac{j^2}{2 \cdot 4^L}\right) = 1- \frac{(i - j)^2}{2 \cdot 4^L} + \frac{i^2j^2}{4^{2L + 1}} \le 1- \frac{1}{2 \cdot 4^L} + \frac{i^2j^2}{4^{2L + 1}} \le 1.
$$
\end{proof}

Now we reduce the {\em lopsided set disjiontness} problem to testing feasibility of linear programs.
In this problem, for a choice of universe size $U$, the last server $P_s$ receives an element $u \in [U]$, and for each $i < s$, server $P_i$ receives a set $S_i \subseteq [U]$.
The goal is to test whether there exists $i$ such that $u \in S_i$. 
We reduce  this problem with $U = 2^{L / 100}$ to testing the feasibility of linear programs for $d = 2$, where $L$ is the bit complexity of the linear program. 

For the reduction, server $P_s$ adds a constraint $x = m_u$, for the element $u \in [U]$ that $P_s$ receives. I.e., server $P_s$ forces the solution $x$ to be $m_u$.
For each $i < s$, for each $v \in S_i$, server $P_i$ adds a constraint $\langle m_v, x \rangle \le 1$.
Here $m_u$ and $m_v$ are as defined in Lemma \ref{lem:dot}.
By Lemma \ref{lem:dot}, this linear program is feasible if and only if $u \notin \bigcup_{i < s} S_i$.

In the remaining part of this section, we show the lopsided set disjointness problem has an $\Omega(s \log U /  \log s)$ randomized communication complexity lower bound in the coordinator model, which implies an $\Omega(s \log L /  \log s)$ lower bound for testing feasiblity of linear programming, even for $d = 2$.
An $\Omega(s + L)$ lower bound also holds in the blackboard model, since when $s = 2$ the coordinator model is equivalent to the blackboard model, up to a constant factor in the communication complexity. 

We first consider the two-player case, in which Alice receives an element $u \in [U]$ and Bob receives a set $S \subseteq [U]$.
The goal is to test whether $u \in S$ or not. 
Let $\mu$ be the distribution where $u$ is chosen uniformly at random from $[U]$, and $S$ is a subset of $[U]$ such that each element $u \in U$ is included independently with probability $1 / 2$.
Let $\mu_y$ be the conditional distribution of $\mu$ given $u \in S$, and $\mu_n$ be the conditional distribution of $\mu$ given $u \notin S$.
In \cite[Section 2.2]{andoni2006optimality}, it has been shown that any communication protocol that succeeds with probability at least $2/3$ on the distribution $\mu$ requires $\Omega(\log U)$ bits of communication in the worst case. 
By applying Markov's inequality and stopping the protocol early once the communication complexity is too large, this implies any randomized protocol that succeeds with probability at least $3 / 4$ on the distribution $\mu$ requires $\Omega(\log U)$ bits of communication in expectation.
In fact, this implies a stronger hardness result, that for any protocol that succeeds with probability at least $3/4$ on $\mu$, its expected communication complexity is $\Omega(\log U)$ on both $\mu_y$ and $\mu_n$.

Consider a new distribution $\mu'$ which is $\mu_y$ with probability $1/ s^2$ and $\mu_n$ with probability $1 - 1 / s^2$.
Suppose a protocol $\mathcal{P}$ succeeds with probability at least $1 - 1 / 100s^2$. Then by averaging $\mathcal{P}$ succeeds with probaility at least $4 / 5$ on both $\mu_n$ and $\mu_y$, which implies the expected communication complexity of $\mathcal{P}$ is $\Omega(\log U)$ on both $\mu_y$ and $\mu_n$.
Now by linearity of expectation, the expected communication complexity of $\mathcal{P}$ on $\mu'$ is lower bounded by $\Omega(\log U)$.
This, in particular, implies any protocol that succeeds with probability at least $1 - 1 / 100s^2$ on $\mu'$ should have expected communication complexity $\Omega(\log U)$.
At this point, Theorem 1.1 in \cite{WZ13} implies that for the $s$-player case, any communication protocol that succeeds with probability at least $1 - 1 / s^3$ has worst case communication complexity at least $\Omega(s \log U)$.
By standard repitition arguments this implies an $\Omega(s \log U / \log s)$ lower bound for protocols that succeed with constant probability.

Formally, we have the following theorem.
\begin{theorem}\label{thm:rand_lp_lb}
Any randomized protocol that succeeds with probability at least $0.99$ for testing feasibility of linear programs requires $\Omega(s \log L /  \log s)$ bits of communication in the coordinator model and $\Omega(s + L)$ bits of communication in the blackboard model. 
The lower bound holds even when $d = 2$.
\end{theorem}

Notice that by Theorem \ref{thm:rand_feasibility}, testing feasibility of linear systems for $d = 2$ requires only $O(s \log L)$ randomized communication complexity.
This shows an exponential separation between testing feasibility of linear systems and linear programs, in the communication model. 
 
\newcommand{\trunc}{\mathsf{trunc}}
\section{Clarkson's Algorithm} \label{sec:clarkson}
\subsection{The Communication Complexity}
In this section, we discuss how to implement Clarkson's algorithm to solve linear programs in the distributed setting. 
The protocol is described in Figure \ref{alg:clarkson}.
During the protocol, each server $P_i$ maintains a multi-set $H_i$ of constraints (i.e., each constraint can appear more than once in $H_i$).
Initially, $H_i$ is the set of constraints stored on $P_i$.
Furthermore, the coordinator maintains $|H_i|$, which is initially set to be the number of constraints stored on each server.

\begin{figure}[!htb]
\begin{framed}
\begin{enumerate}
\item \label{step:lp_sample} The coordinator obtains $9d^2$ constraints $R$, by sampling uniformly at random from $H_1 \cup H_2 \cup \ldots \cup H_s$.
\item \label{step:lp_solve}The coordinator calculates the optimal solution $x_R$, which is the optimal solution to the linear program satisfying all constraints in $R$. The coordinator sends $x_R$ to each server.
\item Each server $P_i$ calculates the total number of constraints that are stored on $P_i$ and violated by $x_R$, i.e., $|V_i|$ where $V_i = \{h \in H_i \mid x_R \text{ violates } h\}$.
\item \label{step:clarkson_test}The coordinator calculates $|V| = \sum_{i=1}^s |V_i|$ where $V = \{h \in H \setminus R \mid x_R \text{ violates } h\}$ and sends $|V|$ to each server.
\begin{enumerate}
\item If $|V| = 0$ then $x_R$ is a feasible solution and the protocol terminates.
\item If $|V| \le \frac{2}{9d - 1}|H| $, then each server updates $H_i \gets H_i \cup V_i$ and the coordinator updates $|H_i| \gets |H_i| + |V_i|$.
\item Goto Step \ref{step:lp_sample}.
\end{enumerate} 
\end{enumerate}
\end{framed}
\caption{Clarkson's Algorithm}
\label{alg:clarkson}
\end{figure}

The protocol in Figure \ref{alg:clarkson} is basically Clarkson's algorithm \cite{clarkson1995vegas}, implemented in the distributed setting. 
Using the analysis in \cite{clarkson1995vegas}, the expected number of iterations is $O(d \log n)$.
The correctness also directly follows from the analysis in \cite{clarkson1995vegas}.
Now we analyze the communication complexity for each iteration.

To implement the sampling process in Step \ref{step:lp_sample}, the coordinator first determines the number of constraints to be sampled from each server $P_i$ and sends this number to $P_i$.
The total communication complexity for this step is $O(s \log n)$ in both the coordinator model and the blackboard model.
Then each server $P_i$ samples accordingly and sends these constraints to the coordinator. 
The total communication for this step is $O(d^3L)$ in both models.

To implement Step \ref{step:lp_solve}, we first verify the bit complexity of the optimal solution $x_R$.
One of the optimal solutions $x_R$ is a vertex of the polyhedron $Ax \le b$.
From polyhedral theory we know that there exists a non-singular subsystem of $Ax \le b$, say $Bx \le c$, such that $x_R$ is the unique solution of $Bx = c$.
Thus, by Cramer's rule, each entry of $x$ is a fraction whose numerator and denominator are integers between $-d!2^{dL}$ and $d!2^{dL}$, and thus can be represented by using at most $O(dL + d \log d)$ bits.
This implies the bit complexity of all entries in the vector $x$ calculated at Step \ref{step:lp_solve} is upper bounded by $\widetilde{O}(d^2 L)$. Thus the communication complexity for Step \ref{step:lp_solve} is upper bounded by $\widetilde{O}(sd^2 L)$ in the coordinator model and $\widetilde{O}(d^2 L)$ in the blackboard model.
The communication complexity of the last two steps of the protocol is upper bounded by $O(s \log n)$ in both models.
Thus, the expected communication complexity is $\widetilde{O}(sd^3L + d^4L)$ in the coordinator model and $\widetilde{O}(sd + d^4L)$ in the blackboard model. 

\begin{theorem}
The expected communication complexity of the protocol in Figure \ref{alg:clarkson} is $\widetilde{O}(sd^3L + d^4L)$ in the coordinator model and $\widetilde{O}(sd + d^4L)$ in the blackboard model
\end{theorem}


\subsection{Running Time of Clarkson's Algorithm in Unit Cost RAM}\label{apx:clarkson}

In this section, we show how to implement Clarkson's algorithm in the unit cost RAM model on words of size $O(\log(nd))$ so that the running time is upper bounded by $\tilde{O}(nd^\omega L + \poly(dL))$, and prove Theorem \ref{thm:LPbit}.

A description of Clarkson's algorithm can be found in Figure \ref{alg:clarkson}. 
This algorithm runs in $O(d \log n)$ rounds in expectation.
In each round, it samples $O(d^2)$ constraints $R$, and calculates an optimal solution $x_R$ that satisfies all constraints in $R$.
This optimal solution $x_R$ can be calculated using any polynomial time linear programming algorithm, which always has running time $\poly(dL)$.
The bottleneck in the unit cost RAM model is Step \ref{step:clarkson_test} of the algorithm in Figure \ref{alg:clarkson}, i.e., for each of the $n$ constraints, testing whether $x_R$ satisfies the constraint or not. 
Formally, we just need to output $Ax_R$, and then compare each entry with $b$. 
In the remaining part of this section we show how to caculate $Ax_R$ in $\tilde{O}(nd^\omega L)$ time

Since each entry of $x_R$ has bit complexity $\widetilde{O}(dL)$, we first calculate a $d \times d$ matrix $X$, where each entry of $X$ has bit complexity $\widetilde{O}(L)$, and the entry $X_{i, 1}, X_{i, 2}, \ldots, X_{i, d}$ consists of the first $\widetilde{O}(L)$ bits of $(x_R)_i$, the second $\widetilde{O}(L)$ bits of $(x_R)_i, \ldots,$ the last $\widetilde{O}(L)$ bits of $(x_R)_i$. 
Now we calculate $A \cdot X$. Since all entries in $A$ and $X$ have bit complexity $\widetilde{O}(L)$, and caculating the matrix mutilplication of two $d \times d$ matrices with bit complexity $\widetilde{O}(L)$ requires only $\widetilde{O}(d^{\omega}L)$ time \cite{harvey2018complexity}, $A \cdot X$ can therefore be calculated in $\widetilde{O}((n / d) \cdot d^{\omega} \cdot L) = \widetilde{O}(nd^{\omega - 1} L)$ time. 
Given $AX$, one can then easily calculate $Ax_R$ in $\widetilde{O}(ndL)$ time. 
Thus, the total expected running time is upper bounded by $O(d \log n (\widetilde{O}(nd^{\omega - 1} L) + \widetilde{O}(ndL) + \poly(dL))) = \tilde{O}(nd^\omega L + \poly(dL))$.

\subsection{Smoothed Analysis of Communication Complexity}
In this section we define our model for smoothed analysis of communication complexity of communication protocols for solving linear programming. 

For a randomized communication protocol $\mathcal{P}$, we use $C_{\mathcal{P}}(A, b, c)$ to denote its communication complexity on the linear programming instance
\begin{equation}\label{equ:lp}
\min_{Ax \le b} c^T x,
\end{equation}
where $A \in \mathbb{R}^{n \times d}$, $b \in \mathbb{R}^n$ and $c \in \mathbb{R}^d$.
The standard definition \cite{spielman2001smoothed} of smoothed analysis assumes that each entry of $A$ is perturbed by i.i.d. Gaussian noise with zero mean and standard deviation $\sigma$.
However, since we are measuring the communication complexity in terms of bit complexity, we cannot allow the noise to be arbitrary real numbers. 
Instead, in our model, we use discrete Gaussian random variables as the noise. 

Formally, we use $ \trunc_t: \R \to \R$ to denote the function that rounds a real number to its nearest integer multiple of $2^{-t}$.
For notational convenience, we define $\trunc_{\infty}(x) = x$.
We say a communication protocol solves the linear program instance \eqref{equ:lp} with  {\em smoothed communication complexity} $SC_{\mathcal{P}, \sigma, t}(A, b, c)$ if with probability at least $0.99$, 
the protocol correctly solves the instance 
\begin{equation}\label{equ:smooth_lp}
\min_{(A + G_{t, \sigma})x \le b} c^T x
\end{equation}
with communication complexity $C_{\mathcal{P}}(A + G_{t, \sigma}, b, c) \le SC_{\mathcal{P}, \sigma, t}(A, b, c)$, where all entries of $G_{t, \sigma}$ are i.i.d. copies of $\trunc_t(g)$ and $g$ is a Gaussian random variable with zero mean and $\sigma \le 1$ standard deviation. 
Here the probability is defined over the randomness of the protocol and the noise $G_{t, \sigma}$.
Notice that when $t = \infty$, $G_{\infty, \sigma}$ is a matrix whose all entries are i.i.d. Gaussian random variables with standard deviation $\sigma$.
\subsection{Smoothed Analysis of Clarkson's Algorithm}
\begin{figure}[!htb]
\begin{framed}
\begin{enumerate}
\item \label{step:slp_sample} The coordinator obtains $9d^2$ constraints $R$, by sampling uniformly at random from $H_1 \cup H_2 \cup \ldots \cup H_s$.
\item \label{step:slp_solve}The coordinator calculates the optimal solution $x_R$, which is the optimal solution to the linear program satisfying all constraints in $R$. 
{\bf The coordinator rounds each entry of $x_R$ to its nearest integer multiple of $\delta = O(1  / \poly(nd \cdot 2^L / \sigma))$ to obtain $\hat{x}_R$, and sends $\hat{x}_R$ to each server.}
\item Each server $P_i$ calculates the total number of constrains that are stored on $P_i$ and violated by $\hat{x}_R$, i.e., $|V_i|$ where $V_i = \{h \in H_i \mid \hat{x}_R \text{ violates } h\}$.
\item The coordinator calculates $|V| = \sum_{i=1}^s |V_i|$ where $V = \{h \in H \ \setminus R \mid \hat{x}_R \text{ violates } h\}$ and sends $|V|$ to each server.
\begin{enumerate}
\item If $|V| = 0$ then $x_R$ is a feasible solution and the protocol terminates.
\item if $|V| \le \frac{2}{9d - 1}|H| $, then each server updates $H_i \gets H_i \cup V_i$ and the coordinator updates $|H_i| \gets |H_i| + |V_i|$.
\item Goto Step \ref{step:lp_sample}.
\end{enumerate} 
\end{enumerate}
\end{framed}
\caption{Smoothed Clarkson's Algorithm}
\label{alg:smooth_clarkson}
\end{figure}

In this section, we present our variant of Clarkson's algorithm for solving smoothed linear programming instances. 
The protocol is described in Figure \ref{alg:smooth_clarkson}.
The main difference is in Step \ref{step:slp_solve}, where the coordinator rounds each entry of the solution $x_R$ before sending it to other servers. 

\subsubsection{Correctness of the Protocol}
We first prove the correctness of the protocol.
Our plan is to show if $t = \Omega(\log(nd / \sigma) + L)$, then our modified Clarkson's algorithm follows the computation path of the original Clarkson's algorithm in Figure \ref{alg:clarkson} when executing on the perturbed instance, with high probability, and thus prove the correctness of the protocol. 

 We need the following bound on the condition number of a matrix.
\begin{lemma}\label{lem:slp_cond}
For a matrix $B \in \R^{d \times d}$ with all entries in $[0, 2^L - 1]$, for any integer $n > 0$, $\sigma \le 1$ and $t \ge \Omega(\log(nd / \sigma) + L)$, we have
$$
\Pr[\|B + G_{\sigma, t}\|_2 \ge \poly(nd \cdot 2^L / \sigma)] \le 1 - 1 /  \poly(nd)
$$
and
$$
\Pr[\|(B + G_{\sigma, t})^{-1}\|_2 \ge \poly(nd \cdot 2^L/  \sigma)] \le 1 - 1 /  \poly(nd).
$$
\end{lemma}
\begin{proof}
To prove the first inequality, notice that
$$
\|B + G_{\sigma, t}\|_2 \le \|B + G_{\sigma, t}\|_F \le \|B\|_F + \|G_{\sigma, t}\|_F \le \poly(d2^L) + \|G_{\sigma, t}\|_F.
$$
Thus, the first inequality just follows from tail inequalities of the Guassian distribution. 

To analyze $\|(B + G_{\sigma, t})^{-1}\|_2$, we write $G_{\sigma, \infty}$ to denote a matrix whose entries are the Gaussian random variables of $G_{\sigma, t}$ before applying the truncation operation. 
Notice that this implies $\|G_{\sigma, t} - G_{\sigma, \infty}\|_2 \le \poly(d) \cdot 2^{-t} \le 1 / \poly(nd / \sigma)$.
We invoke Theorem 3.3 in \cite{sankar2006smoothed}, which states that with probability $1 - 1 / \poly(nd)$,
$$
\|(B + G_{\sigma, \infty})^{-1}\|_2 \le \poly(nd / \sigma),
$$  
which implies with probability $1 - 1 / \poly(nd)$,
\begin{align*}
&\|(B + G_{\sigma, t})^{-1}\|_2
=  \left(\inf_{\|x\|_2 = 1} \|(B + G_{\sigma, t})x\|_2 \right)^{-1}\\
 = &  \left(\inf_{\|x\|_2 = 1} \|(B + G_{\sigma, \infty} + (G_{\sigma, t} - G_{\sigma, \infty}))x\|_2 \right)^{-1} \le \left(\inf_{\|x\|_2 = 1} \|(B + G_{\sigma, \infty})x\|_2 -\|G_{\sigma, t} - G_{\sigma, \infty}\|_2  \right)^{-1} \\
\le & \poly(nd\cdot 2^L / \sigma).
\end{align*}
\end{proof}

\begin{lemma}\label{lem:x_norm}
During the execution of the protocol in Figure \ref{alg:smooth_clarkson}, each time Step \ref{step:slp_solve} is executed,
if $x_R \neq 0$,
with probability at least $1 - 1 / \poly(nd)$, $x_R$ satisfies
$$
1 / \poly(nd \cdot 2^L / \sigma) \le \|x_R\|_2 \le \poly(nd \cdot 2^L / \sigma).
$$
\end{lemma}
\begin{proof}
From polyhedral theory we know that there exists a non-singular subsystem of the sampled $9d^2$ constraints $R$, say $Bx \le c$, such that $x_R$ is the unique solution of $Bx = c$.

If $c \neq 0$,
since each entry of $B$ was pertubed by a discrete Gaussian noise, and all entries of $c$ are integers in the range $[0, 2^{L} - 1]$, 
by Lemma \ref{lem:slp_cond} we have
$$
\|x_R\|_2 \le \|B^{-1}\|_2 \|c\|_2 \le \poly(nd \cdot 2^L / \sigma).
$$
Furthermore, since $\|B\|_2 \le \poly(nd \cdot 2^L / \sigma)$,
$$
\|x_R\|_2 \ge 1 /  \poly(nd \cdot 2^L / \sigma).
$$

If $c = 0$, then we must have $x_R = 0$, since $Bx = c$ is non-singular and thus $x_R = 0$ is the unique solution.
\end{proof}

Now we create a family of events $\{ \mathcal{E}_i\}_{i= 1}^{ \infty}$.
We use $\mathcal{E}_i$ to denote the event that, during the $i$-th loop of the execution of the protocol in Figure \ref{alg:smooth_clarkson}, for each constraint $h \notin R$, the constraint $h$ can be satisfied by $x_R$ if and only if it can be satisfied by $\hat{x}_R$.
Notice that for those constraints in $R$, $x_R$ can always satisfy them, by definition of $x_R$.

Now we show that for each $\mathcal{E}_i$, the probability that $\mathcal{E}_i$ holds is at least $1 - 1 / \poly(nd)$. 
By showing this, we have actually shown our algorithm follows the computation path of the original Clarkson's algorithm in Figure \ref{alg:clarkson} when executing on the perturbed instance, with high probability. 
Since the original Clarkson's algorithm in Figure \ref{alg:clarkson} terminates in $O(d \log n)$ rounds with probability at least $0.999$, the correcntess of our algorithm follows by applying a union bound over all events $\{\mathcal{E}_i\}_{i = 1}^{O(d \log n)}$.

To show that for each $\mathcal{E}_i$, the probability that $\mathcal{E}_i$ holds is at least $1 - 1 / \poly(nd)$, by applying a union bound over all constraints, it suffces to show that for each constaint $h \notin R$, $x_R$ can satisfy $h$ if and only if $\hat{x}_R$ can satisfy $h$, with probability $1 - 1 / \poly(nd)$.

\begin{lemma}
For each constraint $h \notin R$, 
with probability $1 - 1 / \poly(nd)$,
$x_R$ can satisfy $h$ if and only if $\hat{x}_R$ can satisfy $h$.
\end{lemma}
\begin{proof}
If $x_R = 0$, then $\hat{x}_R = x_R = 0$, in which case the lemma follows trivially. Thus we assume $x_R \neq 0$ in the remaining part of this proof.

The constraint $h$ can be written as $(a_h + g_{\sigma, t}) x \le b_h$, for some vector $a_h \in \R^d$ and some $b_h \in \R$, and all entries of $g_{\sigma, t} \in \R^d$ are i.i.d. copies of $\trunc_t(g)$ and $g$ is a Gaussian random variable with zero mean and $\sigma$ standard deviation. 
Notice that since $h \notin R$, the vector $g_{\sigma, t}$ and the vector $x_R$ are independent. 
By Lemma \ref{lem:x_norm}, with probability at least $1 - 1 / \poly(nd)$, 
$$
1 / \poly(nd \cdot 2^L / \sigma) \le \|x_R\|_2 \le \poly(nd \cdot 2^L / \sigma).
$$
Furthermore, the probability that $x_R$ can satisfy $h$ but $\hat{x}_R$ cannot satisfy $h$, or $x_R$ cannot satisfy $h$ but $\hat{x}_R$ can satisfy $h$, is at most
$$
\Pr[|\langle  g_{\sigma, t}, x_R\rangle | \le  |\langle  a_h + g_{\sigma, t}, x_R - \hat{x}_R\rangle |].
$$
We first analyze the right hand side of the inequality. Notice that $\|a_h\|_2 \le \poly(d \cdot 2^L)$, and $\|g_{\sigma, t}\|_2 \le \poly(nd)$ with probability at least $1 - 1 / \poly(nd)$ by tail inequalities of the Gaussian distribution and $\sigma \le 1$. Moreover, $\|x_R - \hat{x}_R\|_2 \le \poly(d) \cdot \delta \le 1 / \poly(dn \cdot 2^L / \sigma)$.
Thus by Cauchy-Schwarz, 
$$
\Pr[ |\langle  a_h + g_{\sigma, t}, x_R - \hat{x}_R\rangle | \le 1 / \poly(dn \cdot 2^L / \sigma)] \ge 1 - 1 / \poly(nd).
$$

On the other hand,  if we write $g_{\sigma, \infty}$ to denote a vector whose entries are the Gaussian random variables of $g_{\sigma, t}$ before applying the truncation operation, then $\|g_{\sigma, \infty} - g_{\sigma, t}\|_2 \le \poly(d)2^{-t}$.

Thus, by taking $t = \Omega(\log(nd / \sigma) + L)$, we have
$$
|\langle  g_{\sigma, t}, x_R\rangle| \ge | \langle  g_{\sigma, \infty}, x_R\rangle| - \poly(d) 2^{-t} \|x_R\|_2 \ge| \langle  g_{\sigma, \infty}, x_R\rangle| - 1 / \poly(nd \cdot 2^L / \sigma).
$$
By the lower tail inequality of the Gaussian distribution and the fact that $\|x_R\|_2 \ge 1 / \poly(nd \cdot 2^L / \sigma)$, we have with probability at least $1 - 1 / \poly(nd)$,
$$
|\langle  g_{\sigma, t}, x_R\rangle| \ge 1 / \poly(nd \cdot 2^L / \sigma).
$$
Thus, the lemma follows by appropriately adjusting the constant in $O(1  / \poly(nd \cdot 2^L / \sigma)) = \delta$.
\end{proof}
\subsubsection{Communication Complexity of the Algorithm}
The analysis in the preceding section shows that with high probability, our modified Clarkson's algorithm follows the computation path of the original Clarkson's algorithm, and thus also terminates within $O(d \log n)$ rounds with probability at least $0.999$.
Furthermore, with high probability, the discrete Gaussian noise of all entries is upper bounded by $O(nd)$. Thus, the bit complexity of sending each constraint will be $\widetilde{O}(d(L + t))$, with high probability.

The sampling process in Step \ref{step:slp_sample} requires $\wt{O}(d^3(L + t) + s)$ bits of communication to sample $O(d^2)$ constraints. 
To implement Step \ref{step:slp_solve}, we need to verify the bit complexity of $\hat{x}_R$.
Since we round each entry of $x_R$ to its nearest integer multiple of $\delta$, and by Lemma \ref{lem:x_norm}, with high probability, $\|x_R\|_2 \le \poly(nd \cdot 2^L / \sigma)$, the communication compleixty for sending $\hat{x}_R$ is upper bounded by $\wt{O}(sd(L + \log(1 / \sigma)))$.
The communication complexity of the last two steps of the protocol is still upper bounded by $O(s \log n)$. 
Thus, the smoothed communication complexity is $\widetilde{O}(sd^2(L + \log(1 / \sigma)) + d^4(L + t))$ in the coordinator model. 

\begin{theorem}
For $t = \Omega(\log(nd / \sigma) + L)$, the protocol in Figure \ref{alg:smooth_clarkson} correctly solves smoothed linear programming with probability at least $0.99$, and the smoothed communication compleixty is
$$
SC_{\mathcal{P}, \sigma, t} \le \wt{O}(sd^2(L + \log(1 / \sigma)) + d^4(L + t))
$$
in the coordinator model. 
\end{theorem}

\section{The Center of Gravity Method}\label{sec:cog}

In this section, we discuss how to implement the center-of-gravity cutting-plane method \cite{Grunbaum1960} in the distributed setting.
The description of the protocol can be found in Figure \ref{alg:cog}.

The servers each maintain a polytope $P$ (the same one for all servers), adding a constraint in each iteration. Each server also maintains the center of the polytope $z$ and its covariance $C$. 

For any vector $a \in \R^d$, its $\eps$-rounding $\tilde{a}$ w.r.t. to $C$ is defined as follows: Let $B = C^{1/2}$. We take the unit vector $B^Ta/\|B^Ta\|_2$, round it down to the nearest multiple of $\eps$ in each coordinate. 
So we have $\|\tilde{a}-B^Ta / \|B^Ta\|_2\|_2 \le \eps\sqrt{d}$.


\begin{figure}
\begin{framed}
\begin{enumerate}
\item Set $P=[-R,R]^n, z=0, C = I$.

\item Repeat, while no feasible solution found, and at most $T$ times:
\begin{enumerate}
\item Set $z$ to be the centroid of $P$ and 
$C$ to be the covariance matrix of the uniform distribution over $P$, i.e., $C = \E_{x \sim P} (x-z)(x-z)^T$.
\item Each server $P_i$ checks $z$ against their subset of constraints. If a violated constraint $a \cdot z > b_i$ is found, and no violation has been reported by other servers so far, then it rounds $a$ to $\tilde{a}$ and broadcasts. If a violation has been reported, it uses the $\tilde{a}$ broadcast by the server that first found the violation.
\item Set $P= P \cap \{x: C^{-1/2}\tilde{a} \cdot x \le C^{-1/2}\tilde{a}\cdot z + \eps d^{3/2}\|C^{-1/2}\tilde{a}\|_2\}$. 
\end{enumerate}
\end{enumerate}
\end{framed}
\caption{Center-of-Gravity Algorithm}
\label{alg:cog}
\end{figure}

If each server were to report the exact violated constraint, the volume of $P$ would drop by a constant factor in each iteration. To reduce the communication, we round the constraint and shift it away a bit to make sure that the rounded constraint (1) is still valid for the target LP and (2) it is close enough that the volume still drops by a constant factor.

\begin{lemma}[\cite{BV04}]\label{lem:cut}
Let $z$ be the center of gravity of an isotropic convex body $K$ in $\R^d$. Then, for any halfspace $H$ within distance $t$ of $z$, we have
\[
\vol(K \cap H) \ge \left(\frac{1}{e}-t\right)\vol(K).
\] 
\end{lemma}

\begin{lemma}
For $\eps < 0.1/d\sqrt{d}$, the volume of the polytope $P$ maintained by each server drops by a constant factor in each iteration.
\end{lemma}

\begin{proof}
Assume without loss of generality that $P$ is isotropic. If the centroid $z=0$ is not feasible, we get a violated constraint such that the entire feasible region lies in the halfspace $a\cdot x \le 0$ with $\|a\|_2=1$. Now we replace $a$ by $\tilde{a}$. As a result,
\[
\tilde{a}\cdot x \le a \cdot x + \|\tilde{a}-a\|_2 \|x\|_2 \le \eps d^{3/2}.
\]
Here we used the fact that in isotropic position any convex body is contained in a ball of radius $d$, so $\|x\|_2\le d$ for all of $P$ and therefore for the feasible region. Thus the constraint imposed by the algorithm is valid.

Next, we note that the distance of the constraint from the origin is at most $\eps d^{3/2}$, so for $\eps < 0.1/d^{3/2}$, it is less than $0.1$ (in isotropic position). By Lemma \ref{lem:cut}, with $t=0.1$, the volume of $P$ drops by a constant factor.
\end{proof}

\begin{theorem}
The protocol in Figure \ref{alg:cog} is a deterministic protocol for solving linear programming with communication complexity $O(sd^3L\log^2 d)$ in both the coordinator model and the blackboard model. 
\end{theorem}
\begin{proof}
The algorithm runs for $T = O(d^2L\log d)$ rounds. To see this we note that the each vertex of the feasible region is the solution of a subset of the linear equalities taken to be equalities. Thus, each coordinate of each vertex is a ratio of two determinants of matrices whose entries are $L$-bit numbers and so the maximum distance of a vertex from the origin is $R=d^{O(dL)}$, which upper bounds the volume by $d^{O(d^2L)}$. The smallest any coordinate can be is similarly $d^{-O(dL)}$. The minimum volume we need to go to is the volume spanned by a simplex of vertices, which itself is a determinant with entries of this size. Thus, the volume is at least $d^{-O(d^2L)}$. Since the volume of the polytope maintained drops by a constant factor in each iteration\footnote{The ellipsoid method uses the same argument, except that each round reduces the volume by only $(1-1/d)$ \cite{GLS}.}, the number of rounds is $O(d^2L\log d)$. Each round includes a broadcast of a single vector, with $O(d\log d)$ bits. This is because the size of the $\eps$-net used is $d^{O(d)}$. By viewing the objective function as a constraint, we note that the volume bounds used above apply to the optimization version as well. At the end, we use diophantine approximation to get an exact solution \cite{GLS}.
\end{proof}

A similar argument applies to general convex programming.
\begin{theorem}
The communication complexity of the protocol in Figure \ref{alg:cog} for solving convex programming is $O(sd^2\log d\log(Rd / \eps))$.
\end{theorem}
\begin{proof}
The initial volume is at most $R^d$ and the algorithm can stop when the volume is $(d/\eps)^{-O(d)}$. Therefore the number of rounds is $O(d \log(Rd / \eps))$. Each round uses $O(sd\log d)$ bits giving the final bound.  
\end{proof}

%



\section{Seidel's Algorithm}\label{sec:seidel}
We give an alternative constant dimensional linear programming algorithms in the blackboard model, based on Seidel's classical algorithm \cite{seidel1990linear}.
Here we additionally assume that each constraint in the linear program is placed on a random server.
This assumption is essential to get rid of the $\log n$ dependence in the communication complexity. 
Here we also assume that the linear program is bounded. 

To implement Seidel's algorithm in the blackboard model, we go through all servers $P_1, P_2, \ldots, P_s$, and for each server $P_i$, we go through all constraints stored on $P_i$ in a random order. 
We maintain the optimal solution $x^*$ to the set of constraints that we have already went through.
For a new constraint $\langle a, x \rangle \le b$, the current server first checks whether the constraint is satisfied or not.
The current server proceeds to the next constraint if it is indeed satisfied.
If it is not satisfied, then the current constraint $\langle a, x \rangle \le b$ must be one of the $d$ constraints that determines the current optimal solution.
In this case, the current server broadcasts the current constraint $\langle a, x \rangle \le b$ to all other servers, and makes a recursive call to figure out the optimal solution, by adding an equality constraint $\langle a, x \rangle = b$ to the set of constraints. 
Notice that if the first server $P_1$ finds a violated constraint, $P_1$ does not need to broadcast the violated constraint, since $P_1$ can simply add the equality constraint to the beginning of all constraints owned by $P_1$.

One major difference between the classical Seidel's algorithm and our implementation is that each time we make a recursive call, we {\em do not} randomly permute the constraints again in the recursive calls.
Instead, the order that we go through the servers is fixed, in different recursive calls.
Due to this difference, there will be subtle dependence between the communication complexity of different recursive calls.

We use $\mathcal{E}$ to denote the event that the number of constraints stored on $P_1$ is at least $\Omega(n / s)$.
By a Chernoff bound, $\mathcal{E}$ holds with probability at least $0.99$.
For the $i$-th constraint (in the order that we go through all constraints), we let the random variable $V_i$ be $1$ if the $i$-th constraint is one the $d$ constraints that determines the optimal solution among the first $i$ constrains, and $0$ otherwise.

Since each constraint in the linear program is placed on a random server, by standard backward analysis, $\E[V_i] = d/i$.
Furthermore, $\E[V_i \mid \mathcal{E}] \ge \E[V_i] / \Pr[\mathcal{E}] = O(d / i)$.
However, conditioned on $\mathcal{E}$, the first $\Omega(n / s)$ constraints will not be broadcasted and there will be no recursive calls associated with them, since they are stored on the first server $P_1$.
Thus, conditioned on $\mathcal{E}$, the expected number of broadcasts (and thus recursive calls) is upper bounded by 
$$
\sum_{i = \Omega(n/s)}^n O(d / i) = O(d \log s).
$$

We use $\mathcal{F}_d$ to denote the event that there are at most $O(d^2 \log s)$ recursive calls made at the top layer of the recursive tree corresponding to Seidel's algorithm.
Conditioned on $\mathcal{E}$, by Markov's inequality, $\mathcal{F}_d$ holds with probability at least $1 - 1 / (100d)$.
Similarly, we use $\mathcal{F}_{d - 1}$ to denote the event that there are at at most $O((d^2 \log s)^2)$ recursive calls made at the second layer of the recursive tree corresponding to Seidel's algorithm.
Conditioned on $\mathcal{F}_d$ and $\mathcal{E}$, again by Markov's inequality, $\mathcal{F}_{d - 1}$ holds with probability at least $1 - 1 / (100d)$.
We similarly define $\mathcal{F}_{d - 2}, \mathcal{F}_{d - 3}, \ldots, \mathcal{F}_{1}$.
Thus, conditioned on $\mathcal{E}$, with probability at least $0.99$, $\mathcal{F}_i$ holds for all $i \in [d]$.
Which implies the total number of broadcasts is upper bounded by $O(\log^ds)$.

In each broadcast, the current server needs to broadcast the current constraint, which has bit complexity $O(dL)$.
Moreover, after the recursive call, the current server broadcasts the current solution vector $x^*$.
By polyhedral theory, $x^*$ can be achieved by setting $d$ inequality constraints to be equality constraints.
Thus, by Cramer's rule, each entry of $x^*$ has bit complexity $O(d \log d + dL)$.
The total communication complexity is hence upper bounded by $O(s + L \cdot \log^ds)$, with probability at least $0.9$.
Here the randomness is over the initial random assignment of each constraint and the random coins tossed by the algorithm.

\begin{theorem}\label{thm:seidel}
Seidel's algorithm can be used to solve linear programs in constant dimension with $O(s+L\log^d s)$ communication in the blackboard model, if each constraint in the linear program is placed on a random server.
Here the randomness is over the initial random assignment of each constraint and the random coins tossed by the algorithm.

\end{theorem}

\section{Singularity Probability}\label{sec:singularity}
The goal of this section is to prove Theorem \ref{thm:singularity_main}. We restate it here for convenience.

\noindent\textbf{Theorem~\ref{thm:singularity_main}.} (restated)
\textit{
Let $M_n$ be a matrix whose entries are i.i.d. random variables with the same distribution as $\sumsign_t$, for sufficiently large $t$,
$$
\Pr\left[M_n \text{ is singular}\right] \le t^{-C n},
$$
where $C > 0$ is an absolute constant. }

Our proof of Theorem \ref{thm:singularity_main} follows very closely the proof of Theorem 1.5 in \cite{tao2006random}.
Throughout this section we use $\lambda$ to denote $t^{-1 / 2}$.
We use $X_i$ to denote the $i$-th row of $M_n$.

We need the following lemma on generalized binomial distributions.
\begin{lemma}\label{lem:binomial_zero}
We have
\begin{equation}\label{equ:B_1}
\Pr[\sumsign_t = 0] \le c_1 t^{- 1 / 2}
\end{equation}
and
\begin{equation}\label{equ:B_mu}
\Pr[\sumsign_t^{(\lambda e^{-\lambda})} = 0] \le  c_2 t^{- 1 / 4}.
\end{equation}
Here $c_1$ and $c_2$ are absolute constants.
\end{lemma}
\begin{proof}
By Stirling's approximation we have
$$
\Pr[\sumsign_t = 0] = \binom{t}{t / 2} / 2^t = \Theta(t^{-1 / 2}),
$$
which proves (\ref{equ:B_1}).
To prove (\ref{equ:B_mu}), by a Chernoff bound swe have that the number of non-zero terms in the summation of $\sumsign_t^{(\lambda e^{-\lambda})}$ is $\Theta(t \lambda e^{-\lambda}) = \Theta(t^{1 / 2})$ with probability $1 - \exp(-\Omega(t^{1 / 2}))$.
Conditioned on this event, we can then prove (\ref{equ:B_mu}) by using the same estimation as (\ref{equ:B_1}).
\end{proof}

The following lemma is a direct implication of Lemma \ref{lem:binomial_zero} and Odlyzko's results in \cite{odlyzko1988subspaces}. See also Lemma 2.1 in \cite{tao2006random} and Section 3.2 in \cite{kahn1995probability}.
\begin{lemma} \label{lem:vector_in_hyperplane}
Let $W \subseteq \mathbb{R}^n$ be an arbitrary subspace and $X^{(\mu)} \in \mathbb{R}^n$ whose entries are i.i.d. random variables with the same distribution as $\sumsign_t^{(\mu)}$.
We have 
$$
\Pr[X \in W] = \left(c_1 t^{- 1 / 2}\right)^{n - \mathrm{dim}(W)}
$$
and
$$
\Pr[X^{(\lambda e^{-\lambda})} \in W] \le \left(c_2 t^{- 1 / 4}\right)^{n - \mathrm{dim}(W)}.
$$
\end{lemma}

By Lemma 5.1 in \cite{tao2006random}, we have
$$
\Pr[X_1, X_2, \ldots, X_n \text{ are linearly independent}] \le t^{o(n)} \Pr[X_1, X_2, \ldots, X_n \text{ span a hyperplane}],
$$
which implies we only need to consider the case when $X_1, X_2, \ldots, X_n$ span a hyperplane.

We say a hyperplane $V$ is non-trivial if $V$ is spanned by its intersection with $\{-t, -(t - 1), \ldots, -1, 0, 1, \ldots, t - 1, t\}^n$.
Notice that a hyperplane $V$ has 
$$
\Pr[X_1, X_2, \ldots, X_n \text{ span } V] > 0
$$
only when $V$ is non-trivial. Thus, we focus only on non-trivial hyperplanes in the remaining part of the proof.

\begin{definition}
Let $X \in \mathbb{R}^n$ whose entries are i.i.d. random variables with the same distribution as $\sumsign_t$.
For a hyperplane $V \subseteq \mathbb{R}^n$, define the {\em discrete codimension} $d(V)$ of $V$ to be the unique integer multiple of $1 / nt$ such that
$$
\left( c_1 t^{-1 / 2} \right)^{-d(V) - 1 / nt} \le \Pr[X \in V] \le \left( c_1 t^{-1 / 2}\right)^{-d(V)}.
$$
\end{definition}
According to the definition, it is clear from Lemma \ref{lem:vector_in_hyperplane} that $1 \le d(V) \le O(n)$.

We first dispose hyperplanes with high discrete codimension using the following lemma, which is a direct corollary of Lemma 1 in \cite{kahn1995probability}.
\begin{lemma}\label{lem:high_codimension}
Suppose $X \in \mathbb{R}^n$ whose entries are i.i.d. random variables with the same distribution as $\sumsign_t$, then
$$
\Pr\left[\bigcup_{V \mid d(V) \ge d_0} X_1, X_2, \ldots, X_n \text{ span } V\right] \le n\cdot  \left( c_1 t^{-1 / 2}\right)^{-d_0}.
$$
\end{lemma}

Let $1 /2 \ge \varepsilon > 0$ be a constant to be determined.
Using Lemma \ref{lem:high_codimension}, we have
$$
\Pr\left[\bigcup_{V \mid d(V) \ge (\varepsilon - o(1))n} X_1, X_2, \ldots, X_n \text{ span } V\right] \le n\cdot  \left( c_1 t^{-1 / 2}\right)^{(\varepsilon - o(1))n} = t^{-\Omega(n)}.
$$
Thus, in the remaining part of the proof we will focus only on the case when $d(V) \le (\varepsilon - o(1))n$ .

We say a hyperplane $V$ to be non-degenerate if its normal vector $n(V)$ satisfies $\|n(V)\|_0 \ge \ceil{\log \log n / \log t}$.
Here we use $\|n(V)\|_0$ to denote the number of non-zero entries in the normal vector $n(V)$.
The following lemma, which is a simple adaption of Lemma 5.3 in \cite{tao2006random}, provides a crude estimation of the number of degenerate hyperplanes.
\begin{lemma}\label{lem:number_degenerate}
The number of degenerate non-trivial hyperplanes is at most $t^{o(n)}$.
\end{lemma}
Combining Lemma \ref{lem:vector_in_hyperplane} and Lemma \ref{lem:number_degenerate}, we then have
$$
\sum_{V \mid \text{$V$ is degenerate}}  \Pr[X_1, X_2, \ldots, X_n \text{ span } V] \le \sum_{V \mid \text{$V$ is degenerate}}  \Pr[X_1, X_2, \ldots, X_n \in V]  \le t^{o(n)} \left(  c_1 t^{- 1 / 2}\right)^n = t^{-\Omega(n)}.
$$
Thus, we can just focus on non-degenerate hyperplanes.

The following theorem, which first appeared in \cite{kahn1995probability} as Theorem 2 (see also Section 7 in \cite{tao2006random}), is based on Fourier-analytic arguments by Hal\'asz \cite{halasz1975distribution, halasz1977estimates}.
\begin{theorem}\label{thm:halasz}
Suppose $V \subseteq \mathbb{R}^n$ is a non-trivial hyperplane.
Let $Y^{(\mu)} \in \mathbb{R}^n$ whose entries are i.i.d. random variables with the same distribution as $\sumsign^{(\mu)}$,
$\lambda < 1$ be a positive number and $k$ be a positive integer such that $4\lambda k^2 < 1$.
We have
$$
\Pr[Y \in V] \le \left(\frac{1}{k(1 - 4 \lambda k^2)} + \frac{1}{1 - 4 \lambda} e^{-(1 - 4 \lambda)\|n(V)\|_0 / (4k^2)}\right)\Pr[Y^{(\lambda e^{-\lambda})} \in V],
$$
where we use $n(V)$ to denote the normal vector of $V$ and $\|n(V)\|_0$ to denote the number of non-zero entries of $n(V)$.
\end{theorem}
\begin{corollary}\label{coro:halasz}
Suppose $W \subseteq \mathbb{R}^n$ is a non-degenerate non-trivial hyperplane.
Let $X^{(\mu)} \in \mathbb{R}^n$ whose entries are i.i.d. random variables with the same distribution as $\sumsign_t^{(\mu)}$.
For sufficiently large $t$, we have
$$
\Pr[X \in W] \le O(t^{- 1/ 4})\Pr[X^{(\lambda e^{-\lambda})} \in W] \le t^{- 1/ 5} \Pr[X^{(\lambda e^{-\lambda})} \in W].
$$
\end{corollary}
\begin{proof}
We note that 
\begin{equation}\label{equ:X_in_W}
\Pr[X^{(\mu)} \in W] = \Pr[\langle X^{(\mu)}, n(W) \rangle = 0] = \Pr\left[\sum_{i=1}^n \sum_{j=1}^t Y^{(\mu)}_{i, j} n_i(W) = 0\right],
\end{equation}
where $Y^{(\mu)}_{i, j}$ are i.i.d. random variables with the same distribution as $\sumsign^{(\mu)}$ and $n_i(W)$ is the $i$-th coordinate of the normal vector $n(W)$.
This enables one to apply Theorem \ref{thm:halasz}.
Notice that when applying Theorem \ref{thm:halasz} we have $\|n(V)\|_0 = t\|n(W)\|_0$, since each non-zero entry of $n(W)$ appears $t$ times in the summation of (\ref{equ:X_in_W}).
Recall that $\lambda = t^{- 1 / 2}$.
We set $k$ to be an integer which is at least $\Omega(t^{1 / 4})$.
Since $V$ is non-degenerate, we have $\|n(V)\|_0 = t\|n(W)\|_0 \ge t \cdot \ceil{\log \log n / \log t}$, which implies
$$
\frac{1}{1 - 4 \lambda} e^{-(1 - 4 \lambda)\|n(V)\|_0 / (4k^2)} = o(1).
$$
The correctness of the corollary thus follows from our choice of $k$.
\end{proof}

For a non-degenerate non-trivial hyperplane $V$ which satisfies $1 \le d(V) \le (\varepsilon - o(1))n$, define $A_V$ to be the event that
$$
X_1^{(\lambda e^{-\lambda})}, X_2^{(\lambda e^{-\lambda})}, \ldots, X_{(1 - \eta)n}^{(\lambda e^{-\lambda})}, X_1',X_2', \ldots, X_{(\eta - \varepsilon')n}' \text{ are linearly independent in $V$},
$$
where $X_i^{(\lambda e^{-\lambda})}$ are independent random vectors in $\mathbb{R}^n$ whose entries are i.i.d. random variables with the same distribution as $\sumsign_{t}^{(\lambda e^{-\lambda})}$ and $X_i'$ are random vectors in $\mathbb{R}^n$ whose entries are i.i.d. random variables with the same distribution as $\sumsign_{t}$.
Here $\eta = 3d(V) / n$ and $\varepsilon' = \min\{\eta, \varepsilon\}$ where $1 / 2 \ge \varepsilon > 0$ is a constant to be determined.

We first prove that
$$
\Pr[A_V] \ge t^{(1 - \eta)n / 5}  \left(c_1 t^{-1 / 2}\right)^{(1 - \varepsilon')n \cdot d(V) + o(n)}.
$$
To prove this, we define $A_V'$ to be the event that
$$
X_1^{(\lambda e^{-\lambda})}, X_2^{(\lambda e^{-\lambda})}, \ldots, X_{(1 - \eta)n}^{(\lambda e^{-\lambda})}, X_1',X_2', \ldots, X_{(\eta - \varepsilon')n}' \in V.
$$

By Corollary \ref{coro:halasz},
$$
\Pr[A_V'] \ge t^{(1 - \eta)n / 5}  \left(c_1 t^{-1 / 2}\right)^{(1 - \varepsilon')n \cdot d(V) + o(n)}.
$$

Now we show that 
$$
\Pr[A_V \mid A_{V}'] = t^{-o(n)}.
$$
According to the definition of discrete codimension $d(V)$,
we have
$$
\Pr[X \in V] \ge (1 - O(1 / n)) \left(c_1 t^{-1 / 2}\right)^{d(V)}.
$$
By Corollary \ref{coro:halasz} we have
$$
\Pr[X^{(\lambda e^{-\lambda})} \in V] \ge t^{1 / 5}(1 - O(1 / n)) \left(c_1t^{-1 / 2}\right)^{d(V)}.
$$
On the other hand, by Lemma \ref{lem:vector_in_hyperplane}, we have
$$
\Pr[X^{(\lambda e^{-\lambda})} \in W] \le \left(c_2 t^{- 1/ 4}\right)^{n - \mathrm{dim}(W)}.
$$
Thus,
$$
\Pr[X^{(\lambda e^{-\lambda})} \in W \mid X^{(\lambda e^{-\lambda})} \in V] \le (t^{1 / 5}(1 - O(1 / n)))^{-1} \left(\sqrt{t} / c_1\right)^{d(V)} \cdot \left(c_2 t^{-1 / 4}\right)^{n - \mathrm{dim}(W)},
$$
which implies
\begin{align*}
&\Pr\left[X_1^{(\lambda e^{-\lambda})}, \ldots, X_{k + 1}^{(\lambda e^{-\lambda})} \text{ are independent} \mid X_1^{(\lambda e^{-\lambda})}, \ldots, X_{k}^{(\lambda e^{-\lambda})} \text{ are independent} \land A_V' \right]\\
 \ge &1 - (t^{1 / 5}(1 - O(1 / n)))^{-1}  \left(\sqrt{t} / c_1\right)^{d(V)} \cdot \left(c_2 t^{-1 / 4}\right)^{n - k}.
\end{align*}
Using the estimation given above, for sufficiently large $t$, we have
$$
\Pr\left[X_1^{(\lambda e^{-\lambda})}, \ldots, X_{(1 - \eta)n}^{(\lambda e^{-\lambda})} \text{ are independent} \mid A_V'\right] \ge t^{-o(n)}
$$
since $(1 - \eta)n = n - 3d(V)$.

Similarly, when $\varepsilon' < \eta$, i.e., $\varepsilon' = \varepsilon$, we have
\begin{align*}
&\Pr \big[X_1^{(\lambda e^{-\lambda})}, \ldots, X_{(1 - \eta)n}^{(\lambda e^{-\lambda})}, X_1', \ldots, X_{k + 1}' \text{ are independent} \\
& \quad \mid X_1^{(\lambda e^{-\lambda})}, \ldots, X_{(1 - \eta)n}^{(\lambda e^{-\lambda})}, X_1', \ldots, X_{k}' \text{ are independent} \land A_V \big]\\
 \ge &1 - (1 - O(1 / n))^{-1} \left(\sqrt{t} / c_1\right)^{d(V)} \cdot \left(c_1 t^{-1 / 2}\right)^{n - k - (1 - \eta)n}\\
 \ge &1 - \frac{1}{100}\left(\sqrt{t} / c_1\right)^{k + (\varepsilon - \eta)n} \tag{$d(V) \le (\varepsilon - o(1))n$}.
\end{align*}
Again we have
$$
\Pr[X_1^{(\lambda e^{-\lambda})}, X_2^{(\lambda e^{-\lambda})}, \ldots, X_{(1 - \eta)n}^{(\lambda e^{-\lambda})}, X_1',X_2', \ldots, X_{(\eta - \varepsilon')n}' \text{ are independent} \mid A_V'] 
=\Pr[A_V \mid A_V'] 
= t^{-o(n)}.
$$
We define $B_V$ to be the event that $X_1, X_2, \ldots, X_n$ span the hyperplane $V$. 
Since $A_V$ and $B_V$ are independent, we have
\begin{equation}\label{equ:B_V}
\Pr[B_V] = \Pr[A_V \land B_V] / \Pr[A_V] \le \Pr[A_V \land B_V] t^{-(1 - \eta)n / 5}  \left(\sqrt{t} / c_1\right)^{(1 - \varepsilon')n \cdot d(V) + o(n)}.
\end{equation}
Consider a set
$$
X_1^{(\lambda e^{-\lambda})}, X_2^{(\lambda e^{-\lambda})}, \ldots, X_{(1 - \eta)n}^{(\lambda e^{-\lambda})}, X_1',X_2', \ldots, X_{(\eta - \varepsilon')n}', X_1, X_2, \ldots, X_n
$$
which satisfies $A_V \land B_V$.
There exist $\varepsilon'n - 1$ vectors
$$
X_{j_1}, X_{j_2}, \ldots, X_{j_{\varepsilon'n - 1}}
$$
such that
$$
X_1^{(\lambda e^{-\lambda})}, X_2^{(\lambda e^{-\lambda})}, \ldots, X_{(1 - \eta)n}^{(\lambda e^{-\lambda})}, X_1',X_2', \ldots, X_{(\eta - \varepsilon')n}', X_{j_1}, X_{j_2}, \ldots, X_{j_{\varepsilon'n - 1}} \text{ span $V$}.
$$
By using a union bound of size $\binom{n}{\varepsilon'n - 1} = 2^{nh(\varepsilon') + o(n)}$, we can just assume $j_i = i$.
Here we use $h(\varepsilon')$ to denote the binary entropy function. 
Thus, 
\begin{align*}
\Pr[A_V \land B_V] &\le 
2^{nh(\varepsilon') + o(n)} \Pr \left[ X_1^{(\lambda e^{-\lambda})}, X_2^{(\lambda e^{-\lambda})}, \ldots, X_{(1 - \eta)n}^{(\lambda e^{-\lambda})}, X_1',X_2', \ldots, X_{(\eta - \varepsilon')n}', X_{1}, X_{2}, \ldots, X_{{\varepsilon'n - 1}} \text{ span $V$}\right]\\
& \cdot \Pr[X_{\varepsilon'n}, X_{\varepsilon'n + 1}, \ldots, X_n \in V].
\end{align*}
Thus, by using (\ref{equ:B_V}) and Lemma \ref{lem:vector_in_hyperplane} we have
\begin{align*}
\Pr[B_V] &\le t^{-(1 - \eta)n / 5}  \left(\sqrt{t} / c_1\right)^{(1 - \varepsilon')n  d(V) + o(n)} \cdot 2^{nh(\varepsilon') + o(n)}  \left(c_1 t^{-1 / 2}\right)^{((1 - \varepsilon')n + 1)d(V)}\\
& \cdot \Pr \left[ X_1^{(\lambda e^{-\lambda})}, X_2^{(\lambda e^{-\lambda})}, \ldots, X_{(1 - \eta)n}^{(\lambda e^{-\lambda})}, X_1',X_2', \ldots, X_{(\eta - \varepsilon')n}', X_{1}, X_{2}, \ldots, X_{{\varepsilon'n - 1}} \text{ span $V$}\right].
\end{align*}
Notice that
$$
\sum_{V} \Pr \left[ X_1^{(\lambda e^{-\lambda})}, X_2^{(\lambda e^{-\lambda})}, \ldots, X_{(1 - \eta)n}^{(\lambda e^{-\lambda})}, X_1',X_2', \ldots, X_{(\eta - \varepsilon')n}', X_{1}, X_{2}, \ldots, X_{{\varepsilon'n - 1}} \text{ span $V$}\right] \le 1.
$$
Thus, for any $1 \le d_0 \le (\varepsilon - o(1))n$ and sufficiently large $t$, we have
\begin{align*}
\sum_{V \mid d(V) = d_0} \Pr[B_V] 
\le &t^{-(1 - \eta)n / 5}  \left(\sqrt{t} / c_1\right)^{(1 - \varepsilon')n  d_0 + o(n)} \cdot 2^{nh(\varepsilon') + o(n)}  \left(c_1 t^{-1 / 2}\right)^{((1 - \varepsilon')n + 1)d_0}\\
\le & t^{-(1 - \eta)n / 5 + nh(\varepsilon') / 5 - d_0 / 5 + o(n)}\\
\le & t^{(-1 + 3d_0 / n  + h(\varepsilon)  - d_0 / n ) n/ 5 + o(n)} \\
\le & t^{(-1 + 2 \varepsilon  + h(\varepsilon) ) n/ 5 + o(n)}\\
\le & t^{-\varepsilon n/ 5 + o(n)}.\\
\end{align*}
Here the second inequality follows since $2^{nh(\varepsilon')} \le t^{nh(\varepsilon') / 5}$ for sufficiently large $t$.
The third inequality is due to the monotonicity of the binary entropy function $h(\cdot)$ on $[0, 1 / 2]$ and the fact that $0 < \varepsilon' \le \varepsilon \le 1 / 2$.
The fourth inequality follows from the fact that $d_0 / n \le \varepsilon$.
The last inequality follows by setting $\varepsilon$ to be the solution of $h(\varepsilon) + 3\varepsilon = 1$. 
A numerical calculation shows that $\varepsilon > 0.177$.
Theorem \ref{thm:singularity_main} thus follows by using a union bound for all possible $d_0$, which has at most $O(n^2t) = t^{o(n)}$ different valid values and setting $C = \varepsilon / 5$.

We remark that the choice of parameters here is mainly for simplicity and not optimized.

\section{Discussion} 
The lens of communication complexity reveals surprising structure about well-known optimization problems. A very interesting open question is to fully resolve the randomized communication complexity of linear programming as a function of $s, d,$ and $L$. Another interesting direction is to design more efficient linear programming algorithms in the RAM model with unit cost operations on words of size $O(\log(nd))$ bits; such algorithms while being inherently useful may also give rise to improved communication protocols. While our regression algorithms illustrated various shortcomings of previous techniques, there are still interesting gaps in our bounds to be resolved. 

\bibliography{main}

\end{document}